\documentclass[a4paper, 12pt]{article}
\usepackage{srcltx,graphicx}
\usepackage{amsmath, amssymb, amsthm}
\usepackage{color}
\usepackage{lscape}
\usepackage{multirow}
\usepackage{psfrag}
\usepackage[hang]{subfigure}
\usepackage{tikz}
\usetikzlibrary{matrix}

\newtheorem{theorem}{Theorem}[section]
\newtheorem{lemma}[theorem]{Lemma}
\newtheorem{definition}[theorem]{Definition}
\newtheorem{property}{Property}
{\theoremstyle{remark} \newtheorem{remark}{Remark}}
\newtheorem{corollary}[theorem]{Corollary}
\newtheorem{example}{Example}
\newtheorem{conjecture}[theorem]{Conjecture}

\newcommand\bbR{\mathbb{R}}
\newcommand\bbN{\mathbb{N}}
\newcommand\bxi{\boldsymbol{\xi}}
\newcommand\bx{\boldsymbol{x}}
\newcommand\bv{\boldsymbol{v}}
\newcommand\bq{\boldsymbol{q}}
\newcommand\bu{\boldsymbol{u}}

\newcommand\br{\boldsymbol{r}}
\newcommand\bR{\boldsymbol{R}}
\newcommand\bA{\boldsymbol{A}}
\newcommand\bB{\boldsymbol{B}}

\newcommand\dd{\,\mathrm{d}}
\newcommand\He{\mathit{He}}
\newcommand\HeT{\mathit{He}^{[\Theta]}}
\newcommand\Het{\mathit{He}^{[\theta]}}

\newcommand\bw{\boldsymbol{w}}
\newcommand\mH{\mathcal{H}}
\newcommand\mHT{\mathcal{H}^{[\Theta]}}
\newcommand\mG{\mathcal{G}}
\newcommand\weight{w^{[\Theta]}}
\newcommand\cS{{\cal{S}}}
\newcommand\mN{{\mathcal N}_{D}}

\newcommand\rC[2]{{\rm{C}}_{{#1},{#2}}}
\newcommand\bX{\boldsymbol{X}}
\newcommand\rank{\mathrm{rank}}
\newcommand\rotation{^*}

\newcommand\thA{\tilde{\hat{A}}}
\newcommand\tbAM{\tilde{\bA}_M}
\newcommand\tbAdotM{\tilde{\bA}'_M}

\newcommand\NRxx{NR$xx$}

\newcommand\pd[2]{\dfrac{\partial #1}{\partial #2}}
\newcommand\od[2]{\dfrac{\dd #1}{\dd #2}}
\newcommand\odd[2]{\dfrac{\mathrm{D} #1}{\mathrm{D} #2}}

\numberwithin{equation}{section}

\graphicspath{{images/}}

\title{Globally Hyperbolic Moment System by Generalized Hermite Expansion}

\author{Yuwei Fan\thanks{School of Mathematical Sciences, Peking
    University, Beijing, China, email: {\tt ywfan@pku.edu.cn}.},~~ Ruo
  Li\thanks{CAPT, LMAM \& School of Mathematical Sciences, Peking
    University, Beijing, China, email: {\tt rli@math.pku.edu.cn}.}}

\begin{document}
\maketitle
\begin{abstract}
  In a recent paper \cite{Grad13toR13}, it was revealed that a
  modified 13-moment system taking intrinsic heat fluxes as variables,
  instead of the heat fluxes along the coordinate vectors which is
  adopted in the classical Grad 13-moment system, attains some
  additional advantages than the classical Grad 13-moment system,
  particularly including that the equilibrium is turned to be the
  interior point of its hyperbolicity region. The modified 13-moment
  system was actually derived from the generalized Hermite expansion
  of the distribution function, where the anisotropy of Hermite
  expansion is specified by the full temperature tensor. We extend the
  method therein in this paper to high order of generalized Hermite
  expansion to derive arbitrary order moment systems, and proposed a
  globally hyperbolic regularization to achieve locally well-posedness
  similar to the method in \cite{Fan_new}. Furthermore, the structure
  of the eigen-system of the coefficient matrix and all characteristic
  waves are fully clarified. The obtained systems provide a systematic
  class of hydrodynamic models as the refined version of Euler
  equations, which is gradually approaching the Boltzmann equation
  with increasing order of the expansion.

  \vspace*{4mm}
\end{abstract}

\section{Introduction}
In 1949 \cite{Grad}, Grad proposed the moment expansion method for the
Boltzmann equation to derive the macroscopic hydrodynamic systems, as
the refined models beyond the Euler equations and the
Navier-Stokes-Fourier (NSF) equations. Among the models derived
therein, Grad's 13-moment system is one of the most well-known models.
This system was derived by expanding the distribution function into
isotropic Hermite series \cite{Grad1949note}. Soon after the model
proposed, it was found that this model is problematic in a number of
aspects, one fatal point of which was that Grad's 13-moment system is
not globally hyperbolic. Actually, the hyperbolicity can
only be obtained near the equilibrium \cite{Muller} even for 1D
flows. The loss of hyperbolicity directly breaks the local
well-posedness of the system, and thus the capability of this model is
strictly limited. Historically, Grad's moment system has been included
in the textbooks for decades while there are very seldom reports on
its success, in spite of the elegant mathematical formation of the
system. Aiming on improved well-posedness of Grad's moment system,
different efforts has been made both for the 13 moment system and
higher order moment system. The approaches may be divided into two
folds, including proposing certain dissipation terms derived from the
collision term and considering different closure to extent the
hyperbolicity region. We refer the regularized Burnett equations
\cite{Jin}, regularized 13-moment equations \cite{Struchtrup2003,
  Struchtrup2005}, the Pearson-13-moment equations
\cite{Torrilhon2010}, et. al. These methods may alleviate the problem
of hyperbolicity to some extent \cite{Slemrod, TorrilhonRGD,
  Torrilhon2010}. 

In a recent study \cite{Grad13toR13}, the authors pointed out that the
thermodynamic equilibrium is always on the boundary of the
hyperbolicity region of Grad's 13-moment system. More precisely, it
was proved therein that if an arbitrary small perturbation is applied
to the phase density from the equilibrium, the hyperbolicity may break
down. This reveals that there does not exist a neighbourhood of the
equilibrium such that all the states in this neighbourhood lead to the
hyperbolicity of Grad's 13-moment system. Without the hyperbolicity in
a neighbourhood of the equilibrium, the well-posedness of the Grad's
13-moment system is not guaranteed even the phase density is extremely
close to the equilibrium. This severe drawback may be the possible
reason why there are hardly any positive evidences for the Grad's
13-moment system in the last decades. Noticing that the anisotropy
plays an essential role in breaking down the hyperbolicity, it was
then proposed in \cite{Grad13toR13} a new modified 13-moment model
such that the equilibrium state lies in the interior of the
hyperbolicity region, even without any hyperbolicity regularization
techniques used such as in \cite{Fan_new}. This modified system is
derived by a generalized Hermite expansion instead of the isotropic
Hermite expansion in Grad's method, where the anisotropy is specified
by the full temperature tensor. It was found that once the generalized
Hermite expansion is adopted, the equilibrium is turned into an
interior point of the hyperbolicity region of the full 3D system with
13 moments. It is indicated that the generalized Hermite expansion may
be an essential point in further development of high order moment
method.

As a macroscopic hydrodynamic model derived from the Boltzmann
equation, a necessary requirement is that the model derived has to be
invariant under Galilean transformation. To achieve this point, Grad
in \cite{Grad} adopted the Ikenberry type polynomials as the weight
functions to retrieve the macroscopic quantities from the distribution
function. Actually, the Grad's 13-moment system is the first model
obtained beyond classical hydrodynamic system following this
way. Since all the components of the temperature tensor are included
in the variables of the macroscopic model, it looks inappropriate
insisting to expand the distribution function using the isotropic
Hermite polynomials. To study higher order moment method than the
13-moment system proposed in \cite{Grad13toR13}, we are motivated to
adopt expansion of the distribution function using generalized Hermite
polynomials. Particularly, we will propose a regularization to the
derived systems following the method in \cite{Fan_new} to achieve the
globally hyperbolicity, thus the local well-posedness of the
regularized system may be attained.

The rest part of this paper is arranged as follows. In section 2, we
first derived the moment system based on generalized Hermite expansion
to arbitrary order for any dimensional cases. In section 3, the
coefficient matrix of the obtained moment system is studied in detail,
and we point out that the obtained moment system is lack of global
hyperbolicity. In section 4, we propose a globally hyperbolic
regularization for the moment system obtained. The eigenvalues and the
eigenvectors of the regularized system are explicitly calculated, and
the global hyperbolicity of the regularized system is rigorously
proved. In section 5, the Riemann problem is investigated. All
characteristic waves are either genuinely nonlinear or linearly
degenerate, and some properties of the rarefaction waves, the contact
waves and the shock waves are investigated. In the appendix, we
present the properties and formulas of generalized Hermite polynomials
and the proof of some technical results in detail.

Before we start the main text, a conjecture on the distribution of the
zeros of Hermite polynomials is presented as below at first: we
conjecture that there are no same non-zero zeros of $\He_n(x)$ and
$\He_m(x)$ for all $m,n\in\bbN$ and $m\neq n$. Precisely,
\begin{conjecture}\label{con:He}
  for any $m,n\in\bbN$ and $m\neq n$, there are no common non-zero
  zeros of $\He_n(x)$ and $\He_m(x)$, -i.e.  $\nexists x\in\bbR
  \backslash \{ 0 \}$, such that $\He_n(x)=\He_m(x)=0$.
\end{conjecture}
For the detailed description of this conjecture, please see the
appendix of this paper.


\section{Derivation of Moment System}
\label{sec:expansion}
We consider the Boltzmann equation for kinetic theory of gases as
\begin{equation}\label{eq:Boltzmann}
    \pd{f}{t}+\sum_{d=1}^D\xi_d\pd{f}{x_d}=Q(f,f),
\end{equation}
where $f(t,\bx,\bxi)$ is the distribution function and
$(t,\bx,\bxi)\in\bbR^+\times\bbR^D\times\bbR^D$, $\bx = (x_1, \cdots,
x_D)$, $\bxi = (\xi_1, \cdots, \xi_D)$, $D$ is the dimension and
$Q(f,f)$ is the collision term. The macroscopic density $\rho$, mean flow
velocity $\bu = (u_1, \cdots, u_D)$, pressure tensor $p_{ij}$ and heat
flux $\bq = (q_1, \cdots, q_D)$ of the gas along the axis of
coordinates are related with the distribution function by, for
$i,j=1,\cdots,D$,
\begin{align*}
\rho &= \int_{\bbR^D} f \dd\bxi,    &
\rho u_i &= \int_{\bbR^D} \xi_i f \dd\bxi, \\
p_{ij} &= \int_{\bbR^D} (\xi_i - u_i)(\xi_j - u_j) f \dd\bxi, &
q_i &= \frac{1}{2} \int_{\bbR^D} |\bxi - \bu|^2 (\xi_i - u_i) f
\dd\bxi.
\end{align*}
The original collision term of the Boltzmann equation is quite complex,
and in the present work we only consider the BGK-type collision term
as
\begin{equation}
    Q(f,f)=\nu(\mG-f),
\end{equation}
where $\nu$ is the collision frequency, and $\mG$ is a certain
distribution function depending on the collision model under
consideration. For BGK collision model \cite{BGK},
\begin{equation}
    \mG=f_M=\frac{\rho}{(2\pi\theta)^{D/2}}\exp\left(
    -\frac{|\bxi-\bu|^2}{2\theta} \right),
\end{equation}
where $\theta=\displaystyle\sum_{d=1}^D\dfrac{p_{dd}}{D\rho}$ is the macroscopic
temperature, and for ES-BGK collision model \cite{Holway},
\begin{equation}
    \mG=\frac{\rho}{\sqrt{\det{(2\pi\Lambda)}}}\exp\Big(
    -\frac{1}{2}(\bxi-\bu)^T\Lambda^{-1}(\bxi-\bu) \Big),
\end{equation}
where $\Lambda_{ij}=bp_{ij}/\rho+(1-b)\theta\delta_{ij}$,
$b=1-\dfrac{1}{\Pr}\in[-1/2,1]$, and $\Pr$ is the Prandtl number which
is approximately equal to $2/3$ for a monatomic gas. In particular, if
$\Pr=1$, the ES-BGK model degrades into the BGK model.

In 1949, Grad \cite{Grad} made an Hermite expansion for distribution
function $f$ and obtained the well-known Grad 20 and Grad 13 moment
equations. Cai and Li \cite{NRxx} extended it to more general case and
obtained a class moment equations of arbitrary order, which is called
{\NRxx} method. Below we inherit the basic approach of \NRxx, adopt a
class of generalized Hermite polynomials, and make a generalized
Hermite expansion to derive a class of anisotropic moment system.

Consider the weight function 
\begin{equation}\label{eq:weight}
  \weight(\bv)=\frac{1}{\sqrt{\det{(2\pi\Theta)}}} \exp \left(-\frac{1}{2} 
    \bv^T\Theta^{-1}\bv\right),
\end{equation}
where $\bv = (v_1, \cdots, v_D) \in\bbR^D$ and
$\Theta=(\theta_{ij})\in\bbR^{D\times D}$ is a symmetrical positive
definite matrix. The generalized Hermite polynomials are defined as
\begin{equation}
    \He^{[\Theta]}_{\alpha}(\bv)=
    \frac{(-1)^{|\alpha|}}{\weight(\bv)}\pd{^{\alpha}}{\bv^{\alpha}}\weight(\bv),
    \quad \alpha\in\bbN^D,
\end{equation}
where $\alpha=(\alpha_1,\cdots,\alpha_D)$ is a $D$-dimensional
multi-index,
$\pd{^{\alpha}}{\bv^{\alpha}}=\dfrac{\partial^{|\alpha|}}{\partial
  v_1^{\alpha_1}\cdots \partial v_D^{\alpha_D}}$, and $|\alpha| =
\alpha_1 + \cdots + \alpha_D$. The difference between the generalized
Hermite polynomials and the Hermite polynomials is the anisotropy in
the weight function, where the matrix $\Theta$ in the generalized
Hermite polynomials is a scalar in the Hermite polynomials. Then we
define the generalized Hermite functions as
\begin{equation}\label{eq:basis}
    \mHT_{\alpha}(\bv)=\weight(\bv)\He^{[\Theta]}_{\alpha}(\bv),
    \quad \alpha\in\bbN^D.
\end{equation}
Clearly, the weight function $\weight(\bv)$ is
normalized that
\[ \int_{\bbR^D}\weight(\bv)\dd\bv=1. \] If any component of $\alpha$
is negative, we take $(\cdot)_{\alpha}$ to be zero for convenience. If
$D=1$ and $\Theta=1$, $\mHT_{\alpha}(\bv)$ degenerate into the Hermite
polynomials with Gaussian distribution as weight function (see
\cite{Abramowitz} for details), thus the definitions are consistent
to 1D case. The generalized Hermite function is studied in details in
the Appendix \ref{sec:HermitePolynomial}, and below we summarize some
useful properties of $\mHT_{\alpha}(\bv)$:
\begin{enumerate}
    \item
        Recursion relation: 
        \[ v_d\mHT_{\alpha}(\bv)=\sum_{j=1}^D\theta_{jd}\mHT_{\alpha+e_j}(\bv)+\alpha_d\mHT_{\alpha-e_d}; \]
    \item
        Quasi-orthogonality relations:
        \[ \displaystyle \int_{\bbR^D}\mHT_{\alpha}(\bv)\mHT_{\beta}(\bv)\dfrac{1}{\weight}\dd
        \bv=C_{\alpha, \beta}\delta_{|\alpha|,|\beta|}, \] where
        $C_{\alpha,\beta}$ is constant dependent on $\alpha,\beta$,
        and $\Theta$.
    \item
        Differential relation:
        \begin{equation}
            \od{\mH_{\alpha}^{[\Theta(t)]}(\bv(t))}{t}
            =-\sum_{i=1}^D\mH_{\alpha+e_i}^{[\Theta(t)]}(\bv(t))\od{v_i(t)}{t}
            +\frac{1}{2}\sum_{i,j=1}^D\mH_{\alpha+e_i+e_j}^{[\Theta(t)]}(\bv(t))\od{\theta_{ij}(t)}{t},
        \end{equation}
        where $e_i$, $i=1,\cdots,D$ is the $D$-dimensional unit
        multi-index with its $i$-th entry equal to 1.
\end{enumerate}

We expand the distribution function $f(t,\bx,\bxi)$ into the series
of $\mHT_{\alpha}(\bv)$ as
\begin{equation}\label{eq:expansion}
    f(t,\bx,\bxi)=\sum_{\alpha\in\bbN^D}f_{\alpha}(t,\bx)\mHT_{\alpha}(\bxi-\bu).
\end{equation}
Substituting the expansion \eqref{eq:expansion} into the Boltzmann
equation \eqref{eq:Boltzmann}, and comparing the coefficient of
$\mHT_{\alpha}(\bxi-\bu)$, we obtain
\begin{equation}\label{eq:origalsystem}
    \begin{split}
        \odd{f_{\alpha}}{t} ~+ & \sum_{d, k = 1}^D \left( \theta_{d k} \pd{f_{\alpha-e_k}}{x_d}
        + \left( \alpha_k + 1 \right) \delta_{k d}
        \pd{f_{\alpha + e_k}}{x_d} \right) + \\ 
        \sum_{i = 1}^D f_{\alpha - e_i}
        \odd{u_i}{t} ~+ & \sum_{i, d, k = 1}^D \left( \theta_{d k} f_{\alpha - e_i - e_k} +
        \left( \alpha_k + 1 \right) \delta_{k d} f_{\alpha - e_i + e_k} \right)
        \pd{u_i}{x_d} + \\
        \sum_{i, j = 1}^D \frac{
        f_{\alpha - e_i - e_j}}{2} \odd{\theta_{i j}}{t} ~+ & \sum_{i, j, d, k = 1}^D \frac{1}{2} 
        \left(\theta_{kd} f_{\alpha - e_i - e_j - e_k} + 
        \left( \alpha_k + 1 \right)\delta_{k d}
        f_{\alpha - e_i - e_j + e_k} \right) \pd{\theta_{i j}}{x_d} \\
        =& ~\nu(\mG_{\alpha}-f_{\alpha}),
    \end{split}
\end{equation}
where $\odd{~}{t}$ is the material derivation standing for
\[
    \odd{~}{t}=\pd{~}{t}+\sum_{d=1}^Du_d\pd{~}{x_d},
\]
and $\mG$ is expanded as
\[
    \mG=\sum_{\alpha\in\bbN^D}\mG_{\alpha}\mHT_{\alpha}(\bxi-\bu).
\]

In particular, if we let $\Theta=\theta\boldsymbol{I}$, then the
moment system \eqref{eq:origalsystem} is exactly the same as the
system derived in \cite{NRxx_new} without hyperbolic
regularization. In this paper hereafter, we let
$\theta_{ij}=p_{ij}/\rho$. The expansion \eqref{eq:expansion} together
with the quasi-orthogonality relation of $\mHT_{\alpha}(\bv)$ yields
\begin{equation}\label{eq:freeparameter}
    f_0=\rho,\quad f_{e_i}=0,\quad f_{e_i+e_j}=0,\quad 
    q_i=2f_{3e_i}+\sum_{d=1}^Df_{e_i+2e_d},\quad i,j=1,\cdots,D.
\end{equation}
Direct calculations give us
\begin{equation}\label{eq:collision_expansion}
    \mG_{0}=\rho,
    \quad
    \mG_{e_i+e_j}=\frac{1-b}{1+\delta_{ij}}(p\delta_{ij}-p_{ij}),
    \quad i,j=1,\cdots,D,
    \quad
    \mG_{\alpha}=0, \quad\text{if }|\alpha|\text{ is odd}.
\end{equation}

In particular, in case of $\alpha = {\mathbf{0}}$ and noticing
$f_{e_i} = 0$ for $i = 1, \cdots, D$, we deduce the continuity
equation from \eqref{eq:origalsystem} as
\begin{equation}\label{eq:massconservation}
    \odd{\rho}{t}+\rho\sum_{d=1}^D\pd{u_d}{x_d}=0.
\end{equation}
By setting $\alpha=e_i$ with $i=1,\cdots,D$ in
\eqref{eq:origalsystem}, using \eqref{eq:freeparameter}, we obtain the
equation of momentum conservation as
\begin{equation}\label{eq:momentumconservation}
    \rho\odd{u_i}{t}+\sum_{d=1}^D\left(
    \theta_{id}\pd{\rho}{x_d}+\rho\pd{\theta_{id}}{x_d} \right)=0.
\end{equation}
By setting $\alpha=e_i+e_j$ with $i,j=1,\cdots,D$ and $i\geq j$ in
\eqref{eq:origalsystem}, using \eqref{eq:freeparameter}, we have the
conservation laws of pressure tensor as
\begin{equation}\label{eq:pressuretensor}
    \frac{2-\delta_{ij}}{2}\rho\odd{\theta_{ij}}{t}
    +\sum_{d=1}^D\left[(1+\delta_{id}+\delta_{jd})\pd{f_{e_i+e_j+e_d}}{x_d}
    +\rho\theta_{id}\pd{u_j}{x_d}+\rho\theta_{jd}\pd{u_i}{x_d}(1-\delta_{ij})
    \right]=\nu\mG_{e_i+e_j}.
\end{equation}
For the case $D=3$, if we let $f_{\alpha}=0$, $|\alpha|=0$, then
\eqref{eq:massconservation}, \eqref{eq:momentumconservation} and
\eqref{eq:pressuretensor} are the well-known 10-moment system
\cite{Holway10m}.

Substituting \eqref{eq:momentumconservation} and
\eqref{eq:pressuretensor} into \eqref{eq:origalsystem} to eliminate
the material derivation of $u_i$ and $\theta_{ij}$, $i,j=1,\dots,D$,
we get the governing equation of $f_{\alpha}$ as
\begin{equation}\label{eq:origalsystem_not}
  \begin{split}
  \odd{f_{\alpha}}{t} + & \sum_{d,k = 1}^D \theta_{dk}
  \pd{f_{\alpha - e_k}}{x_d} + \sum_{d = 1}^D (\alpha_d+1)\pd{f_{\alpha+e_d}}{x_d}\\
  - & \sum_{i,d = 1}^D f_{\alpha - e_i} \left( \frac{\theta_{id}}{\rho}
  \pd{\rho}{x_d} + \pd{\theta_{i d}}{x_d} \right) 
  +  \sum_{i,d = 1}^D (\alpha_d+1)f_{\alpha-e_i+e_d}\pd{u_i}{x_d} \\
  + & \frac{1}{2} \sum_{i, j, d, k = 1}^D ( \theta_{k d} f_{\alpha - e_i -
  e_j - e_k} + \delta_{k d} ( \alpha_k + 1) f_{\alpha - e_i - e_j + e_k})
  \pd{\theta_{i j}}{x_d} \\
   - & \sum_{i,j,d=1}^D
   \frac{1+\delta_{i d}+\delta_{j d}}{2 - \delta_{i j}}
   \frac{f_{\alpha-e_i-e_j}}{\rho} \pd{f_{e_i+e_j+e_d}}{x_d} 
   =
   \nu\left(\mG_{\alpha}-f_{\alpha}+\sum_{i,j=1}^D\frac{\mG_{e_i+e_j}}{\rho}f_{\alpha-e_i-e_j}\right). \\
  \end{split}
\end{equation}
Then \eqref{eq:massconservation}, \eqref{eq:momentumconservation},
\eqref{eq:pressuretensor} and \eqref{eq:origalsystem_not} constitute a
moment system with infinite equations, which is a quasi-linear system.

To attain a system with finite number of equations, a truncation has
to be applied. Due to the quasi-orthogonality of the basis functions
$\mHT_{\alpha}(\bv)$, we let $M\in\bbN$, $M\geq2$, and adopt the
finite set of closure coefficients $\{f_{\alpha}\}_{|\alpha| \leq M}$,
and discard all the equations with $\mathrm{D}f_{\alpha}/\mathrm{D}t$
with $|\alpha|>M$. Then we get a moment system with finite
equations. However, the system obtained is not closed yet since in the
equations with $\mathrm{D} f_{\alpha}/\mathrm{D} t$, $|\alpha|=M$, the
terms of $f_{\alpha+e_d}$, $d=1,\dots,D$ are involved. The simplest
way to close the system is to inherit Grad's idea \cite{Grad} to let
$f_{\alpha}=0$ with $|\alpha|=M+1$ in the moment system. Here we first
use Grad's way to close the system, which results in a generalized
Grad-type moment system.

We remark here that for $D=1$, the moment system above is the same as
the \NRxx~ method in 1D case \cite{NRxx}, thus the \NRxx~ method,
which is the isotropic Grad-type moment system, may be regarded as a
special case of the system obtained here.


\section{Analysis of Moment System}
\label{sec:hyperbolic}
As has been pointed out in \cite{Muller}, the moment system in
\cite{Grad} is not globally hyperbolic. In \cite{Fan}, the authors
showed that the \NRxx~ is also not globally hyperbolic even with
$D=1$. In this section, we will point out that the generalized
Grad-type moment system shares the same problem. For this purpose, we
focus on the properties the coefficient matrix of the generalized
Grad-type moment system, and prove that the moment system is not
globally hyperbolic.

At first, we reformulate the generalized Grad-type moment system
obtained in the previous section as below. Equations
\eqref{eq:momentumconservation} and \eqref{eq:pressuretensor} are as
\begin{align}
    \odd{u_i}{t}&+\sum_{d=1}^D
    \frac{1}{\rho}\pd{p_{id}}{x_d}=0,
    \label{eq:momentumconservation2}\\
    \begin{split}
    \odd{p_{ij}}{t}&+\sum_{d=1}^D\left(
    p_{ij}\pd{u_d}{x_d}
    +p_{id}\pd{u_j}{x_d}+p_{jd}\pd{u_i}{x_d}
    +(e_i+e_j+e_d)!\cdot\pd{f_{e_i+e_j+e_d}}{x_d}
    \right)\\
    &\qquad\qquad\qquad \qquad\qquad\qquad \qquad\qquad 
    =(1+\delta_{ij})\nu\mG_{e_i+e_j},
    \label{eq:pressuretensor2}
    \end{split}
\end{align}
where $\alpha!$ is defined as $\alpha!=\prod_{d=1}^D\alpha_d!$ and
$i,j=1,\cdots,D$.
Since $p=\frac{1}{D}\sum_{i=1}^Dp_{ii}$, we have
\begin{equation}
    \odd{p}{t}+\sum_{i,d=1}^D\frac{2}{D}p_{id}\pd{u_i}{x_d}
    +\sum_{d=1}^D\left( p\pd{u_d}{x_d}+\frac{2}{D}\pd{q_d}{x_d} \right)
    =0,
\end{equation}
where $\sum_{i=1}^D\mG_{2e_i}=0$ is used. Since
\[
    \pd{\theta_{ij}}{x_d}=\frac{1}{\rho}\pd{p_{ij}}{x_d}-\frac{\theta_{ij}}{\rho}\pd{\rho}{x_d}
\]
holds for any $i,j,d=1,\cdots,D$, the moment equations
\eqref{eq:origalsystem_not} is reformulated as
\begin{equation}\label{eq:momentsystem}
  \begin{split}
  \odd{f_{\alpha}}{t} + & \sum_{d,k = 1}^D \theta_{dk}
  \pd{f_{\alpha - e_k}}{x_d} + \sum_{d = 1}^D (\alpha_d+1)\pd{f_{\alpha+e_d}}{x_d}\\
  +&\sum_{i,j,d=1}^D\frac{C_{ijd}(\alpha)}{2\rho}
  \left(\pd{p_{ij}}{x_d}- \theta_{ij}\pd{\rho}{x_d}\right)
  +  \sum_{i,d = 1}^D (\alpha_d+1)f_{\alpha-e_i+e_d}\pd{u_i}{x_d}\\ 
  - & \sum_{i,d = 1}^D \frac{f_{\alpha-e_i}}{\rho} \pd{p_{id}}{x_d} 
   -  \sum_{i,j,d=1}^D \frac{(e_i+e_j+e_d)!}{2}
   \frac{f_{\alpha-e_i-e_j}}{\rho} \pd{f_{e_i+e_j+e_d}}{x_d}\\
   = &
   \nu\left(\mG_{\alpha}-f_{\alpha}+\sum_{i,j=1}^D\frac{\mG_{e_i+e_j}}{\rho}f_{\alpha-e_i-e_j}\right), 
  \end{split}
\end{equation}
where $C_{ijd}(\alpha)$ is
\begin{equation}\label{eq:def_Cijd}
    C_{ijd}(\alpha)= \sum_{k=1}^D\theta_{k d} f_{\alpha - e_i -
  e_j - e_k} + ( \alpha_d + 1) f_{\alpha-e_i-e_j+e_d}.
\end{equation}

For later usage, some conventional notations are introduced as follows.
\newcommand\seq[2]{#1\!:\!#2}
\begin{align*}
\text{For a vector }~ \pmb{a}=(a_1,\cdots,a_n)\in\bbR^{n},& \text{we denote}
    ~~\pmb{a}(\seq{i}{j})=(a_i,\cdots,a_j); \\
\text{For a matrix }~\bA=(a_{ij})_{n\times n}\in\bbR^{n\times n},& \text{we denote} \\
    ~~\bA(i,\,\seq{j}{k})=(a_{i,j},\cdots,a_{i,k}), &
    ~~\bA(i,\,:)=\bA(i,1:n), \\
    ~~\bA(\seq{i}{l},\,\seq{j}{k})=&\begin{pmatrix}
            a_{i,j}& a_{i,j+1} & \cdots& a_{i,k}\\
            a_{i+1,j}& a_{i+1,j+1} & \cdots& a_{i+1,k}\\
            \vdots &\vdots&\ddots&\vdots\\
            a_{l,j}& a_{l,j+1} & \cdots& a_{l,k}
            \end{pmatrix},
\end{align*}
Let
\[
 \cS_{D,M} = \{\alpha\in\bbN^D\mid |\alpha| \le M\},
\]
we permute the elements of $\cS_{D,M}$ by lexicographic order. Then
for any $\alpha\in \cS_{D,M}$,
\begin{equation}\label{eq:def_mN}
\mN(\alpha) = \sum_{i=1}^D\binom{\sum_{k=D-i+1}^D\alpha_k + i-1}{i}+1
\end{equation}
holds, where $\mN(\alpha)$ is the ordinal number of $\alpha$ in
$\cS_{D,M}$. Noticing $Me_D$ is the last element of $\cS_{D,M}$, the
cardinal number of set $\cS_{D,M}$ is
\[
    N = \mN(Me_D) = \binom{M+D}{D},
\]
which is total number of variables in the truncated moment system if a
truncation with $|\alpha|\le M$ is considered.

With the notations above, we collect the variables in the truncated
moment system to form a vector $\bw\in\bbR^N$ as
\begin{align*}
    w_1 & =\rho,    &   w_{i+1}&=u_i,\\
    w_{\mN(e_i+e_j)}&=p_{ij}/(1+\delta_{ij}),
    &   w_{\mN(\alpha)}&=f_{\alpha}, \mbox{~else~} \alpha.
\end{align*}
where $i,j=1,\cdots,D$, and $|\alpha|\leq M$. Fig. \ref{fig:order_w}
shows the permutation of entries $\bw$ as the variables of the
truncation moment system. Collecting together
\eqref{eq:massconservation}, \eqref{eq:momentumconservation2},
\eqref{eq:pressuretensor2} and \eqref{eq:momentsystem}, we arrive the
following quasi-linear system
\begin{equation}\label{eq:system}
    \odd{\bw}{t}+\sum_{d=1}^D\bA_M^{(d)}\pd{\bw}{x_d}=\nu\boldsymbol{Q}\bw,
\end{equation}
where the entries of $\bA_M^{(d)}$ with $d=1,\dots,D$ and
$\boldsymbol{Q}$ are given in \eqref{eq:massconservation},
\eqref{eq:momentumconservation2}, \eqref{eq:pressuretensor2} and
\eqref{eq:momentsystem}. The matrices $\bA_M^{(d)}$ have quite regular
structure, though complex. Next we will devote to study in detail these
coefficient matrix $\bA_M^{(d)}$.

\subsection{Properties of the coefficient
matrix}\label{sec:propertyOfA}
Without loss of generality, we only investigate $\bA_M^{(1)}$. For
simplicity, we momentarily strip away the supscripts and use $\bA_M$
to replace $\bA_M^{(1)}$ without ambiguity.

\subsubsection{Case $D=1$}
\label{sec:caseD1}
This case has been thoroughly studied in \cite{Fan}. Let us recall the
results therein below for comparison. In this case, the coefficient
matrix is precisely as { \renewcommand{\arraystretch}{1.5}
\begin{equation}\label{eq:1Dmatrix}
    \footnotesize
    \bB_M=
    \begin{pmatrix}
        U&\rho&0&\hdotsfor{6}&0\\
        0&U&2\,{\rho}^{-1}&0&\hdotsfor{5}&0\\
        0&3p_{11}/2&U&3&0&\hdotsfor{4}&0\\
        -1/2\theta_{11}^2&4\,f_{{3}}&\theta_{11}&U&4&0&\hdotsfor{3}&0\\
        -{\frac {5\theta_{11}f_{{3}}}{2{\rho}}}&5\,f_{{4}}&
        3{\frac {f_{3}}{\rho}}&{\theta_{11}}&U&5&0&\hdotsfor{2}&0\\
        -3{\frac {\theta_{11}f_{{4}}}{{\rho}}}&6f_{5}&4{\frac {f_{{4}}}{\rho}}
        &{\frac{-3f_{3}}{\rho}}&\theta_{11}&U&6&0&\cdots&0\\
        \hdotsfor{10}\\
        -\frac{M\theta_{11}f_{M-2}+\theta_{11}^2f_{M-4}}{2\rho} 
        &{\scriptstyle{M}}f_{M-1} & \frac{(M-2)f_{M-2}+\theta_{11}f_{M-4}}{\rho} &
        {\frac{-3f_{M-3}}{\rho}}&0&\cdots&0&\theta_{11}&U&M\\
        -\frac{(M+1)\theta_{11}f_{M-1}+\theta_{11}^2f_{M-3}}{2\rho} 
        &{\scriptstyle{(M+1)}}f_{M} & \frac{(M-1)f_{M-1}+\theta_{11}f_{M-3}}{\rho} &
        {\frac{-3f_{M-2}}{\rho}}&0&\hdotsfor{2}&0&\theta_{11}&U\\
    \end{pmatrix},
\end{equation}
}
where $U = 0$. It is clear that this matrix
\begin{itemize}
\item is independent of $u$ and the diagonal entries are all zeros;
\item is a lower Hessenberg matrix.
\end{itemize}
The characteristic polynomial of this matrix is
\begin{equation}\label{eq:eigenpolynomial1D}
    \theta_{11}^{M+1}\He^{[\theta_{11}]}_{M+1}(\lambda) -(M+1)!\left( \lambda f_M
    +\frac{\lambda^2-\theta_{11}}{2}f_{M-1}\right).
\end{equation}
If $f_M$ and $f_{M-1}$ are taken certain values, the characteristic
polynomial may not have $M+1$ real roots, thus the matrix can not be
diagonalizable with real eigenvalues.
The eigenvector of this matrix for the eigenvalue $\lambda$,
satisfying \eqref{eq:eigenpolynomial1D} is
\begin{equation}\label{eq:eigenvector1D}
    r_1=\rho, ~r_2=\lambda, ~r_3=\frac{\rho\lambda^2}{2},
    ~r_k=\frac{\rho\He_{k-1}^{[\theta_{11}]}(\lambda)}{(k-1)!}-f_{k-2}\lambda 
    -f_{k-3}\frac{\He_2^{[\theta_{11}]}(\lambda)}{2},
    k=4,\cdots,M+1.
\end{equation}

\subsubsection{Case $D \geq 2$}
We are interested in the case $D \geq 2$. Let us investigate some
examples at first for a full clarification of the structure of the
coefficient matrix.
\begin{example}\label{exam:A3}
  If $D=2$, the ordinal number of $\alpha$ in $\cS_{D,M}$ is
  $\mN(\alpha) = \dfrac{(\alpha_1 + \alpha_2 + 1)(\alpha_1 +
    \alpha_2)}{2} + \alpha_2 + 1$. The permutation of entries of $\bw$
  is showed in Fig. \ref{fig:order_w}. For the simple case, the
  matrix $\bA_3$ is
  \begin{equation*}
      \bA_3=
      \begin{tikzpicture}[baseline=-\the\dimexpr\fontdimen22\textfont2\relax]
          \matrix[matrix of math nodes,left delimiter = (,right delimiter = ),row sep=0pt,column sep = 0pt] (m) {
          0&\rho&0&0&0&0&0&0&0&0\\
          0&0&0&2\,{\rho}^{-1}&0&0&0&0&0&0 \\
          0&0&0&0&{\rho}^{-1}&0&0&0&0&0\\
          0&3p_{11}/2&0&0&0&0&3&0&0&0\\
          0&2\,p_{12}&p_{11}&0&0&0&0&2&0&0\\
          0&p_{22}/2&p_{12}&0&0&0&0&0&1&0\\
          -\theta_{11}^2/2&4\,f_{{30}}&0&{\theta_{11}}&0&0&0&0&0&0\\
          -3\theta_{11}\theta_{12}/2&3\,f_{{21}}&3\,f_{{30}}&\theta_{12}&\theta_{11}&0&0&0&0&0\\
          -\theta_{11}\theta_{22}/2-\theta_{12}^2&2\,f_{{12}}&2\,f_{{21}}&0
          &\theta_{12}&\theta_{11}&0&0&0&0\\
          -\theta_{22}\theta_{12}/2&f_{{03}}&f_{{12}}&0
          &0&\theta_{12}&0&0&0&0\\
      };
      \draw (m-1-1.north west) rectangle (m-6-6.south east);
      \end{tikzpicture},
  \end{equation*}
  where $f_{ij}=f_{ie_1+je_2}$, and the upper-left part in the box is
  denoted by $\bA_2$. If $M>3$, for any $\alpha\in\bbN^2$, and
  $3<|\alpha|\leq M$, we have
  \begin{subequations}\label{eq:AMD2}
\begin{align}
&\bA_M(\seq{1}{10}, \, \seq{1}{10}) = \bA_3,\\
&\bA_M(\mN(\alpha), \mN(\alpha)) =  0,\\
&\bA_M(\mN(\alpha), \mN(\alpha-e_k)) =  \theta_{1k}, 
    \quad\text{if }\alpha_k>0, \label{eq:AMD2:3} \\
&\bA_M(\mN(\alpha), \mN(\alpha+e_1)) =  \alpha_1+1, 
    \quad\text{if }|\alpha|< M,\\ 
\begin{split}
&\bA_M(\mN(\alpha),\,\seq{1}{9})
=(-\sum_{i,j=1}^D\frac{\theta_{ij}C_{ij1}}{2\rho}, ~~
(\alpha_1+1)f_{\alpha},~~(\alpha_1+1)f_{\alpha+e_1-e_2},~~\\ 
   & \qquad\qquad\qquad\qquad
   \frac{C_{111}(\alpha)}{\rho}-\frac{2f_{\alpha-e_1}}{\rho},~~
   \frac{C_{121}(\alpha)}{\rho}-\frac{f_{\alpha-e_2}}{\rho},~~
   \frac{C_{221}(\alpha)}{\rho},\\
    &\qquad\qquad\qquad\qquad
    -\frac{3f_{\alpha-2e_1}}{\rho},~ ~
    -\frac{2f_{\alpha-e_1-e_2}}{\rho},~ ~
    -\frac{f_{\alpha-2e_2}}{\rho}), 
\end{split} \label{eq:AMD2:5}
\end{align}
\end{subequations}
where $C_{ijd}(\alpha)$ are given in \eqref{eq:def_Cijd}. We remark
that 
\begin{itemize}
    \item any entry of $\bA_M(i,j)$, if not specified above, is taken as
        zero; 
    \item for $|\alpha|=4$, some entries $\bA_M(i,j)$ may be doublely
        defined in \eqref{eq:AMD2:3} and \eqref{eq:AMD2:5}, the
        value of which is the sum of the both expression;
    \item if any entries of $\alpha$ is negative, $(\cdot)_{\alpha}$
        is taken as zero.
\end{itemize}
\end{example}

\begin{figure}[h]
\centering
  \includegraphics[width=0.75\textwidth]{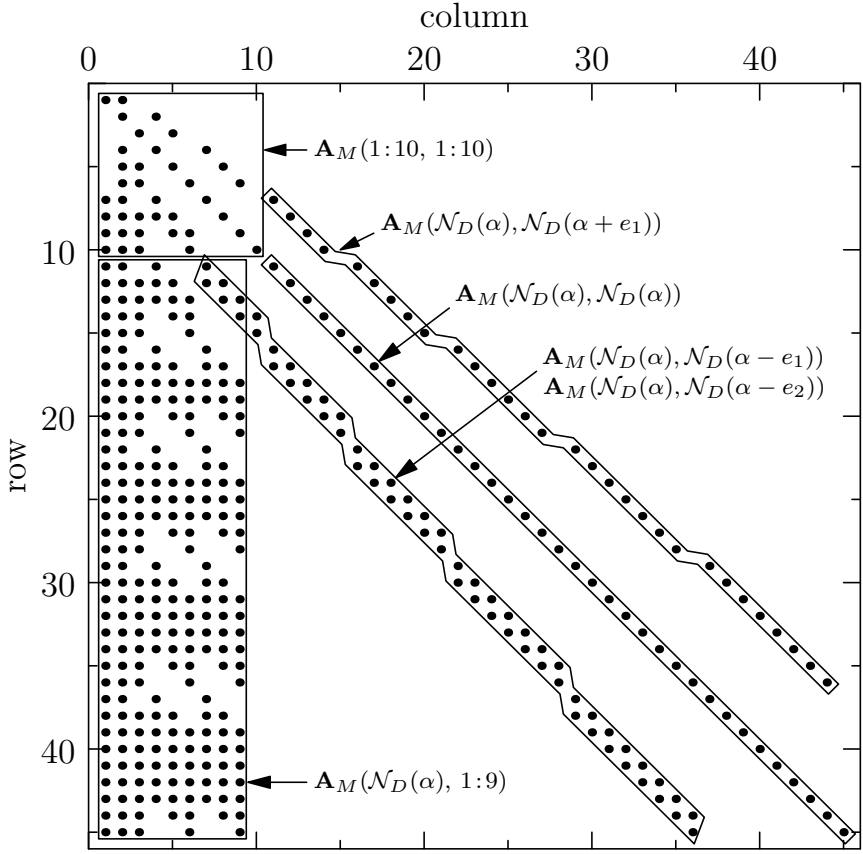}
  \caption{\label{fig:sp_ori}The sparsity pattern of $\bA_M$ with
  $M=8, D=2$. Its nonzero entries are given in \eqref{eq:AMD2}.}
\end{figure}

Clearly the matrix $\bA_M$ is independent of $\bu$, and the diagonal
entries vanish. Actually, in \eqref{eq:massconservation},
\eqref{eq:momentumconservation2}, \eqref{eq:pressuretensor2} and
\eqref{eq:momentsystem}, the coefficients of terms with derivative to
$x_d$, $d=1,\cdots,D$, are independent of $\bu$; and in the equation
containing $\odd{w_i}{t}$, $i=1,\cdots,N$, the coefficients of
$\pd{w_i}{x_d}$, $d=1,\cdots,D$ are zero. Hence, we have that
\begin{property}\label{pro:diagonal}
    The coefficient matrix $\bA$ is independent of $\bu$, thus
    \[
        \pd{\bA}{\bu}=0,
    \]
    and the diagonal entries of $\bA$ are all zeros.
\end{property}
By this property, the moment system is invariant under a Galilean
translation.

In example \ref{exam:A3}, the coefficient matrix $\bA_M$ for $D=2$ is
explicitly given, which makes one able to study the sparsity pattern
of $\bA_M$. Fig. \ref{fig:sp_ori} gives the sparsity pattern of
$\bA_M$ with $D=2$ and $M=8$. It is clear that there are at most one
nonzero entry in $\bA_M(i, i+1:N)$, $i=1,\dots,N$. Actually, in the
equation containing $\odd{f_{\alpha}}{t}$ in \eqref{eq:momentsystem},
the only nonzero entry is $\bA_M(\mN(\alpha), \mN(\alpha+e_1))$ in
$\bA_M(\mN(\alpha), \mN(\alpha)+1:N)$. Thus, we have the following
property.
\begin{property}\label{pro:hessenberg}
    For each $\alpha\in\bbN^D$, $|\alpha|\leq M$, let $i=\mN(\alpha)$,
    then there are no more than one entry of $\bA_M(i,i+1:N)$ to be
    nonzero. In particular, for $D=1$, $\bA_M$ is a lower Hessenberg
    matrix.
\end{property}

Property \ref{pro:hessenberg} provides us the approach to calculate
the eigenvalues and eigenvectors of $\bA_M$, the same as operating on
a lower Hessenberg matrix. Furthermore, its lower triangular part is
quite sparse. Let us try to illustrate its sparsity pattern below.
\begin{example}\label{exam:reducible}
  Let $D=2$. Considering only the coefficient matrix $\bA_M =
  \bA_M^{(1)}$, we assume that $\pd{\cdot}{x_2}=0$
  here. \eqref{eq:massconservation} shows the $\odd{\rho}{t}$ is
  dependent on $\pd{u_1}{x_1}$, and we denote the dependence by
  \[
  \rho\rightarrow u_1.
  \]
  Then dependency relationship of entries in $\bw$ by the equations
  \eqref{eq:massconservation}, \eqref{eq:momentumconservation2},
  \eqref{eq:pressuretensor2} and \eqref{eq:momentsystem} is
  demonstrated by the graph in Fig. \ref{fig:dependency}.
  \begin{figure}[h]
    \centering
    \includegraphics[width=0.70\textwidth]{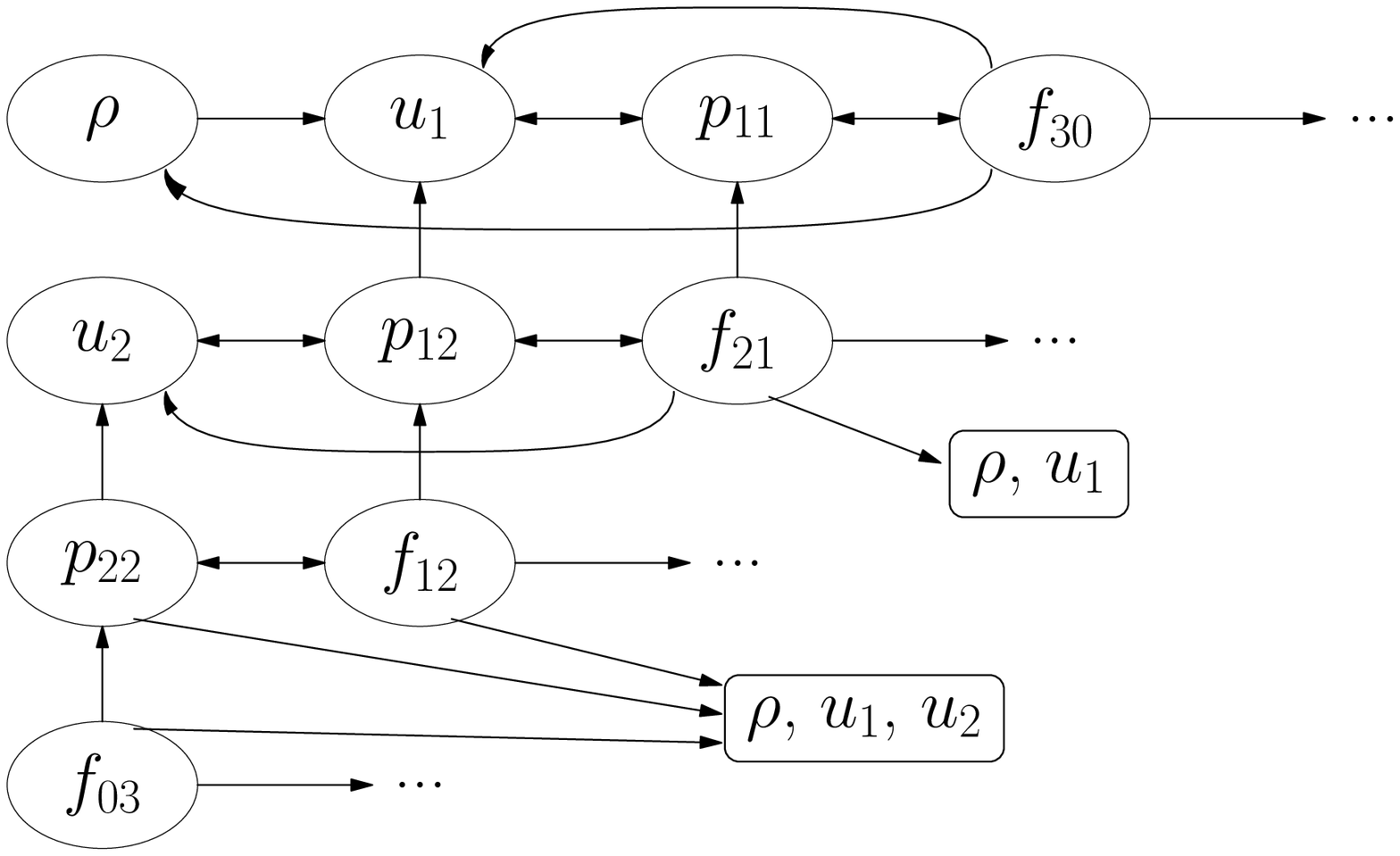}
    \caption{\label{fig:dependency}The dependency relationship of
      $\bw$ with $D=2$ and $\pd{\cdot}{x_2}=0$.  }
  \end{figure}

  It is interesting that in Fig. \ref{fig:dependency} 
  there exists a path from every node to every other node in the same
  row along the direction of the arrow (e.g. there is a path between
  any two entries of $\rho,u_1,p_{11},f_{30},\cdots, f_{Me_1}$), while
  there is no path from one node to any other node in the next row
  (e.g. there is no path from $\rho$ to $u_2$).
  This indicates that the matrix $\bA_M$ is reducible (see Page.
  288-289 of \cite{Demmel} for details). Thus by Fig.
  \ref{fig:dependency}, if we rearrange $\bw$ by the lexicographic
  order of $(\alpha_2,\alpha_1)$, i.e.
  \begin{align*}
      \bw'=(\underbrace{\rho,u_1,p_{11}/2,f_{3e_1},\cdots,f_{Me_1}}
      _{\text{\rm first row}},
      \underbrace{u_2,p_{12},\cdots,f_{(M-1)e_1+e_2}}
      _{\text{\rm second row}},\,
      \cdots,
      \underbrace{f_{Me_2}}_{\text{\rm last row}})^T,
  \end{align*}
  the coefficient matrix $\bA_M$ can be collected into a block lower
  triangular matrix. Fig. \ref{fig:order_w2} shows the permutation of
  $\bw'$ with $D=2$ and $M=8$.

  \begin{figure}[h]
    \centering
    \subfigure[The permutation of $\bw$]{
      \label{fig:order_w}
      \includegraphics[width=0.45\textwidth]{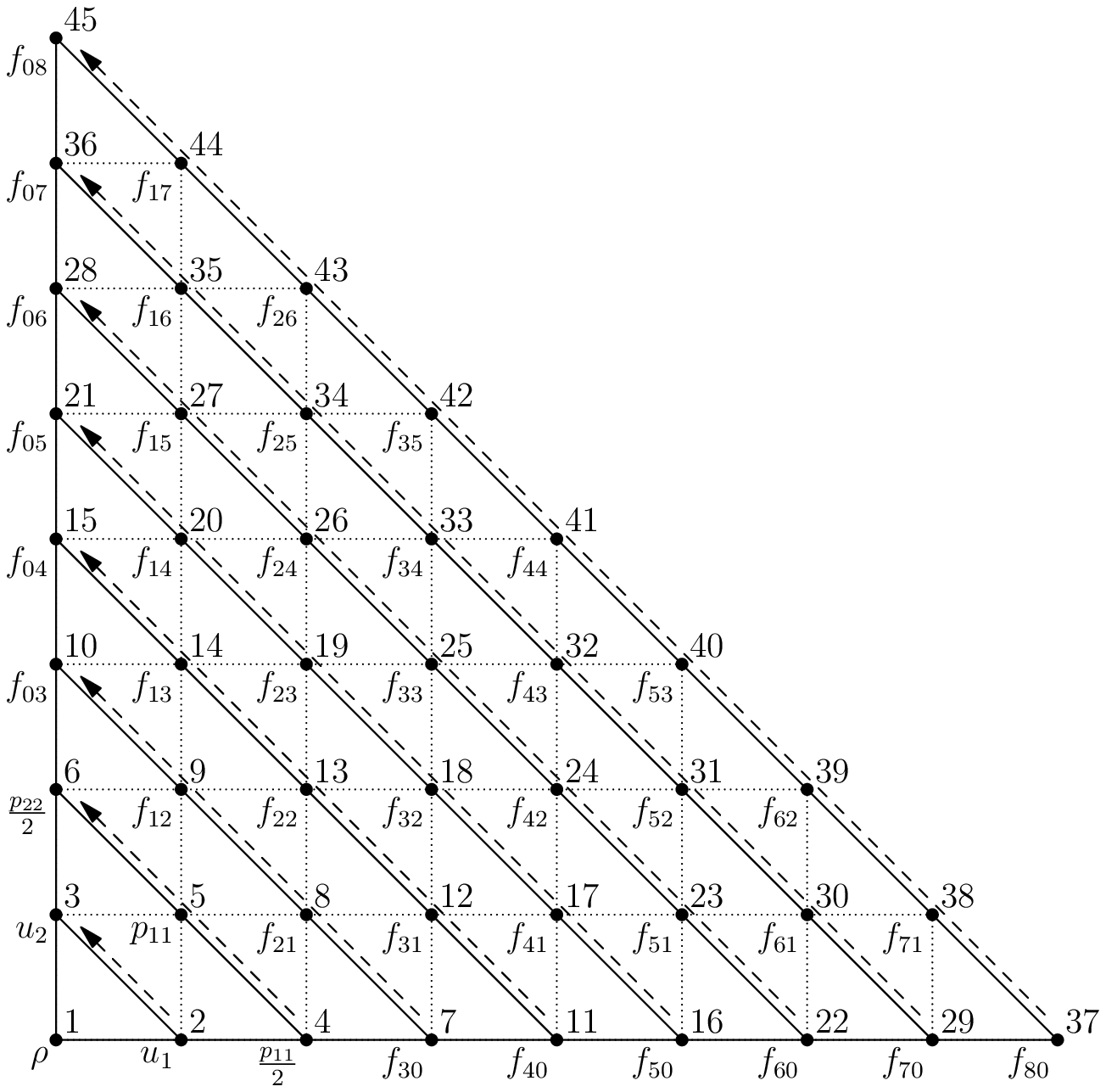}
    }
    \subfigure[A permutation of $\bw'$ defined in example
    \ref{exam:reducible}]{
      \label{fig:order_w2}
      \includegraphics[width=0.45\textwidth]{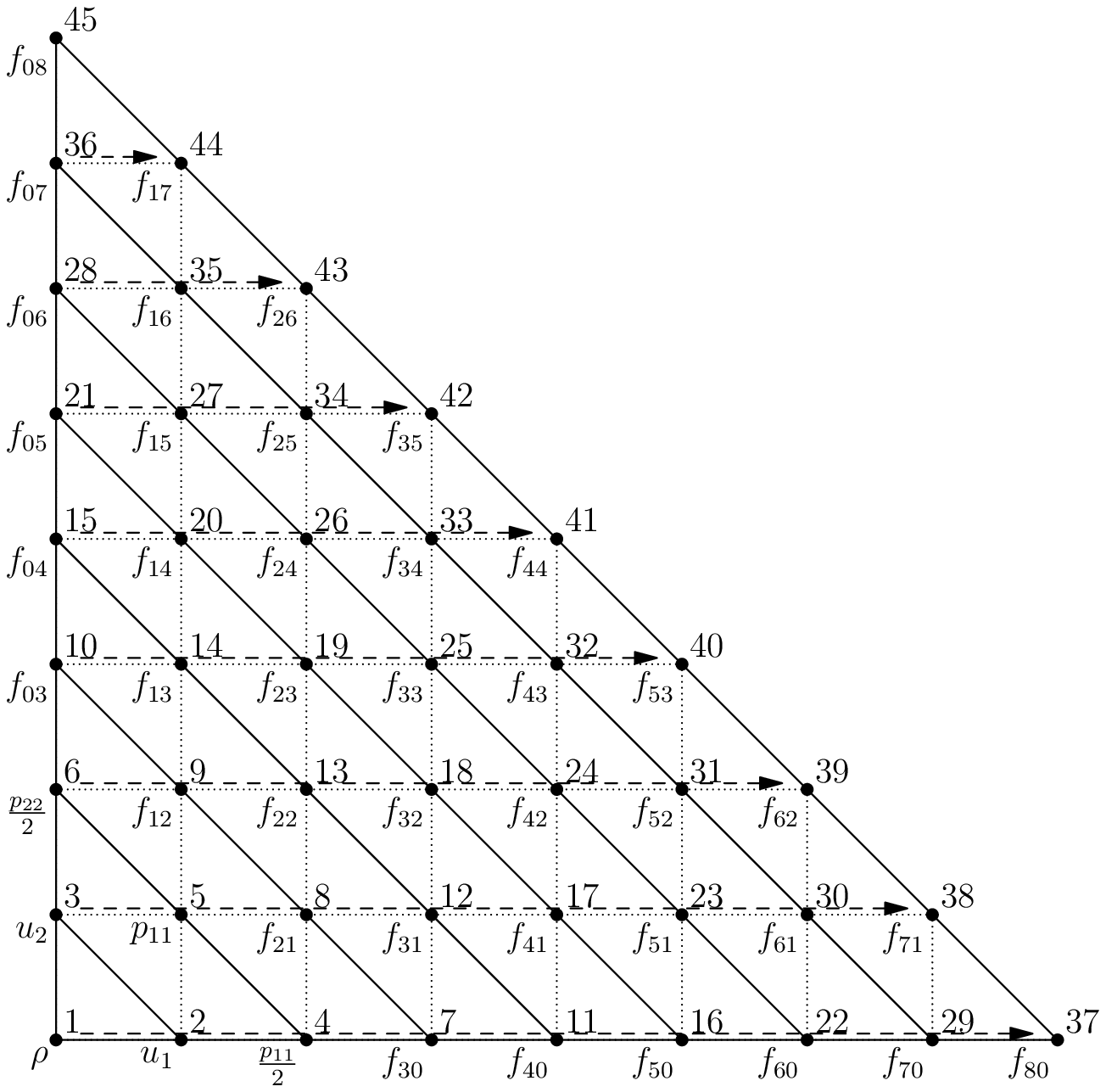}
    }
    \caption{The permutation of the coefficients while $D=2$,
      $M=8$. Each node stands for one coefficient. The marks in the
      lower-left of the node shows the expression of the coefficient,
      while the number in the upper-right represents the ordinal
      number in $\bw$ or $\bw'$. The dashed arrows depict the path of
      the corresponding permutation. The left one is the permutation
      of $\bw$, and the right one is a permutation of $\bw'$ defined
      in example \ref{exam:reducible}.}
  \end{figure}
  \begin{figure}[h!]
    \centering
    \includegraphics[width=0.65\textwidth]{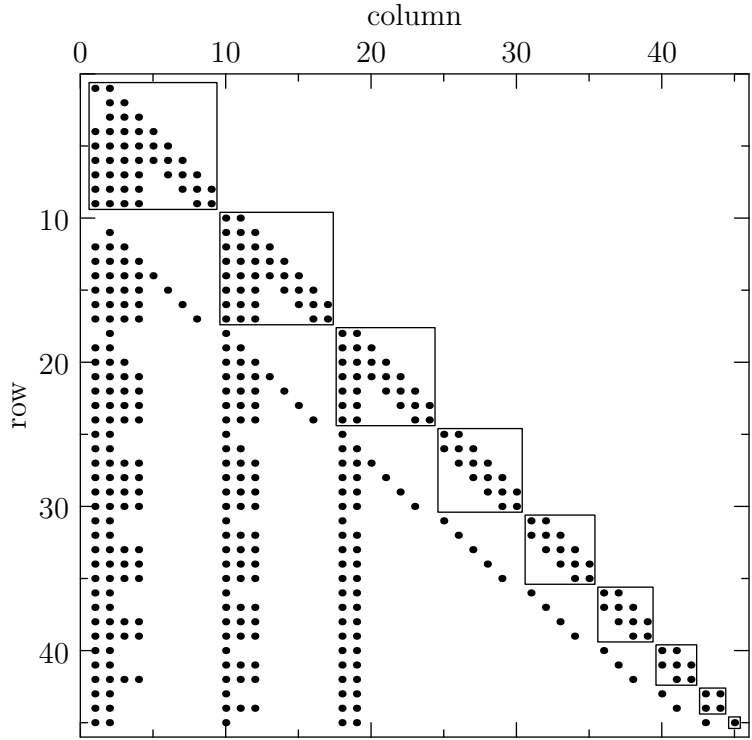}
    \caption{\label{fig:sp_pmt}The sparsity pattern of ${\bf A}_M'$
      with $M=8$, $D=2$.  ${\bf A}_M'$ is reducible and a block lower
      triangular matrix. Each diagonal block is a lower Hessenberg
      matrix.}
  \end{figure}

  The permutation above shows that there exists a permutation matrix
  $\boldsymbol{P}$ such that $\bw'=\boldsymbol{P}\bw$. Let
  $\bA_M'=\boldsymbol{P}\bA_M\boldsymbol{P}^{-1}$, then
  \[
  \odd{\bw'}{t}+\bA_M'\pd{\bw'}{x_1}
  =\nu\boldsymbol{P}\boldsymbol{Q}\boldsymbol{P}^{-1}\bw'
  \]
  holds. Fig. \ref{fig:sp_pmt} gives the sparsity pattern of $\bA_M'$
  with $M=8$. By Fig. \ref{fig:sp_pmt}, it is clear that $\bA_M'$ is
  reducible, and furthermore it is a block lower triangular matrix.
  Precisely, $\bA_M'$ can be written as
  \[
  \bA_M'=\begin{bmatrix}
    \hat{A}_{0}   &   &   &   \\
    *   &   \hat{A}_{1} &   &   \\
    *   &   *   &   \hat{A}_{2}   &   \\
    \hdotsfor{4}\\ 
    *   &   *   &   *   &   \hat{A}_M
  \end{bmatrix},
  \]
  where $\hat{A}_{i}\in\bbR^{(M+1-i)\times (M+1-i)}$, $i=0,\cdots,M$,
  is a lower Hessenberg matrix.

  Let us turn to study the properties of $\hat{A}_{i}$, $i=0,\cdots,M$.
  $\hat{A}_{0}$, $\hat{A}_1$ and $\hat{A}_{2}$ are defined in
  \eqref{eq:hat0}, \eqref{eq:hat1} and \eqref{eq:hat2}, respectively.
  \begin{landscape}
    \renewcommand{\arraystretch}{1.5}
    \begin{equation}\label{eq:hat0}
      \hat{A}_{0}=\begin{pmatrix}
        U&\rho&0&\hdotsfor{6}&0\\
        0&U&2\,{\rho}^{-1}&0&\hdotsfor{5}&0\\
        0&3p_{11}/2&U&3&0&\hdotsfor{4}&0\\
        -1/2\theta_{11}^2&4\,f_{3e_1}&\theta_{11}&U &4&0&\hdotsfor{3}&0\\
        -{\frac {5\theta_{11}f_{3e_1}}{2{\rho}}}&5\,f_{4e_1}&
        3{\frac {f_{3e_1}}{\rho}}&{\theta_{11}}&U&5&0&\hdotsfor{2}&0\\
        -3{\frac
          {\theta_{11}f_{{4e_1}}}{{\rho}}}&6f_{5e_1}&4{\frac
          {f_{{4e_1}}}{\rho}}
        &{\frac{-3f_{3e_1}}{\rho}}&\theta_{11}&U&6&0&\cdots&0\\
        \hdotsfor{10}\\
        -\frac{M\theta_{11}f_{(M-2)e_1}+\theta_{11}^2f_{(M-4)e_1}}{2\rho} 
        &{\scriptstyle{M}}f_{(M-1)e_1} &
        \frac{(M-2)f_{(M-2)e_1}+\theta_{11}f_{(M-4)e_1}}{\rho} &
        {\frac{-3f_{(M-3)e_1}}{\rho}}&0&\cdots&0&\theta_{11}&U&M\\
        -\frac{(M+1)\theta_{11}f_{(M-1)e_1}+\theta_{11}^2f_{(M-3)e_1}}{2\rho} 
        &{\scriptstyle{(M+1)}}f_{Me_1} &
        \frac{(M-1)f_{(M-1)e_1}+\theta_{11}f_{(M-3)e_1}}{\rho} &
        {\frac{-3f_{(M-2)e_1}}{\rho}}&0&\hdotsfor{2}&0&\theta_{11}&U\\
      \end{pmatrix},
    \end{equation}
    \begin{equation}\label{eq:hat1}
      \hat{A}_1=
      \begin{pmatrix}
        U&\rho^{-1}&0&\hdotsfor{5}&0\\
        p_{11}&U&2&0&\hdotsfor{4}&0\\
        3f_{3e_1}&\theta_{11}&U&3&0&\hdotsfor{3}&0\\
        4f_{4e_1}&\frac{3f_{3e_1}}{\rho}&\theta_{11}&U &4&0&\hdotsfor{2}&0\\
        5f_{5e_1}&\frac{4f_{4e_1}}{\rho}&-\frac{2f_{3e_1}}{\rho}&\theta_{11}&U &5&0&\hdotsfor{1}&0\\
        \hdotsfor{9}\\
        (M-1)f_{(M-1)e_1}&\frac{(M-2)f_{(M-2)e_1}+\theta_{11}f_{(M-4)e_1}}{\rho}
        &-\frac{2f_{(M-3)e_1}}{\rho}&0&\cdots&0&\theta_{11}&U
        &M-1\\
        Mf_{Me_1}&\frac{(M-1)f_{(M-1)e_1}+\theta_{11}f_{(M-3)e_1}}{\rho}
        &-\frac{2f_{(M-2)e_1}}{\rho}&0&\hdotsfor{2}&0&\theta_{11}&U
      \end{pmatrix},
    \end{equation}
  \end{landscape}
  \begin{equation}\label{eq:hat2}
    \hat{A}_{2}=
    \begin{pmatrix}
      U&1&0&\hdotsfor{4}&0\\
      \theta_{11}&U&2&0&\hdotsfor{3}&0\\
      \frac{3f_{3e_1}}{\rho}&\theta_{11}&U&3&0&\hdotsfor{2}&0\\
      \frac{4f_{4e_1}}{\rho}&-\frac{f_{3e_1}}{\rho}&\theta_{11}&U &4&0&\hdotsfor{1}&0\\
      \hdotsfor{8}\\
      \frac{(M-2)f_{(M-2)e_1}+\theta_{11}f_{(M-4)e_1}}{\rho}
      &-\frac{f_{(M-3)e_1}}{\rho}&0&\hdotsfor{1}&0&\theta_{11}&U
      &M-2\\
      \frac{(M-1)f_{(M-1)e_1}+\theta_{11}f_{(M-3)e_1}}{\rho}
      &-\frac{f_{(M-2)e_1}}{\rho}&0&\hdotsfor{2}&0&\theta_{11}&U \\
    \end{pmatrix},
  \end{equation}
  and for $\hat{A}_i$, $i=3,\cdots,M$, they all have exactly the same form as
  \begin{equation}
    \hat{A}_i=
    \begin{pmatrix}
      U   &   1   &   0   &   \hdotsfor{3}    &   0\\
      \theta_{11} &   U   &   2   &   0   &   \hdotsfor{2}    & 0\\
      0   &\theta_{11} &   U   &   3   &   0&\cdots&0\\
      \hdotsfor{7}\\
      0&\hdotsfor{2}&0&\theta_{11}&U&M-i\\
      0&\hdotsfor{3}&0&\theta_{11}&U\\
    \end{pmatrix},
    \quad i=3,\cdots,M,
  \end{equation}
  where $U=0$. Consider the matrix $\bB_M$ in \eqref{eq:1Dmatrix}, and
  denote $\bB_M(\rho, \theta_{11}, f_3, \dots, f_M) = \bB_M$, then
  \[
  \hat{A}_0=\bB_M(\rho, \theta_{11}, f_{3e_1},\dots,f_{Me_1}).
  \]
  Hence, $\hat{A}_{0}$ has exactly the same structure as $\bB_M$.
\end{example}

\begin{example}
    The properties of $\hat{A}_0$ is listed in Section \ref{sec:caseD1}.
    Now let us study the properties of $\hat{A}_1$ and $\hat{A}_2$.

    It is clear that the coefficient matrix $\hat{A}_1$ and
    $\hat{A}_2$
    \begin{itemize}
        \item
            are independent of $u$ and the diagonal entries are all
            zeros;
        \item
            are lower Hessenberg matrices.
    \end{itemize}
    Then we study the characteristic polynomials of $\hat{A}_1$ and
    $\hat{A}_2$. Let $\br\neq0$ be an eigenvector of $\hat{A}_1$
    corresponding to the eigenvalue $\lambda$,
    e.g. $\hat{A}_1\br=\lambda\br$. Since $\hat{A}_1$ is a lower
    Hessenberg matrix, we assert $r_1\neq0$.  Assume $r_1=1$, then
    $\hat{A}_1(1,:)\br=\lambda r_1$ gives $r_2=\rho\lambda$. Using the
    same skill on $\hat{A}_1(k,:)\br=\lambda r_k$, $k=2,\cdots,M-1$,
    we can obtain
    \[
        r_{k+1}=\rho\frac{\Het_{k}(\lambda)}{k!}-f_{ke_1}-\lambda
        f_{(k-1)e_1},\quad k=2,\cdots,M-1.
    \]
    $\hat{A}_1(M,:)\br=\lambda r_M$ can be written as 
    \[
        \rho\frac{\Het_{M}(\lambda)}{M!}-(-1)^{M}\left(  
        f_{Me_1}-\lambda f_{(M-1)e_1}
        \right)=0.
    \]
    Hence, the above equation has to be satisfied, if $\lambda$ is an
    eigenvalue of $\hat{A}_1$, which indicates the characteristic
    polynomial of $\hat{A}_1$ is
    \begin{equation}\label{eq:characteristicpolynomial_hatA1}
        \Het_{M}(\lambda)-(-1)^{M}M!\left(  
        f_{Me_1}-\lambda f_{(M-1)e_1}
        \right)/\rho.
    \end{equation}
    Similarly, the characteristic polynomial of $\hat{A}_2$ is
    \begin{equation}\label{eq:characteristicpolynomial_hatA2}
        \Het_{M-1}(\lambda)+(-1)^{M}(M-1)!
        f_{(M-1)e_1}/\rho.
    \end{equation}
\end{example}

For $\alpha\in\bbN^D$, let
\begin{equation}
    \hat{\alpha}=(\alpha_2,\dots,\alpha_D),
\end{equation}
and we denote $\hat{e}_1=0$, $\hat{e}_2=(1,0,\cdots,0)\in\bbR^{D-1}$,
$\cdots$, $\hat{e}_D=(0,\cdots,0,1)\in\bbR^{D-1}$.

For any $2\leq M\in\bbN$, and any $D\in\bbN^+ $, let
$\boldsymbol{P}\in\bbR^{N\times N}$ be the permutation matrix that
$\bw'=\boldsymbol{P}\bw$ in the lexicographic order of
$(\alpha_2,\cdots,\alpha_D,\alpha_1)$, and
$\bA_M'=\boldsymbol{P}\bA_M\boldsymbol{P}$. Then we have the following
results.
\begin{property}\label{pro:reducible}
    For $D\geq2$, the matrix $\bA_M'$ is a block lower triangular
    matrix, and the diagonal blocks
    \[
        (\bA_M')_{ii}=\hat{A}_{\hat{\alpha}}
    \]
    are irreducible, where
    $i=\mathcal{N}_{D-1}(\hat{\alpha})$. Precisely,
    $\hat{A}_{\hat{\alpha}}=\hat{A}_{|\hat{\alpha}|}$,
    $|\hat{\alpha}|=\alpha_2+\cdots+\alpha_D$, where $\hat{A}_{i}$,
    $i=0,\cdots,M$, is defined in Example \ref{exam:reducible}.
\end{property}

\subsection{Lack of global hyperbolicity}
\label{sec:lackofhyperbolicity}
We are ready to show the major result in this section, that the
moment system obtained is not globally hyperbolic for any
$D\in\bbN^+ $, $M\geq3$.
\begin{theorem}\label{thm:nonhyperbolic}
    The moment system obtained in Section \ref{sec:expansion} is not
    globally hyperbolic for any $D\in\bbN^+ $ and $M\geq3$.
\end{theorem}
\begin{proof}
    To prove the theorem, we need only to prove $\bA_M^{(1)}$ is not
    always diagonalizable with real eigenvalues.  
    
    Since $\boldsymbol{P}$ is a permutation matrix, it is enough to
    examine $\bA_M'$. Property \ref{pro:reducible} shows $\bA_M'$ is a
    block lower triangular matrix, thus if $(\bA_M')_{11}$ is not
    diagonalizable with real eigenvalues, $\bA_M'$ is also not.  Since
    $(\bA_M')_{11}=\bB_M(\rho, \theta_{11}, f_{3e_1},\dots,f_{Me_1})$
    and \eqref{eq:eigenpolynomial1D} indicates that if $f_{Me_1}$ and
    $f_{(M-1)e_1}$ take certain values, $(\bA_M')_{11}$ has complex
    eigenvalues. This proves the theorem.
\end{proof}

For the case $M=2$, if $D=1$, then the moment system obtained is
exactly the Euler equations, which is hyperbolic. If $D=3$, then the
moment system is the well-known 10-moment system, which has been
studied in, e.g. \cite{Brown10m35m, Suzuki10m,Holway10m}.

\section{Globally Hyperbolic Regularization}
In this section, we propose a regularization to the moment system to
obtain a globally hyperbolic moment system, following the idea in
\cite{Fan_new}.

\subsection{In one-dimensional spatial space}\label{sec:regularization1D}
For any $3\leq M\in\bbN$, the generalized Grad-type moment system
obtained in Section. \ref{sec:expansion} gives \textit{accurate}
evolution equations for all the variables except for those
$f_{\alpha}$ with $|\alpha|=M$, since $f_{\alpha+e_d}$, $d=1,\cdots,D$
appear in the equations of them, and are taken to be zero in Grad's
closure. The regularization methods given in such as \cite{Levermore,
  Struchtrup2003, Torrilhon2010} were trying to propose a modified form
for $f_{\alpha+e_d}$, $|\alpha|=M$. Actually, noticing that the terms
$f_{\alpha+e_d}$, $|\alpha|=M$, appear only in the evolving equation
of $f_{\alpha}$, $|\alpha|=M$ in the form of its derivatives, a
reasonable regularization should only modify the evolving equations of
$f_{\alpha}$, $|\alpha|=M$ by proposing a suitable form of the
derivatives $\partial f_{\alpha+e_d}/\partial x_d$, $|\alpha|=M$,
$d=1,\cdots,D$. Property \ref{pro:reducible} show us that the
coefficient matrix $\bA'_M$ has the form
\begin{equation} \label{eq:matrix_form}
    \bA'_M=\begin{pmatrix}
        \hat{A}_0   &   0\\
        *   &   *
    \end{pmatrix},
\end{equation}
since the variables $u_1$, $\theta_{11}$ and $f_{ke_i}$, $k = 0,
\cdots, M$ are independent of the other variables. It is natural to
require the regularization to preserve such structure. The
regularization we are proposing below can fulfil all these
constraints, and at the same time achieves the global hyperbolicity.
For convenience, we call
\begin{definition}\label{def:admissible}
  A regularization for the generalized Grad-type moment system is
  {\rm{admissible}}, if
    \begin{enumerate}
        \item
            it only modifies the governing equations of $f_{\alpha}$,
            $|\alpha|=M$;
        \item
            it keeps the regularized coefficient matrix have the form
            as \eqref{eq:matrix_form}. 
    \end{enumerate}
\end{definition}

The proof of Theorem \ref{thm:nonhyperbolic} shows that
$(\bA_M')_{ii}$, $i=1,\cdots,\hat{N}$ is diagonalizable with real
eigenvalues is a necessary condition for that $\bA_M$ is
diagonalizable with real eigenvalues. In this subsection, we first
study the regularization of $(\bA_M')_{ii}$, $i=1,\cdots,\hat{N}$,
then prove that the regularization also make $\bA_M$ diagonalizable
with real eigenvalues.

As discussed above, only the last row of $\hat{A}_0$ are to be
modified in the regularization. Property \ref{pro:reducible} shows
$\hat{A}_0 = \bB_M(\rho,\theta_{11},f_{3e_1},\cdots,f_{Me_1})$. And
for $D=1$, the coefficient matrix $\bA'_M=\hat{A}_0$. In \cite{Fan},
the regularization with $D=1$ is studied in details, and the result
therein we will need later on is as below.
\begin{lemma}\label{lem:1D}
    Let 
    \[
        \tilde{\bB}_M\pd{\bw}{x}=\bB_M\pd{\bw}{x}-
        (M+1)\left(f_{Me_1}\pd{u}{x}
        +\frac{f_{M-1}}{2\rho}\left(\pd{p}{x}-\theta\pd{\rho}{x}\right)
        \right)I_{M+1},
    \]
    for any admissible $\bw$, i.e.,
    \begin{equation}\label{eq:modifyB}
        \tilde{\bB}_M=\bB_M-I_{M+1} \mathcal{R}_0^T,
    \end{equation}
    where $\mathcal{R}_0=(M+1)(-\theta f_{M-1}/2\rho, f_M,
    f_{M-1}/\rho,0,\cdots,0)^T\in\bbR^{M+1}$ and $I_{M+1}$ is the last
    column of the $(M+1)\times(M+1)$ identity matrix.  Then
    $\tilde{\bB}_M$ is diagonalizable with real eigenvalues.
    Precisely, the characteristic polynomial of $\tilde{\bB}_M$ is
    \[
        \det(\lambda\boldsymbol{I}-\tilde{\bB}_M)=
        \theta^{M+1}\He^{[\theta]}_{M+1}(\lambda),
    \]
    and the eigenvalues of $\tilde{\bB}_M$ are
    $\sqrt{\theta}\rC{1}{M+1},\cdots,\sqrt{\theta}\rC{M+1}{M+1}$,
    where $\rC{j}{k}$ is the $j$-th root of Hermite polynomial
    $\He_k(x)$, noticing that $\He_k(x), k\in\bbN$ has $k$ different
    zeros, which read $\rC{1}{k},\dots,\rC{k}{k}$, and satisfy
    $\rC{1}{k}<\dots<\rC{k}{k}$.
    Let $\br\in\bbR^{M+1}$ and
    \begin{align*}
        &r_1=1,\quad r_2=\lambda/\rho,\quad
        r_3=\lambda^2/2,\\
        &r_k=\He_{k-1}^{[\theta]}(\lambda)/(k-1)!-\lambda
        f_{k-2}/\rho-(\lambda^2-1)f_{k-3}/(2\rho),\quad k=4,\cdots,M+1,
    \end{align*}
    where $\lambda$ is an eigenvalue of $\bB_M$, then $\br$ is an
    eigenvector of $\bB_M$ for the eigenvalue $\lambda$.

    Moreover, the regularization is admissible and the admissible
    regularization to modify $\bB_M$ to be diagonalizable with real
    eigenvalues with the characteristic polynomial
    $\theta^{M+1}\He^{[\theta]}_{M+1}$ is unique.
\end{lemma}
\begin{remark}
  Since $f_{\alpha}$ are related to $f(t, \bx, \bxi)$ by
  \eqref{eq:expansion}, the positivity of the distribution function
  will impose some constraints on the $f_{\alpha}$.  Particularly,
  $\rho$ and $\Theta$ satisfy
  \begin{equation}\label{eq:constraints}
    \rho>0\quad \text{and} \quad\Theta
    \text{ being a symmetrical positive definite matrix.}
  \end{equation}
  Though \eqref{eq:constraints} is not enough to ensure the positivity
  of $f(t, \bx, \bxi)$, the discussion in this section requires no
  further constraints on all the other variables. Hence, the
  admissible $\bw'$ stands for the $\bw'$ satisfying
  \eqref{eq:constraints} in this section.
\end{remark}

We extend the results of $D=1$ to any dimensional case.
\begin{definition}\label{def:regularization}
    $\tbAM$ is called the regularized matrix of $\bA_M$, if it
    satisfies that for any admissible $\bw$,
    \begin{equation}\label{eq:regularization}
        \begin{split}
            \tbAM\pd{\bw}{x_1}&=\bA_M\pd{\bw}{x_1}\\
            &-\sum_{|\alpha|=M}(\alpha_1+1)\left( 
            \sum_{i=1}^Df_{\alpha+e_1-e_i}\pd{u_i}{x_1} 
            +\sum_{i,j=1}^D\frac{f_{\alpha+e_1-e_i-e_j}}{2\rho}
            \left(\pd{p_{ij}}{x_1}-\theta_{ij}\pd{\rho}{x_1}\right)
            \right)I_{\mN(\alpha)},
        \end{split}
    \end{equation}
    where $I_{k}$ is the $k$-th column of the $N\times N$ identity
    matrix.
\end{definition}

In this the definition of the regularized matrix $\tbAM$, $\tbAM$ is
obtained by changing a few entries of $\bA_M$. Precisely for any
$|\alpha|=M$, let $k=\mN(\alpha)$
\begin{align*}
    \tbAM(k,1) &= \bA_M(k,1) +
    (\alpha_1+1)\sum_{i,j=1}^D\frac{\theta_{ij}f_{\alpha+e_1-e_i-e_j}}{2\rho},\\
    \tbAM(k,d+1) &= \bA_M(k,d+1) -
    (\alpha_1+1)f_{\alpha+e_1-e_d}, \quad d=1,\dots,D,\\
    \tbAM(k,\mN(e_i+e_j)) &= \bA_M(k,\mN(e_i+e_j)) -
    (\alpha_1+1)\frac{f_{\alpha+e_1-e_i-e_j}}{\rho}
    \quad i,j=1,\dots,D.
\end{align*}
Other entries of $\tbAM$ remain the same values as those of $\bA_M$.

For convenience, we list the regularized collisionless moment system
with $\pd{\cdot}{x_2}=\cdots=\pd{\cdot}{x_D}=0$, which is the case of
1D spatial space, as following:
\begin{subequations}\label{eq:modified_ms}
\begin{align}
    &\odd{\rho}{t}+\rho\pd{u_1}{x_1}=0,
    \label{eq:modified_mass}\\
    &\odd{u_i}{t}+\frac{1}{\rho}\pd{p_{1i}}{x_1}=0,
    \label{eq:modified_momentum}\\
    &\odd{p_{ij}}{t}+p_{ij}\pd{u_1}{x_1}
    +p_{1i}\pd{u_j}{x_1}+p_{1j}\pd{u_i}{x_1}
    +(e_i+e_j+e_1)!\pd{f_{e_i+e_j+e_1}}{x_1}
    =0,\label{eq:modified_pressure}\\
    \begin{split}\label{eq:modified_f}
        &\odd{f_{\alpha}}{t} +  \sum_{k = 1}^D \theta_{1k}
        \pd{f_{\alpha - e_k}}{x_1} +
        (1-\delta_{|\alpha|,M})(\alpha_1+1)\pd{f_{\alpha+e_1}}{x_1}\\
        &\qquad+\sum_{i,j=1}^D\frac{\tilde{C}_{ij}(\alpha)}{2\rho}\left(
        \pd{p_{ij}}{x_1}-\theta_{ij}\pd{\rho}{x_1} \right)
        +  \sum_{i = 1}^D
        (1-\delta_{|\alpha|,M})(\alpha_1+1)f_{\alpha-e_i+e_1}\pd{u_i}{x_1}\\ 
        &\qquad-  \sum_{i = 1}^D \frac{f_{\alpha-e_i}}{\rho} \pd{p_{i1}}{x_1} 
        -  \sum_{i,j=1}^D \frac{(e_i+e_j+e_1)!}{2}
        \frac{f_{\alpha-e_i-e_j}}{\rho} \pd{f_{e_i+e_j+e_1}}{x_1}
        = 0,
  \end{split}
\end{align}
\end{subequations}
where $\tilde{C}_{ij}$ is 
\begin{equation}\label{eq:def_tildeCijd}
    \tilde{C}_{ij}(\alpha)= \sum_{k=1}^D\theta_{k 1} f_{\alpha - e_i -
    e_j - e_k} + (1-\delta_{|\alpha|,M})( \alpha_1 + 1) f_{\alpha-e_i-e_j+e_1}.
\end{equation}
Clearly, the equations \eqref{eq:modified_mass},
\eqref{eq:modified_momentum}, \eqref{eq:modified_pressure} and
\eqref{eq:modified_f} is the simplified formulation of the regularized
moment system, and the entries of the matrix $\tbAdotM$ can be
retrieved directly from the system.

Notice that $\tbAdotM=\boldsymbol{P}\tbAM\boldsymbol{P}^{-1}$ is the
regularized matrix of $\bA'_M$, where $\boldsymbol{P}$ is the
permutation matrix, such that $\bA'_M = \boldsymbol{P} \bA_M
\boldsymbol{P}^{-1}$. Clearly, only the rows of $\tbAdotM$
corresponding to the last rows of $\hat{A}_k$, $k=1,\dots,\hat{N}$ are
different from those of $\bA'_M$. Particularly, $\tilde{\bB}_M$ is
defined in \eqref{eq:modifyB}, and $\thA_1$ and $\thA_2$ are denoted
by
\begin{subequations}
    \begin{align}
        \thA_1 &=\hat{A}_1-I_M^{(M)}\mathcal{R}_1^T,\\
        \thA_2 &= \hat{A}_2-I_{M-1}^{(M-1)}\mathcal{R}_2^T,
    \end{align}
\end{subequations}
where $I_k^{(n)}$ is the $k$-th column of the $n\times n$ identity
matrix, and 
\begin{align*}
    \mathcal{R}_1&=M(f_{Me_1},~
    f_{(M-1)e_1}/\rho,~0,\cdots,0)^T\in\bbR^M,\\
    \mathcal{R}_2&=(M-1)(f_{(M-1)e_1}/\rho,~0,\cdots,0)^T\in\bbR^{M-1}.
\end{align*}
For $\hat{A}_k$, $k=3,\cdots,M$, we have $\thA_k=\hat{A}_k$.  Hence,
the regularization \eqref{eq:regularization} is admissible.

Here we give a note on the convention of the notations used
here. $\bA_M$ is the coefficient of the moment system
\eqref{eq:system} on the direction $x$, and
$\bA_M'=\boldsymbol{P}\bA_M\boldsymbol{P}^{-1}$ is a lower block
triangular matrix, where $\boldsymbol{P}$ satisfies
$\bw'=\boldsymbol{P}\bw$. $\hat{\alpha}=(\alpha_2,\cdots,\alpha_D)$
and $\hat{A}_{\hat{\alpha}}$, $|\hat{\alpha}|\leq M$ are diagonal
blocks of $\bA_M'$ defined in Property \ref{pro:reducible}.
$\tilde{\cdot}$ stands for the regularized matrix, such as, $\tbAM'$
is the regularized matrix of $\bA_M'$, and $\thA_{\hat{\alpha}}$ is
the regularized matrix of $\hat{A}_{\hat{\alpha}}$.

Lemma \eqref{lem:1D} shows for $D=1$, the regularization defined in
Definition \ref{def:regularization} make the coefficient matrix
diagonalizable with real eigenvalues. For arbitrary dimensional case,
we have the following results.
\begin{theorem}\label{thm:hyperbolic1}
    The regularized moment system
    \[
        \odd{\bw}{t}+\tbAM\pd{\bw}{x_1}=0
    \]
    is globally hyperbolic for any admissible $\bw$.
\end{theorem}
To prove this theorem, we need to verify the regularized matrix
$\tbAM$ is diagonalizable with real eigenvalues for any admissible
$\bw$. Since $\tbAdotM$ is a similar matrix of $\tbAM$, next we first
study the eigenvalues and eigenvectors of $\thA_k$, $k=1,\cdots,M$,
then we can obtain the characteristic polynomial of $\tbAdotM$, and
verify that all the eigenvalues of $\tbAdotM$ are real. Furthermore,
any eigenvector of $\thA_k$, $k=0,\cdots,M$ can be extended to an
eigenvector of $\tbAdotM$ under relevant constraints, and then we can
prove $\tbAdotM$ have $N$-linearly independent eigenvectors, which
means $\tbAdotM$ is diagonalizable.

Since $\thA_1$ is a lower Hessenberg matrix, it is possible to
calculate its eigenvector, once the eigenvalue is given. Actually, we
have the following lemma.
\begin{lemma}\label{lem:bAsecond}
  The matrix
  $\thA_1\in\bbR^{M\times M}$ is diagonalizable with real eigenvalues
  for any $\rho>0$, $\theta_{11}>0$, $f_{ke_1}\in\bbR$,
  $k=3,\cdots,M-1$.  Precisely, its characteristic polynomial is
  \begin{equation}
      \det(\lambda\boldsymbol{I}-\thA_1)=
      \theta_{11}^{M}\He^{[\theta_{11}]}_{M}(\lambda),
  \end{equation}
and the eigenvalues of $\thA_1$ are
$\sqrt{\theta_{11}}\rC{1}{M},\dots,\sqrt{\theta_{11}}\rC{M}{M}$.
Let $\br\in\bbR^M$ and 
\begin{align*}
    &r_1=1,\quad r_2=\rho\lambda,\quad
    r_k=\rho\He_{k-1}^{[\theta_{11}]}(\lambda)/(k-1)!-f_{(k-1)e_1}-\lambda
    f_{(k-2)e_1},\quad k=2,\dots,M,
\end{align*}
then $\br$ is an eigenvector of $\thA_1$ for the eigenvalue $\lambda$.
\end{lemma}
The proof is trivial but rather tedious, which is presented in the
Appendix \ref{sec:proofLemmabAsecond}. Analogously, the matrix
$\thA_2\in\bbR^{(M-1)\times (M-1)}$ has the following properties.
\begin{lemma}\label{lem:bAthird}
The matrix $\thA_2\in\bbR^{(M-1)\times (M-1)}$ is diagonalizable with
real eigenvalues for any $\rho>0$, $\theta_{11}>0$, $f_{ke_1}\in\bbR$,
$k=3,\cdots,M-2$.  Precisely, its characteristic polynomial is
\begin{equation}
    \det(\lambda\boldsymbol{I}-\thA_2)=
    \theta_{11}^{M-1}\He^{[\theta_{11}]}_{M-1}(\lambda),
\end{equation}
and the eigenvalues of $\thA_2$ are
$\sqrt{\theta_{11}}\rC{1}{M-1},\cdots,\sqrt{\theta_{11}}\rC{M-1}{M-1}$.
Let $\br\in\bbR^{M-1}$ and 
\begin{align*}
    &r_1=1,\quad 
    r_k=\He_{k-1}^{[\theta_{11}]}(\lambda)/(k-1)!-f_{(k-1)e_1},
    \quad k=2,\dots,M-1,
\end{align*}
then $\br$ is an eigenvector of $\thA_2$ for the eigenvalue $\lambda$.
\end{lemma}
For the matrix $\thA_n$, $n=3,\dots,M$, we have the following results.
\begin{lemma}\label{lem:bAelse}
The matrix $\thA_n$, $n=3,\cdots,M$ is diagonalizable with real
eigenvalues for any $\theta_{11}>0$. Precisely, let $m=M+1-n$, then
the characteristic polynomial of $\thA_n$ is
\begin{equation}
    \det(\lambda\boldsymbol{I}-\thA_n )=
    \theta_{11}^{m}\He^{[\theta_{11}]}_{m}(\lambda),
\end{equation}
and the eigenvalues of $\thA_n$ are
$\sqrt{\theta_{11}}\rC{1}{m},\dots,\sqrt{\theta_{11}}\rC{m}{m}$.
Let $\br\in\bbR^m$ and 
\begin{align*}
    r_k=\He_{k-1}^{[\theta_{11}]}(\lambda)/(k-1)!, \quad k=1,\dots,m,
\end{align*}
then $\br$ is an eigenvector of $\thA_n$ for the eigenvalue $\lambda$.
\end{lemma}
The proof of these two lemmas are almost the same as that of Lemma
\ref{lem:bAsecond} thus we omit it.

Property \ref{pro:reducible} and Lemma \ref{lem:1D},
\ref{lem:bAsecond}, \ref{lem:bAthird} and \ref{lem:bAelse} show the
regularized matrix $\tbAdotM$ is a block lower triangular matrix and
the characteristic polynomial of each diagonal block is known.  Since
$\tbAdotM$ is a similar matrix of $\tbAM$, we have the following
result on the characteristic polynomial of $\tbAM$.
\begin{lemma}\label{lem:eigenpolynomial}
    Let
\begin{align}
    \mathcal{P}_{1,m} &= \theta_{11}^{m+1}\He^{[\theta_{11}]}_{m+1}(\lambda),\quad m\in\bbN,\\
\mathcal{P}_{d,m} &=
\prod_{k=0}^{m}\mathcal{P}_{d-1,k}, \quad 1<d\in\bbN^+ .\label{eq:eigenpolynomials}
\end{align}
$\mathcal{P}_{D,M}$ is the characteristic polynomial of
 $\tbAM$.
\end{lemma}
\begin{proof}
    If $D=1$, it is part of Lemma \ref{lem:1D}. Next we consider the
    case $D \ge 2$.  Since $\tbAM$ is similar to $\tbAdotM$, the
    characteristic polynomial of $\tbAM$ is same as that of
    $\tbAdotM$. Thus,
    \begin{align*}
        \det(\lambda\boldsymbol{I}-\tbAM) &=
        \det(\lambda\boldsymbol{I}-\tbAdotM)\\
        &=\prod_{|\hat{\alpha}|\leq M} 
        \det(\lambda\boldsymbol{I}-\tilde{A}_{\hat{\alpha}})\\
        &=\prod_{m=0}^{M}\left(\theta_{11}^{m+1}\He^{[\theta_{11}]}_{m+1}
        \right)^{\binom{D-2}{M-m+D-2}}\\
        &=\mathcal{P}_{D,M}.
    \end{align*}
    The second equality is then obtained by induction on $D$.
\end{proof}

Until now, we have revealed that the regularized matrix $\tbAdotM$ is
a block lower triangular matrix and each diagonal block is
diagonalizable with real eigenvalues. Unfortunately, it is not
sufficient to conclude that $\tbAdotM$ is diagonalizable yet, since
some eigenvalues may be not semi-simple. This pushes us to clarify the
structure of the eigen-subspace of $\tbAdotM$. Actually, we will
demonstrate that based on any eigenvector of $\thA_k$, $k=0,\cdots,M$,
we can construct an eigenvector of $\tbAdotM$. For this purpose, we
start with an example.
\begin{example}\label{exam:eigenvectorprolongation}
    Consider the block lower triangular matrix
    \begin{equation}
        \bA =
        \left(\begin{array}{ccc|cc}
            1   &   14  &   36  & 0 & 0 \\
            16  &   -10 &   -54 & 0 & 0 \\
            -10 &   4   &   27  & 0 & 0 \\
            \hline
            2   &   1   &   1   &   6   &   3\\
            4   &   2   &   2  &   6   &   9
        \end{array}
        \right)
        =\begin{pmatrix}
            A_{11}  &   \mathbf{0}\\
            A_{21}  &   A_{22}
        \end{pmatrix}.
    \end{equation}
    The eigenvalues and eigenvectors of $A_{11}$ and $A_{22}$ that
    \begin{align*}
        A_{11}:&\quad \lambda_1=3, \br_1=(16, -26,11)^T;~
        \lambda_2=6, \br_2=(10, -17, 8)^T;~
        \lambda_3=9, \br_3=(1, -2, 1)^T;\\
        A_{22}:&\quad \lambda_4=3, \br_4=(1, -1)^T;~
        \lambda_5=12, \br_5=(1, 2)^T.
    \end{align*}
    Now we examine whether there is $\bR_i\in\bbR^5$ and
    $\bR_i(1:3)=\br_i$, $i=1,2,3$, satisfying
    $\bA\bR_i=\lambda_i\bR_i$. Actually, it is equivalent to whether
    there is a solution of 
    $A_{21}\br_i+(A_{22}-\lambda_i\boldsymbol{I})\bR_i(4:5)=0$, 
    which has a solution if and only if the augmented matrix of
    $A_{22}-\lambda\boldsymbol{I}$ has the same rank as
    $A_{22}-\lambda\boldsymbol{I}$, i.e.
    \begin{equation}\label{eq:rank_augmented}
        \rank([A_{22}-\lambda_i\boldsymbol{I}, A_{21}\br_i]) 
        = \rank(A_{22}-\lambda_i\boldsymbol{I}).
    \end{equation}
    For $i=2,3$, since $A_{22}-\lambda_i\boldsymbol{I}$ is nonsigular,
    \eqref{eq:rank_augmented} holds. For $i=1$, some simple
    calculations give that $\rank([A_{22}-\lambda_i\boldsymbol{I},
        A_{21}\br_i]) = \rank(A_{22}-\lambda_i\boldsymbol{I})=1$.

    Next we check whether there is $\bR_i\in\bbR^5$ and
    $\bR_i(4:5)=\br_i$, $i=4,5$, satisfying $\bA\bR_i=\lambda_i\bR_i$.
    Actually, $\bR_i(1:3)=(a,b, -2a-b)$ satisfies the condition.
    Particularly, if $a=b=0$, $\bR_i(1:3)=0$.

    Here we call $\bR_i$ is a {\rm prolongation} of $\br_i$, and call
    the $\bR_i$ with $\bR_i(1:3)=0$, $i=4,5$ a {\rm proper prolongation}
    of $\br_i$.  It is obvious that $\bR_i$, $i=1,\cdots,5$, is
    linearly independent.
\end{example}
\begin{definition}\label{def:prolongation}
    For a $k\times k$ block lower triangular matrix $\bA\in\bbR^N$
    with the size of diagonal block $n_i\times n_i$,
    $n_1+\cdots+n_k=N$, $\br_i$ is an eigenvector of the $i$-th
    diagonal block for the eigenvalue $\lambda$. We call $\bR$ is a
    {\rm prolongation} of $\br_i$, if
    $\bR((n_1+\cdots,n_{i-1}+1):(n_1+\cdots+n_i))=\br_i$ and
    $\bA\bR=\lambda\bR$. Particularly, $\bR$ is a {\rm proper
    prolongation} of $\br_i$ if $\bR(1:(n_1+\cdots+n_{i-1}))=0$.
\end{definition}
\begin{property}\label{pro:prolongation}
    $\bA$ is defined same as that in Definition \ref{def:prolongation},
    and each diagonal block of $\bA$ is diagonalizable with
    real eigenvalue. $\br_{i,1},\cdots,\br_{i,n_i}$ are eigenvectors
    of the $i$-th diagonal block of $\bA$. If for each $\br_{i,j},
    i=1,\cdots,k$, $j=1,\cdots,n_i$, there is a proper prolongation
    $\bR_{i,j}$, then $\bR_{i,j}$ are linearly independent.
\end{property}
\begin{proof}
    We permute $\bR_{i,j}$ by the order
    $\vec{\bR}=[\bR_{1,1},\cdots,\bR_{1,n_1},\bR_{2,1},\cdots,\bR_{k,n_k}]$.
    Since $\bR_{i,j}$, $i=1,\cdots,k$, $j=1,\cdots,n_i$ is a proper
    prolongation of $\br_{i,j}$, $\vec{\bR}$ is a block lower triangular
    matrix.  For a fixed $i\in\{1,\cdots,k\}$, $\br_{i,j}$,
    $j=1,\cdots,n_i$ are linearly independent. Hence each diagonal
    block of $\vec{\bR}$ are nonsigular, thus $\vec{\bR}$ is
    nonsigular, which indicates $\bR_{i,j}$, $i=1,\cdots,k$,
    $j=1,\cdots,n_i$ are linearly independent.
\end{proof}

Next we check whether there is a proper prolongation of each
eigenvector of every diagonal block of $\tbAdotM$ and get the
following result.
\begin{lemma}\label{lem:properprolongation}
    $\br_{\alpha}$ is an eigenvector of
    $\hat{A}_{\hat{\alpha}}\in\bbR^{(M+1-|\hat{\alpha}|)\times(M+1-|\hat{\alpha}|)}$,
    for the $\alpha_1$-th eigenvalue
    $\lambda=\sqrt{\theta_{11}}\rC{\alpha_1}{M+1-|\hat{\alpha}|}$,
    then there is a proper prolongation $\bR_{\alpha}$ satisfying
    $\tbAdotM\bR_{\alpha}=\lambda\bR_{\alpha}$.
\end{lemma}
The proof of the lemma is rather long and tedious, so we move the
proof in Appendix \ref{sec:proofLemmaprolongation}. 

Clearly, this lemma is essential to prove Theorem
\ref{thm:hyperbolic1}. With all these preparation, the proof of
Theorem \ref{thm:hyperbolic1} as follows is straight forward:
\begin{proof}[Proof of the Theorem \ref{thm:hyperbolic1}]
    Lemma \ref{lem:eigenpolynomial} shows all the eigenvalues of
    $\tbAM$ are real. Since $\tbAdotM$ is similar to $\tbAM$, all
    the eigenvalues of $\tbAdotM$ are also real. Lemma
    \ref{lem:1D}, \ref{lem:bAsecond}, \ref{lem:bAthird} and
    \ref{lem:bAelse} show that each diagonal block of $\tbAdotM$ is
    diagonalizable with real eigenvalues, and Lemma
    \ref{lem:properprolongation} indicates each eigenvector of each
    diagonal block can be extended to an eigenvector of $\tbAdotM$ by
    a proper prolongation. Hence considering Property \ref{pro:prolongation},
    we obtain that $\tbAdotM$ is diagonalizable with real eigenvalues.
    This finishes the proof.
\end{proof}

In addition, for the eigenvector of $\tbAdotM$, we define
$\bR=(R_{\alpha})^T\in\bbR^N$, where $R_{\alpha}$ is permuted by the
lexicographic order of $(\alpha_2,\cdots,\alpha_D,\alpha_1)$, same as
that of $\bw'$. Particularly, the first entry of $\bR$ is $R_0$. Hence
we have $\boldsymbol{P}\bR=(R_{\alpha})^T\in\bbR^N$, where
$\boldsymbol{P}$ satisfies $\bw'=\boldsymbol{P}\bw$, and $R_{\alpha}$
is permuted same as that of $\bw$.

Before we end this subsection, we give a corollary of Lemma
\ref{lem:properprolongation}, which will be used in Section
\ref{sec:RiemannProblem}.
\newcommand\hHe{\hat{\He}^{[\theta_{11}]}}
\begin{corollary}\label{cor:lambdar}
    Let $\bR\neq0$ be a right eigenvector of the matrix $\tbAdotM$ for the
    eigenvalue $\lambda$. Then 
    \[
        \lambda R_0 \neq0
        \text{ holds if and only if }
        \hHe_{M+1}(\lambda)=0 \text{ and } \lambda\neq 0,
    \]
    where $\hHe_{M+1}(\lambda)$ is the characteristic polynomial of
    $\thA_0$.
\end{corollary}
\begin{proof}
    Let $\br=(R_0,R_{e_1},\cdots,R_{Me_1})^T\in\bbR^{M+1}$. Since
    $\tbAdotM$ is a block lower triangular matrix, and the first
    diagonal block is $\thA_0$, we have $\thA_0\br=\lambda\br$.

    \noindent``$\Rightarrow$''\quad Since $R_0\neq0$ and
    $\lambda\neq0$, $\br\neq0$ is an eigenvector of $\thA_0$ for the
    eigenvalue $\lambda$, which indicates $\hHe(\lambda)=0$.  Thus we
    have $\hHe(\lambda)=0$ and $\lambda\neq0$.

    \noindent``$\Leftarrow$''\quad $\hHe(\lambda)=0$ and
    $\lambda\neq0$ mean $\lambda$ is an eigenvalue of $\thA_0$. Since 
    each nonzero eigenvalue of $\thA_0$ is simple eigenvalue of
    $\tbAdotM$, we get $\br\neq0$. \eqref{eq:eigenvector1D} indicates
    $R_0\neq0$, so $R_0\lambda\neq0$.
\end{proof}

\subsection{In multi-dimensional spatial space}
In this subsection, we give the general hyperbolic moment system
containing all moments with orders not more than $M$. Without the
assumption that the distribution function $f$ is independent on
$x_2,\cdots,x_D$, the collisionless moment system obtained in Section
\ref{sec:expansion} can be written as 
\begin{equation}
    \odd{\bw}{t}+\sum_{d=1}^D\bA^{(d)}_M\pd{\bw}{x_d}=0,
\end{equation}
where $\bw$ and $\bA^{(d)}_M$, $d=1,\cdots,D$ are same as that in
\eqref{eq:system}. And particularly, $\bA^{(1)}_M$ is the matrix
$\bA_M$ discussed in Section \ref{sec:propertyOfA},
\ref{sec:lackofhyperbolicity} and \ref{sec:regularization1D}.
Similar as Definition \ref{def:regularization}, we give the following
definition:
\begin{definition}\label{def:regularizationD}
    For $d=1,\cdots,D$, $\tbAM^{(d)}$ is called the regularized matrix
    $\bA^{(d)}_M$, if it satisfies that for any admissible $\bw$,
    \begin{equation}\label{eq:regularizationD}
        \begin{split}
            \tbAM^{(d)}\pd{\bw}{x_d}&=\bA^{(d)}_M\pd{\bw}{x_d}\\
            &-\sum_{|\alpha|=M}(\alpha_d+1)\left( 
            \sum_{i=1}^Df_{\alpha+e_d-e_i}\pd{u_i}{x_d} 
            +\sum_{i,j=1}^D\frac{f_{\alpha+e_d-e_i-e_j}}{2\rho}
            \left(\pd{p_{ij}}{x_d}-\theta_{ij}\pd{\rho}{x_d}\right)
            \right)I_{\mN(\alpha)},
        \end{split}
    \end{equation}
    where $I_{k}$ is the $k$-th column of the $N\times N$ identity
    matrix.
\end{definition}

Then the multi-dimensional regularized moment system can be written as
\begin{equation}\label{eq:Rsystem}
    \odd{\bw}{t}+\sum_{d=1}^D\tbAM^{(d)}\pd{\bw}{x_d}=0.
\end{equation}
Recalling the definition of $\tbAM^{(d)}$, and noting the
regularized collisionless moment system in one-dimensional space
\eqref{eq:modified_ms}, we can reformulate the regularized
collisionless moment systems as
\begin{equation}\label{eq:regularizedsystem}
    \begin{split}
        \odd{f_{\alpha}}{t} ~+ & \sum_{d, k = 1}^D \left( \theta_{d k} \pd{f_{\alpha-e_k}}{x_d}
        + (1-\delta_{|\alpha|,M})\left( \alpha_k + 1 \right) \delta_{k d}
        \pd{f_{\alpha + e_k}}{x_d} \right) + \\ 
        \sum_{i = 1}^D f_{\alpha - e_i}
        \odd{u_i}{t} ~+ & \sum_{i, d, k = 1}^D \left( \theta_{d k} f_{\alpha - e_i - e_k} +
        (1-\delta_{|\alpha|,M})\left( \alpha_k + 1 \right) \delta_{k d} f_{\alpha - e_i + e_k} \right)
        \pd{u_i}{x_d} + \\
        \sum_{i, j = 1}^D \frac{
        f_{\alpha - e_i - e_j}}{2} \odd{\theta_{i j}}{t} ~+ & \sum_{i, j, d, k = 1}^D \frac{1}{2} 
        \left(\theta_{kd} f_{\alpha - e_i - e_j - e_k} + 
        (1-\delta_{|\alpha|,M})\left( \alpha_k + 1 \right)\delta_{k d}
        f_{\alpha - e_i - e_j + e_k} \right) \pd{\theta_{i j}}{x_d} \\
        =&~0,\qquad |\alpha|\leq M.
    \end{split}
\end{equation}
Actually, \eqref{eq:regularizedsystem} is obtained by using the
regularization on \eqref{eq:origalsystem}. Since the moment system
\eqref{eq:system} is derived from \eqref{eq:origalsystem} by
eliminating the material derivatives of $u_d$ and $\theta_{ij}$, there
exists an invertible matrix $\boldsymbol{T}(\bw)$ depending on $\bw$
such that the regularized moment system is identical to the following
system:
\begin{equation}\label{eq:regularizedsystem_t}
    \boldsymbol{T}(\bw)\odd{\bw}{t}+\sum_{d=1}^D\boldsymbol{T}(\bw)\tbAM^{(d)}\pd{\bw}{x_d}=0.
\end{equation}

The following theorem declares the hyperbolicity\footnote{For
  multi-dimensional quasi-linear systems, we refer the readers to
  \cite{FVM} for the definition of hyperbolicity.} of the
multi-dimensional regularized moment system \eqref{eq:Rsystem}:
\begin{theorem}\label{thm:rotationinvariance}
    The regularized moment system \eqref{eq:Rsystem} is hyperbolic for
    any admissible $\bw$. Precisely, for a given unit vector
    $\boldsymbol{n}=(n_1,\cdots,n_D)$, there exists a constant
    matrix $\boldsymbol{Z}$ partially depending on $\boldsymbol{n}$
    such that
    \begin{equation}\label{eq:rotationinvariance}
        \sum_{d=1}^Dn_d\tbAM^{(d)}(\bw)=\boldsymbol{Z}^{-1}\tbAM^{(1)}(\boldsymbol{Z}\bw)\boldsymbol{Z},
    \end{equation}
    and the matrix is diagonalizable with eigenvalues as
    \begin{equation}\label{eq:rotationeigenvalue}
        \rC{n}{m}\sqrt{\boldsymbol{n}^T \Theta \boldsymbol{n}},
        \quad 1\leq n\leq m\leq M+1.
    \end{equation}
\end{theorem}

Actually, this theorem gives the rotation invariance of the
regularized moment system \eqref{eq:Rsystem} and its globally
hyperbolicity. Property \ref{pro:diagonal} indicates the translation
invariance of the moment system, hence, it is concluded that the
regularized system is Galilean invariant. If another coordinate
$(x_1\rotation ,\cdots,x_D\rotation )$ is adopted and the vector
$\boldsymbol{n}$ is along the $x_1\rotation $-axis, then the rotated
moment system is equivalent to the original one. The rotation
invariance is intuitive: on one hand, the moment system
\eqref{eq:system} is rotationally invariant, since the full $M$-degree
polynomials are used in the truncated expansion; on the other
hand, the regularization is symmetric in every direction. In the
following, we will give a rigorous proof of this theorem.

Let $\boldsymbol{G}=(g_{ij})_{D\times D}$ to be the rotation matrix,
thus $\boldsymbol{G}$ is orthogonal and its determinant is 1. We
define
\begin{equation}
    x_i\rotation =\sum_{i=1}^Dg_{ij}x_j,\quad 
    i=1,\cdots,D,
\end{equation}
and denote by $\rho\rotation $, $\bu\rotation $ and $\Theta\rotation $ the density,
macroscopic velocity and temperature tensor in the new coordinate
$\bx\rotation =(x_1\rotation ,\cdots,x_D\rotation )$. If we define
$\bxi\rotation =\boldsymbol{G}\bxi$, then the orthogonality of
$\boldsymbol{G}$ shows
\begin{subequations}\label{eq:relation_rut}
\begin{align}
    \rho\rotation =\int_{\bbR^D}f(\bxi)\dd\bxi\rotation 
    &=\int_{\bbR^D}f(\bxi)\dd\bxi=\rho,\\
    \rho\rotation \bu\rotation =\int_{\bbR^D}\bxi\rotation f(\bxi)\dd\bxi\rotation 
    &=\int_{\bbR^D}\boldsymbol{G}\bxi
    f(\bxi)\dd\bxi=\rho\boldsymbol{G}\bu,\\
    \rho\rotation \theta_{ij}\rotation =\int_{\bbR^D}\xi\rotation _i\xi\rotation _jf(\bxi)\dd\bxi\rotation 
    &=\int_{\bbR^D}\sum_{k,l=1}^Dg_{ik}\xi_kg_{jl}\xi_lf(\bxi)\dd\bxi
    =\sum_{k,l=1}^Dg_{ik}g_{jl}\theta_{kl},
\end{align}
\end{subequations}
thus we have
\begin{equation}
    \bu\rotation =\boldsymbol{G}\bu,\quad
    \Theta\rotation =\boldsymbol{G}\Theta\boldsymbol{G}^T.
\end{equation}

Consider the two expansions
\begin{equation}\label{eq:doubleexpansion}
    f(\bxi)=\sum_{\alpha\in\bbN^D}f_{\alpha}\mHT_{\alpha}(\bxi-\bu)
    =\sum_{\alpha\in\bbN^D}f\rotation _{\alpha}\mH^{[\Theta\rotation ]}(\bxi\rotation -\bu\rotation ).
\end{equation}
We have the following result.
\begin{lemma}\label{lem:rotation}
    For any $m\in\bbN$, there exists a group of constants
    $Q_{\alpha}^{\beta}$, $|\alpha|=|\beta|=m$, such that
    for any $|\alpha|=m$,
    \begin{subequations}
        \begin{align}
            \begin{split}\label{eq:rotation_1}
            &\odd{\rotation f\rotation _{\alpha}}{\rotation t}+\sum_{i=1}^Df\rotation _{\alpha-e_i}\odd{\rotation u\rotation _i}{\rotation t}
            +\sum_{i,j=1}^D\frac{f\rotation _{\alpha-e_i-e_j}}{2}\odd{\rotation \theta_{ij}}{\rotation t} \\
            &\qquad=\sum_{|\beta|=|\alpha|}Q_{\beta}^{\alpha}\left( 
            \odd{f_{\alpha}}{t}+\sum_{i=1}^Df_{\alpha-e_i}\odd{u_i}{t}
            +\sum_{i,j=1}^D\frac{f_{\alpha-e_i-e_j}}{2}\odd{\theta_{ij}}{t}
            \right),
            \end{split}\\
            \begin{split}\label{eq:rotation_2}
                &\sum_{d,k=1}^D\theta\rotation _{dk}\left(
                \pd{f\rotation _{\alpha-e_k}}{x\rotation _d}+
                \sum_{i=1}^Df\rotation _{\alpha-e_i-e_k}\pd{u\rotation _i}{x\rotation _d}+
                \sum_{i,j}^D\frac{1}{2}f\rotation _{\alpha-e_i-e_j}\pd{\theta\rotation _{ij}}{x\rotation _d}\right)\\
                &\qquad=
                \sum_{|\beta|=|\alpha|}Q_{\beta}^{\alpha}\left(
                \sum_{d,k=1}^D\theta_{dk}\left( \pd{f_{\alpha-e_k}}{x_d}+
                \sum_{i=1}^Df_{\alpha-e_i-e_k}\pd{u_i}{x_d}+
                \sum_{i,j}^D\frac{1}{2}f_{\alpha-e_i-e_j}\pd{\theta_{ij}}{x_d}\right)
                \right),
            \end{split}\\
            \begin{split}\label{eq:rotation_3}
                &\sum_{d=1}^D(\alpha_d+1)\left(
                \pd{f\rotation _{\alpha+e_d}}{x\rotation _d}
                +\sum_{i=1}^Df\rotation _{\alpha-e_i+e_d}\pd{u\rotation _i}{x\rotation _d}
                +\sum_{i,j=1}^D\frac{1}{2}f\rotation _{\alpha-e_i-e_j+e_d}\pd{\theta\rotation _{ij}}{x\rotation _d}
                \right)\\
                &\qquad \sum_{|\beta|=|\alpha|}Q_{\beta}^{\alpha}\left(
                \sum_{d=1}^D(\alpha_d+1)\left( \pd{f_{\alpha+e_d}}{x_d}
                +\sum_{i=1}^Df_{\alpha-e_i+e_d}\pd{u_i}{x_d}
                +\sum_{i,j=1}^D\frac{1}{2}f_{\alpha-e_i-e_j+e_d}\pd{\theta_{ij}}{x_d}
                \right)\right),
            \end{split}
        \end{align}
    \end{subequations}
    where $\odd{\rotation \cdot}{\rotation t}$ denotes
    $\pd{\cdot}{t}+\displaystyle\sum_{d=1}^Du\rotation
    _d\pd{\cdot}{x\rotation _d}$.
\end{lemma}
\begin{proof}
    Since $\bu\rotation =\boldsymbol{G}\bu$ and $\bx\rotation =\boldsymbol{G}\bx$ hold, 
    and $\boldsymbol{G}$ is an orthogonal matrix, we have 
    \[
        \sum_{d=1}^Du\rotation _d\pd{\cdot}{x\rotation _d}=\sum_{d=1}^D\sum_{i=1}^Dg_{di}u_i
        \sum_{j=1}^Dg_{dj}\pd{\cdot}{x_j}=\sum_{i,j=1}^D\delta_{ij}u_i\pd{\cdot}{x_j}
        =\sum_{d=1}^Du_d\pd{\cdot}{x_d}.
    \]
    Here $\sum_{d=1}^Dg_{id}g_{jd}=\delta_{ij}$ is used in the second
    equality. Thus, we have 
    \[
        \odd{\cdot}{t}=\odd{\rotation \cdot}{\rotation t}.
    \]
    Since 
    \begin{align*}
        \odd{f(\bxi)}{t}&=\sum_{\alpha\in\bbN^D}
        \left( \odd{f_{\alpha}}{t} +
        \sum_{i=1}^Df_{\alpha-e_i}\odd{u_i}{t}+
        \sum_{i,j=1}^D\frac{f_{\alpha-e_i-e_j}}{2}\odd{\theta_{ij}}{t}\right)
        \mHT_{\alpha}(\bxi-\bu),
    \end{align*}
    considering the expansion \eqref{eq:doubleexpansion}, we have
    \begin{equation}\label{eq:compare_ft}
    \begin{split}
        &\sum_{\alpha\in\bbN^D}
        \left( \odd{f_{\alpha}}{t} +
        \sum_{i=1}^Df_{\alpha-e_i}\odd{u_i}{t}+
        \sum_{i,j=1}^D\frac{f_{\alpha-e_i-e_j}}{2}\odd{\theta_{ij}}{t}\right)
        \mHT_{\alpha}(\bxi-\bu) \\
        &\qquad\qquad = 
        \sum_{\alpha\in\bbN^D}
        \left( \odd{f\rotation _{\alpha}}{t} +
        \sum_{i=1}^Df\rotation _{\alpha-e_i}\odd{u\rotation _i}{t}+
        \sum_{i,j=1}^D\frac{f\rotation _{\alpha-e_i-e_j}}{2}\odd{\theta\rotation _{ij}}{t}\right)
        \mH^{[\Theta\rotation ]}_{\alpha}(\bxi\rotation -\bu\rotation ).
    \end{split}
    \end{equation}
    The rotation relation \eqref{eq:rotationrelation} of the
    generalized Hermite functions indicates that there exists a group
    of constants $Q_{\alpha}^{\beta}$, $|\alpha|=|\beta|$, such that
    \begin{equation}\label{eq:rotationrelation2}
        \mHT_{\alpha}(\bx)=\sum_{|\beta|=|\alpha|}Q_{\alpha}^{\beta}
        \mH^{[\boldsymbol{G}\Theta\boldsymbol{G}^T]}_{\alpha}
        (\boldsymbol{G}\bx),
    \end{equation}
    and the matrix $(Q_{\alpha}^{\beta})|_{|\alpha|=|\beta|=m}$,
    formulated by collecting these constants $Q_{\alpha}^{\beta}$, is
    non-singular. Substituting \eqref{eq:rotationrelation2} into
    \eqref{eq:compare_ft}, we obtain
    \begin{equation*}
        \begin{split}
            &\sum_{\alpha\in\bbN^D}\sum_{|\beta|=|\alpha|}
            \left( \odd{f_{\alpha}}{t} +
            \sum_{i=1}^Df_{\alpha-e_i}\odd{u_i}{t}+
            \sum_{i,j=1}^D\frac{f_{\alpha-e_i-e_j}}{2}\odd{\theta_{ij}}{t}\right)
            Q_{\alpha}^{\beta}\mH^{[\Theta\rotation ]}_{\beta}(\bxi\rotation -\bu\rotation ) \\
            &\qquad\qquad = 
            \sum_{\alpha\in\bbN^D}
            \left( \odd{f\rotation _{\alpha}}{t} +
            \sum_{i=1}^Df\rotation _{\alpha-e_i}\odd{u\rotation _i}{t}+
            \sum_{i,j=1}^D\frac{f\rotation _{\alpha-e_i-e_j}}{2}\odd{\theta\rotation _{ij}}{t}\right)
            \mH^{[\Theta\rotation ]}_{\alpha}(\bxi\rotation -\bu\rotation ).
        \end{split}
    \end{equation*}
    Comparing the coefficient of
    $\mH^{[\Theta\rotation ]}_{\alpha}(\bxi\rotation -\bu\rotation )$, we obtain
    \eqref{eq:rotation_1}.

    Analogously, since $\Theta\rotation=\boldsymbol{G}\Theta\boldsymbol{G}^T$,
    we have
    \[
        \sum_{d=1}^D(\xi_d\rotation -u_d\rotation )\pd{\cdot}{x\rotation _d}
        =\sum_{d,i,j=1}g_{id}(\xi_i-u_i)g_{jd}\pd{\cdot}{x_j}
        =\sum_{d=1}^D(\xi_d-u_d)\pd{\cdot}{x_d},
    \]
    and 
    \[
        \sum_{k,d=1}^D\theta\rotation _{kd}\pd{ }{\xi\rotation _k}\pd{\cdot}{x\rotation _d}=
        \sum_{k,d,i,j=1}^Dg_{ki}\theta_{ij}g_{dj}g_{ki}\pd{
        }{\xi_i}g_{dj}\pd{\cdot}{x_j}
        =\sum_{k,d=1}^D\theta_{kd}\pd{ }{\xi_k}\pd{\cdot}{x_d}.
    \]
    Since 
    \begin{equation}\label{eq:relation_tmp1}
    \begin{split}
        &\sum_{k,d=1}^D\theta_{kd}\pd{ }{\xi_k}\pd{f(\bxi)}{x_d}=\\
        &-\sum_{\alpha\in\bbN^D}\sum_{k,d=1}^D\theta_{kd}\left( 
        \pd{f_{\alpha-e_k}}{x_d}+\sum_{i=1}^Df_{\alpha-e_k-e_i}\pd{u_i}{x_d}
        +\sum_{i,j=1}^D\frac{f_{\alpha-e_k-e_i-e_j}}{2}\pd{\theta_{ij}}{x_d}
        \right)\mHT(\bxi-\bu),
    \end{split}
    \end{equation}
    using the same procedure in proving \eqref{eq:rotation_1}, we can
    have \eqref{eq:rotation_2}.

    For the operator
    $\displaystyle\sum_{d=1}^D(\xi_d-u_d)\pd{\cdot}{x_d}$, we have
    \begin{align*}
        &\sum_{d=1}^D(\xi_d-u_d)\pd{f(\bxi)}{x_d}=\\
        &-\sum_{\alpha\in\bbN^D}\left(\sum_{k,d=1}^D\theta_{kd}\left( 
        \pd{f_{\alpha-e_k}}{x_d}+\sum_{i=1}^Df_{\alpha-e_k-e_i}\pd{u_i}{x_d}
        +\sum_{i,j=1}^D\frac{f_{\alpha-e_k-e_i-e_j}}{2}\pd{\theta_{ij}}{x_d}
        \right)\right.\\
        &\left.\sum_{d=1}^D(\alpha_d+1)\left(
        \pd{f\rotation _{\alpha+e_d}}{x\rotation _d}
        +\sum_{i=1}^Df\rotation _{\alpha-e_i+e_d}\pd{u\rotation _i}{x\rotation _d}
        +\sum_{i,j=1}^D\frac{1}{2}f\rotation _{\alpha-e_i-e_j+e_d}\pd{\theta\rotation _{ij}}{x\rotation _d}
        \right)
        \right)\mHT_{\alpha}(\bxi-\bu).
    \end{align*}
    Noting that the second line is that in \eqref{eq:relation_tmp1},
    we can obtain \eqref{eq:rotation_3} by using the same procedure in
    proving \eqref{eq:rotation_1}. This ends the proof.
\end{proof}

\begin{proof}[Proof of Theorem \ref{thm:rotationinvariance}]
  Since $\boldsymbol{n}=(n_1,\cdots,n_D)$ is a unit vector, we let
  $\boldsymbol{G}=(g_{ij})_{D\times D}$ be the orthogonal rotation
  matrix with its first row as $(n_1,\cdots,n_D)$. With this rotation
  matrix, we define $\bw\rotation$ as \eqref{eq:relation_rut} and
  \eqref{eq:doubleexpansion}. Then the relation between $\bw$ and
  $\bw\rotation $ is linear. Therefore, there exists a constant matrix
  $\boldsymbol{Z}$ depending on $\boldsymbol{G}$ such that
    \[
        \bw\rotation =\boldsymbol{Z}\bw,
    \]
    and $\boldsymbol{Z}$ is invertible, since $\bw$ can be obtained
    from $\bw\rotation $ by applying the rotation matrix $\boldsymbol{G}^{-1}$.

    Lemma \ref{lem:rotation} have clearly shown that the ``rotated
    equations''
    \begin{equation}\label{eq:rotated_system}
        \boldsymbol{T}(\bw\rotation )\odd{\bw\rotation }{t}+\sum_{d=1}^D\boldsymbol{T(\bw\rotation )}
        \tbAM^{(d)}(\bw\rotation )\pd{\bw\rotation }{x\rotation _d}=0.
    \end{equation}
    can be deduced from \eqref{eq:regularizedsystem_t} by a linear
    transformation. Hence, there exists a square matrix
    $\boldsymbol{H}(\bw)$ such that
    \begin{equation}\label{eq:equivalentsystem}
        \boldsymbol{H}(\bw)\boldsymbol{T}(\bw)\odd{\bw}{t}
        +\sum_{d=1}^D\boldsymbol{H}(\bw)\boldsymbol{T}(\bw)
        \tbAM^{(d)}\pd{\bw}{x_d}=0
    \end{equation}
    is identical to \eqref{eq:rotated_system}. Matching the terms with
    time derivatives, one can find
    $\boldsymbol{H}(\bw)=\boldsymbol{T}(\bw\rotation )
    \boldsymbol{Z}\boldsymbol{T}^{-1}(\bw)$. Thus
    \eqref{eq:equivalentsystem} can be written as
    \[
        \boldsymbol{T}(\bw\rotation )\odd{\bw\rotation }{t}+\sum_{d=1}^D\boldsymbol{T(\bw\rotation )}
        \boldsymbol{Z}\tbAM^{(d)}(\bw)\pd{\bw}{x_d}=0.
    \]
    Noting $\bx\rotation =\boldsymbol{G}\bx$, we can rewrite the upper equation
    as
    \[
        \boldsymbol{T}(\bw\rotation )\odd{\bw\rotation }{t}+\sum_{j,d=1}^Dg_{jd}
        \boldsymbol{T(\bw\rotation )}\boldsymbol{Z}\tbAM^{(d)}(\bw)\pd{\bw}{x\rotation _j}=0.
    \]
    Comparing with \eqref{eq:rotated_system}, one concludes 
    \[
        \sum_{d=1}^Dg_{1d}
        \boldsymbol{T(\bw\rotation )}\boldsymbol{Z}\tbAM^{(d)}(\bw)
        =\boldsymbol{T}(\bw\rotation )\tbAM^{(1)}(\boldsymbol{Z}\bw)\boldsymbol{Z}.
    \]
    Multiplying both sides by
    $\boldsymbol{Z}^{-1}\boldsymbol{T}(\bw\rotation )^{-1}$, and
    noting $g_{1d}=n_d$, we obtain \eqref{eq:rotationinvariance}.

    Since the macroscopic temperature tensor are
    $\Theta\rotation =\boldsymbol{G}\Theta\boldsymbol{G}^T$(see
    \eqref{eq:relation_rut}), and particularly,
    $\theta\rotation _{11}=\sum_{i,j=1}^Dn_in_j\theta_{ij}$, the
    diagonalizability and \eqref{eq:rotationeigenvalue} is instantly
    obtained using Theorem \ref{thm:hyperbolic1} and Lemma
    \ref{lem:eigenpolynomial}.
\end{proof}


\section{Riemann Problem}
\label{sec:RiemannProblem}
Though the regularized moment system \eqref{eq:Rsystem} is given by
moment expansion up to an arbitrary order $M$ thus extremely complex,
the eigenvalues and eigenvectors of the coefficient matrix
$\tbAM^{(d)}$ are rather organized, which makes it possible to study
the structure of the elementary wave of this system with Riemann
initial value, including the rarefaction wave, contact discontinuity
and shock wave. Definitely, the structure of the elementary wave is
fundamental for further investigation into the behavior of the
solution of the system, and is instructional for studying the
Godunov-type Riemann solver. The investigation below shows that the
structure of the elementary wave of the Riemann problem is quite
natural an extension of that of Euler equations, which indicates that
the regularized moment system \eqref{eq:Rsystem} is actually a very
reasonable high order moment approximation of Boltzmann
equations. Similarly as analyses of Euler equations (see \cite{Toro}),
we consider the $x_1$-split, $D$-dimensional Riemann problem as below:
\begin{equation}\label{eq:RiemannProblem}
    \left\{
        \begin{array}{l}
            \pd{\bw}{t}+\left(u_1\boldsymbol{I}+\tbAM\right)\pd{\bw}{x_1}=0,\\
            \bw(x_1,t=0)=\left\{\begin{array}{ll}
                \bw_L, & \text{if } x_1<0,\\
                \bw_R, & \text{if } x_1>0,
            \end{array}\right.
        \end{array}\right.
\end{equation}
where $\tbAM$ is equal to $\tbAM^{(1)}$, and is defined in Definition
\ref{def:regularization}.

Now let us recall the properties of $\tbAM$. The characteristic
polynomial of $\tbAM$ is $\mathcal{P}_{D,M}(\lambda)$ defined in Lemma
\ref{lem:eigenpolynomial}, thus the eigenvalues of $\tbAM$ are
$\rC{i}{m}\sqrt{\theta_{11}}$ (multiplicity is ignored),
$i=1,\cdots,m$, $m=1,\cdots,M+1$, if $D\geq2$, and are
$\rC{i}{M+1}\sqrt{\theta_{11}}$, $i=1,\cdots,M+1$ if $D=1$. 
For each eigenvalue $\lambda$ of $\tbAdotM$, the corresponding
eigenvector $\bR$ can be obtained by extending the corresponding
diagonal block's eigenvector. And $\boldsymbol{P}    \bR$ \footnote{
    $\boldsymbol{P}$ satisfying $\boldsymbol{P}\bw=\bw'$.
} is
eigenvector of $\tbAM$ for the eigenvalue $\lambda$.  Property
\ref{cor:lambdar} indicates that $\lambda R_0\neq0$ holds, if and only
if $\He_{M+1}^{[\theta_{11}]}(\lambda)=0$ and $\lambda\neq0$.
Therefore, for the matrix $u_1\boldsymbol{I}+\tbAM$,
\begin{enumerate}
    \item 
        the characteristic polynomial is
        $\mathcal{P}_{D,M}(\lambda-u_1)$;
    \item 
        the eigenvalues are $u_1+\rC{i}{m}\sqrt{\theta_{11}}$
        (multiplicity is ignored), $i=1,\cdots,m$, $m=1,\cdots,M+1$,
        if $D\geq2$, and are $u_1+\rC{i}{M+1}\sqrt{\theta_{11}}$,
        $i=1,\cdots,M+1$ if $D=1$;
    \item
        for each eigenvalue $\lambda$, the corresponding eigenvectors
        are same as that of $\tbAM$. Particularly,
        $R_{e_1}=\frac{\lambda-u_1}{\rho}R_0$,
        $R_{2e_1}=\frac{(\lambda-u_1)^2}{2}R_0$;
    \item 
        the eigenvalue and the corresponding eigenvector satisfy the
        relation:
        \begin{equation}\label{eq:lambda_r_condition}
            (\lambda-u_1) R_0\neq0 \text{ holds, if and only if }
            \He_{M+1}^{[\theta_{11}]}(\lambda-u_1)=0 \text{ and }
            \lambda-u_1\neq0.
        \end{equation}
\end{enumerate}

Since the eigenvalues and eigenvectors of coefficient matrix
$u_1\boldsymbol{I}+\tbAM$ are clarified, we can obtain the following
result.
\begin{theorem}\label{thm:waves}
    Each characteristic field of \eqref{eq:RiemannProblem} is
    either genuinely nonlinear or linearly degenerate. And one
    characteristic field is genuinely nonlinear if and only if the
    eigenvalue $\lambda=u_1+\mathrm{C}\sqrt{\theta_{11}}$ satisfies
    $\He_{M+1}^{[\theta_{11}]}(\mathrm{C}\sqrt{\theta_{11}})=0$ and $\mathrm{C}\neq0$.
\end{theorem}
\begin{proof}
    Let $\bR$ be an eigenvector of $\tbAdotM$ for the eigenvalue
    $\mathrm{C}\sqrt{\theta_{11}}$, then $\boldsymbol{P}\bR$ is an
    eigenvector of $u_1\boldsymbol{I}+\tbAM$ for the eigenvalue
    $\lambda=u_1+\mathrm{C}\sqrt{\theta_{11}}$.  Since
    \[
        \lambda=u_1+\mathrm{C}\sqrt{\dfrac{p_{11}}{\rho}},
    \]
    depends only on $\rho$, $u_1$, $p_{11}/2$, we have
    \begin{align*}
        \nabla_{\bw}\lambda \cdot \bR &=
        -\frac{\mathrm{C}\sqrt{\theta_{11}}}{2\rho} \cdot R_0 
        + 1\cdot \frac{\mathrm{C}\sqrt{\theta_{11}}}{\rho}R_0
        +\frac{\mathrm{C}}{\rho\sqrt{\theta_{11}}}\cdot\frac{\mathrm{C}^2\theta_{11}}{2}R_0\\
        &=\frac{(\mathrm{C}^2+1)\sqrt{\theta_{11}}}{2\rho}\mathrm{C}R_0.
    \end{align*}
    \eqref{eq:lambda_r_condition} shows that:
    \begin{enumerate}
        \item If
            $\He_{M+1}^{[\theta_{11}]}(\mathrm{C}\sqrt{\theta_{11}})=0$
            and $\mathrm{C}\neq0$, then
            $\mathrm{C}\sqrt{\theta_{11}}R_0\neq0$, thus
            $\nabla_{\bw}\lambda\cdot\bR\equiv0$. Hence, this
            characteristic field is linearly degenerate.
        \item If
            $\He_{M+1}^{[\theta_{11}]}(\mathrm{C}\sqrt{\theta_{11}})\neq0$
            or $\mathrm{C}=0$, then 
            $\mathrm{C}\sqrt{\theta_{11}}R_0=0$, thus
            $\nabla_{\bw}\lambda\cdot\bR\not\equiv0$. Hence, this
            characteristic field is genuinely nonlinear.
    \end{enumerate}
    This completes the proof.
\end{proof}

The waves associated with $\lambda$ satisfying
$\He_{M+1}^{[\theta_{11}]}(\lambda-u_1)\neq0$ or $\lambda-u_1=0$ are
\textit{contact discontinuities}, and those associated with $\lambda$
satisfying $\He_{M+1}^{[\theta_{11}]}(\lambda-u_1)=0$ and
$\lambda-u_1\neq0$ will either be rarefaction waves or shock waves.
Of course one does not know in advance what types of waves will
be present in the solution of the Riemann problem. Below, we will
study each type of waves separately in detail.

\subsection{Rarefaction Waves}
For the Riemann problem \eqref{eq:RiemannProblem}, 
if two states $\bw^L$ and $\bw^R$ are connected by a rarefaction wave
associated with genuinely nonlinear characteristic field $\bR$, which
is a right eigenvector of $\tbAM$ corresponding to the eigenvalue
$\lambda=u_1+\mathrm{C}\sqrt{\theta_{11}}$ satisfying
$\He_{M+1}^{[\theta_{11}]}(\mathrm{C}\sqrt{\theta_{11}})=0$ and
$\mathrm{C}\neq0$, then the following two conditions must be met:
\begin{enumerate}
    \item
        constancy of the \textit{generalized Riemann invariants}
        across the wave, which implies that the integral curve
        $\tilde{\bw}(\zeta)=(\tilde{w}_1(\zeta),\tilde{w}_2(\zeta),\cdots,\tilde{w}_N(\zeta))$
        in the $N$-dimensional phase space satisfies
        \begin{equation}\label{eq:rarefaction_c1}
            \od{\tilde{\bw}(\zeta)}{\zeta}=\bR(\tilde{\bw});
        \end{equation}
    \item
        divergence of characteristics
        \begin{equation}\label{eq:rarefaction_c2}
            \lambda^L=u_1^L+\mathrm{C}\sqrt{\theta_{11}^L}<
            u_1^R+\mathrm{C}\sqrt{\theta_{11}^R}=\lambda^R.
        \end{equation}
\end{enumerate}
Actually, for a given point $\bw^0$ in the phase space, the integral
curve across $\bw^0$ can be given. Since the results are rather
tedious, we only give partial explicit expressions of the integral
curve as below.  For the characteristic field $\bR$ corresponding to
the eigenvalue $\lambda=u_1+\mathrm{C}\sqrt{\theta_{11}}$,
\begin{enumerate}
    \item if $R_0\neq0$, we choose $R_0=\rho$, then
        \[
            R_{e_1}=\mathrm{C}\sqrt{\theta_{11}},\quad
            R_{2e_1} = \frac{\mathrm{C}^2}{2}p_{11},
        \]
        and then we have
        \begin{subequations}\label{eq:integralcurve1}
            \begin{align}
                \tilde{\rho}(\zeta) &= \rho^0\exp(\zeta),\\
                \tilde{u}_1(\zeta) &=
                u_1^0+\frac{2\mathrm{C}}{\mathrm{C}^2-1}\sqrt{\theta_{11}^0}
                \left( \exp\left(
                \frac{\mathrm{C}^2-1}{2}\zeta\right)-1\right),\\
                \tilde{p}_{11} &= p_{11}^0\exp(\mathrm{C}^2\zeta),
            \end{align}
        \end{subequations}
        where $\theta_{11}^0=p_{11}^0/\rho^0$;
    \item if $R_0=0$, then $R_{e_1}=R_{2e_1}=0$, we have
        \begin{equation}\label{eq:integralcurve2}
            \tilde{\rho}(\zeta)=\rho^0,\quad
            \tilde{u}_1(\zeta)=u_1^0,\quad
            \tilde{p}_{11}(\zeta)=p_{11}^0.
        \end{equation}
\end{enumerate}

One finds that \eqref{eq:integralcurve1} and
\eqref{eq:integralcurve2} satisfy \eqref{eq:rarefaction_c1}. Since for
the rarefaction waves, the eigenvalue
$\lambda=u_1+\mathrm{C}\sqrt{\theta_{11}}$ satisfies
$\He_{M+1}^{[\theta_{11}]}(\mathrm{C}\sqrt{\theta_{11}})=0$ and
$\mathrm{C}\neq0$, thus $R_0\neq0$, the eigenvalue of
$\tbAM(\tilde{\bw}(\zeta))$ is 
\begin{align*}
    \lambda(\tilde{\bw}(\zeta)) &=
    \tilde{u}_1(\zeta)+\mathrm{C}\sqrt{\frac{\tilde{p}_{11}(\zeta)}{\tilde{\rho}(\zeta)}}\\
    &= u_1^0+\frac{2\mathrm{C}}{\mathrm{C}^2-1}\sqrt{\theta_{11}^0}
    \left( \exp\left(
    \frac{\mathrm{C}^2-1}{2}\zeta\right)-1\right)
    +
    \mathrm{C}\sqrt{\frac{p_{11}^0\exp(\mathrm{C}^2\zeta)}{\rho^0\exp(\zeta)}}\\
    &=\lambda(\bw^0)+2\mathrm{C}\frac{\mathrm{C}^2+1}{\mathrm{C}^2-1}\sqrt{\theta_{11}^0}
    \left( \exp\left( \frac{\mathrm{C}^2-1}{2}\zeta \right) -1\right).
\end{align*}
It is clear that 
$\frac{\mathrm{C}^2+1}{\mathrm{C}^2-1}\sqrt{\theta_{11}^0}
\left( \exp\left( \frac{\mathrm{C}^2-1}{2}\zeta \right)
-1\right)$ has the same sign as $\zeta$ for any $\mathrm{C}\in\bbR$, hence, 
$\lambda(\tilde{\bw})\gtrless\lambda(\bw^0)$ if and only if
$\mathrm{C}\zeta\gtrless0$. Therefore, for the rarefaction waves,
noting \eqref{eq:rarefaction_c2}, we have that:
$\lambda=u_1+\mathrm{C}\sqrt{\theta_{11}}$ satisfies
$\He^{[\theta_{11}]}_{M+1}(\mathrm{C}\sqrt{\theta_{11}})=0$ and $\mathrm{C}\neq0$, and
\begin{align*}
    \text{if } \mathrm{C}>0, \text{ then } u_1^L<u_1^R,\quad
    p_{11}^L<p_{11}^R;\\
    \text{if } \mathrm{C}<0, \text{ then } u_1^L<u_1^R,\quad p_{11}^L>p_{11}^R.
\end{align*}

\subsection{Contact discontinuity}
The proof of theorem \ref{thm:waves} indicates that the contact
discontinuity can be founded if and only if the eigenvector $\bR$ and
the corresponding eigenvalue
$\lambda=u_1+\mathrm{C}\sqrt{\theta_{11}}$ satisfying
$\mathrm{C}R_0=0$. For a contact discontinuity,
\eqref{eq:rarefaction_c1} is still valid, and the divergence of
characteristics is replaced by
\begin{equation}\label{eq:contact_c}
    \lambda(\bw^L)=\lambda(\bw^R).
\end{equation}

If $\mathrm{C}\neq0$, then $R_0=0$, \eqref{eq:integralcurve2}
indicates that 
\[
    u_1^L=u_1^R,\quad p_{11}^L=p_{11}^R.
\]
If $\mathrm{C}=0$, then we can derive form both
\eqref{eq:integralcurve1} and \eqref{eq:integralcurve2} that the upper
equation is valid.

Summarizing the discussion above, we conclude that for a contact
discontinuity, 
$\lambda=u_1+\mathrm{C}\sqrt{\theta_{11}}$ satisfies
$\He^{[\theta_{11}]}_{M+1}(\mathrm{C}\sqrt{\theta_{11}})\neq0$ or $\mathrm{C}=0$, and
\begin{equation*}
    u_1^L=u_1^R,\quad p_{11}^L=p_{11}^R.
\end{equation*}

\subsection{Shock waves}
As is well known, the jump condition on the shock wave is sensitive to
the form of the hyperbolic equations. Thus, we rewrite
\eqref{eq:RiemannProblem} in an appropriate form, before the
discussion of the shock wave. However, \eqref{eq:RiemannProblem} can
not be written as conservation laws due to the regularization.
Nevertheless, since the regularization only modifies the governing
equations of $f_{\alpha}$ with $|\alpha|=M$, \eqref{eq:RiemannProblem}
can still preserve the conservation of the else moments with
orders from $0$ to $M-1$. Hence, \eqref{eq:RiemannProblem} can be
reformulated into $\mN( (M-1)e_D)$ conservation laws and $N-\mN(
(M-1)e_D)$ non-conservative equations.

Let 
\begin{equation}
    \boldsymbol{F}=(F_0, F_{e_1},
    \cdots,F_{e_D},F_{2e_1},\cdots,F_{Me_D})^T,\quad
    F_{\alpha}=\frac{1}{\alpha!}\int_{\bbR^D}\bxi^Df\dd\bxi,\quad
    |\alpha|\leq M,
\end{equation}
where $\bxi^\alpha=\prod_{d=1}^D\xi_d^{\alpha_d}$, and $F_0$ stands
for $F_{\alpha}|_{\alpha=0}$.  Then \eqref{eq:RiemannProblem} can be
written as
\begin{equation}\label{eq:conservation_form}
    \begin{split}
        &\pd{F_{\alpha}}{t}+(\alpha_1+1)\pd{F_{\alpha+e_1}}{x_1}=0,
        \quad  |\alpha|<M,\\
        &\begin{aligned}
            \pd{F_{\alpha}}{t}&+(\alpha_1+1)\pd{F_{\alpha+e_1}-f_{\alpha+e_1}}{x_1}\\
        &-(\alpha_1+1)\left( 
        \sum_{i=1}^Df_{\alpha+e_1-e_i}\pd{u_i}{x_1}
        +\sum_{i,j=1}^D\frac{f_{\alpha+e_1-e_i-e_j}}{2\rho}\left( 
        \pd{p_{ij}}{x_1} -\theta_{ij}\pd{\rho}{x_1}\right)
        \right)=0,\\
        \end{aligned}\\
        &\hspace{0.8\textwidth}|\alpha|=M.
    \end{split}
\end{equation}
The integral relation \eqref{eq:integralrelation} and the
quasi-orthogonal relation \eqref{eq:He_orthogonal} of the generalized
Hermite polynomial indicate that there exists a function $g_{\alpha}$
such that
\[
    F_{\alpha}=f_{\alpha}+g_{\alpha}(f_0,u_i,p_{ij},f_{\beta}|_{|\beta|<|\alpha|})
\]
Particularly, for $|\alpha|=M$, we have
$F_{\alpha+e_1}-f_{\alpha+e_1}$ only depends on $\boldsymbol{F}$, and 
\[
    \rho=F_0,\quad
    u_i=F_{e_i}/\rho,\quad
    p_{ij}=(1+\delta_{ij})F_{e_i+e_j}-\frac{F_{e_i}F_{e_j}}{F_0}.
\]
For convenience, the quasi-linear form of
\eqref{eq:conservation_form} can be written as
\begin{equation}\label{eq:quasi-form}
    \pd{\boldsymbol{F}}{t}+\boldsymbol{\Gamma}(\boldsymbol{F})\pd{\boldsymbol{F}}{x_1}=0,
\end{equation}
where $\boldsymbol{\Gamma}(\boldsymbol{F})$ is an $N\times N$ matrix
and depends on \eqref{eq:conservation_form}.

Since \eqref{eq:quasi-form} is not a conservative system, we have to
adopt the DLM theory \cite{Maso} to study the shock wave. For a shock
wave the two constant states $\boldsymbol{F}^L$ and $\boldsymbol{F}^R$
are connected through a single jump discontinuity in a genuinely
non-linear field $\bR$, which is a right eigenvector of $\tbAM$
corresponding to the eigenvalue
$\lambda=u_1+\mathrm{C}\sqrt{\theta_{11}}$ satisfying
$\He_{M+1}^{[\theta_{11}]}(\mathrm{C}\sqrt{\theta_{11}})=0$ and
$\mathrm{C}\neq0$, travelling at the speed $S$. Then the following two
conditions apply
\begin{itemize}
    \item
        Generalized Rankine-Hugoniot condition:
        \begin{equation}\label{eq:generalisedRH}
            \int_0^1\left[ S\boldsymbol{I}-
                \boldsymbol{\Gamma}\left(\boldsymbol{\Phi}(\nu;\boldsymbol{F}^L,\boldsymbol{F}^R)\right)
                \right]\pd{\boldsymbol{\Phi}}{\nu}\left( \nu; \boldsymbol{F}^L,\boldsymbol{F}^R\right)
                \dd\nu=0,
        \end{equation}
        where $\boldsymbol{I}$ is the $N\times N$ identity matrix, and 
        $\boldsymbol{\Phi}(\nu;\boldsymbol{F}^L,\boldsymbol{F}^R)$ is
        a locally Lipschitz mapping satisfying
        \[
            \boldsymbol{\Phi}(0;\boldsymbol{F}^L,\boldsymbol{F}^R)=\boldsymbol{F}^L,\quad
            \boldsymbol{\Phi}(1;\boldsymbol{F}^L,\boldsymbol{F}^R)=\boldsymbol{F}^R.
        \]
        We refer the readers to \cite{Maso} for details.
    \item
        Entropy condition:
        \begin{equation}\label{eq:shock-entropy}
            \lambda(\boldsymbol{F}^L)>S>\lambda(\boldsymbol{F}^R).
        \end{equation}
\end{itemize}
For conservation laws, \eqref{eq:generalisedRH} is the same as the
classical Rankine-Hugoniot condition, thus the first $\mN( (M-1)e_D)$
rows of \eqref{eq:generalisedRH} are independent of
$\boldsymbol{\Phi}$, which make it possible to deduce some properties
of the shock waves before specifying the form of $\boldsymbol{\Phi}$.

Since 
\[
    F_0=\rho,\quad
    F_{e_1}=\rho u_1,\quad
    F_{2e_1}=\frac{1}{2}\left( p_{11}+\rho u_1^2 \right),
\]
the first equation and the $(D+1)$-th equation of
\eqref{eq:generalisedRH} are
\begin{align}
    \rho^Lu_1^L-\rho^Ru_1^R&=S (\rho^L-\rho^R),\label{eq:RH1}\\
    p_{11}^L+\rho^L(u_1^L)^2-p_{11}^R-\rho^R(u_1^R)^2
    &=S(\rho^Lu_1^L-\rho^Ru_1^R).\label{eq:RH2}
\end{align}
We assert that $\rho^L\neq\rho^R$. Otherwise, if $\rho^L=\rho^R$, then
\eqref{eq:RH1} indicates $u_1^L=u_1^R$, and \eqref{eq:RH2} indicates
$p_{11}^L=p_{11}^R$. Thus
$\lambda(\boldsymbol{F}^L)=\lambda(\boldsymbol{F}^R)$ holds, which
contradicts \eqref{eq:shock-entropy}.  One can rewrite \eqref{eq:RH1}
as
\begin{equation}
    S = \frac{\rho^Lu_1^L-\rho^Ru_1^R}{\rho^L-\rho^R}\label{eq:S1}.
\end{equation}
Substituting \eqref{eq:S1} into \eqref{eq:shock-entropy}, and
multiplying both sides with $(\rho^L-\rho^R)^2$, we obtain
\begin{subequations}\label{eq:shock-c}
    \begin{align}
        \rho^L(\rho^L-\rho^R)(u_1^L-u_1^R)>\mathrm{C}(\rho^L-\rho^R)^2
        \sqrt{\theta_{11}^R},\label{eq:shock-c1}\\
        \rho^R(\rho^L-\rho^R)(u_1^L-u_1^R)>\mathrm{C}(\rho^L-\rho^R)^2
        \sqrt{\theta_{11}^L}.\label{eq:shock-c2}
    \end{align}
\end{subequations}

If $\mathrm{C}>0$, \eqref{eq:shock-c1} gives
\begin{equation}\label{eq:shock-ru}
    (\rho^L-\rho^R)(u_1^L-u_1^R)>0.
\end{equation}
Thus, we can divide both sides of \eqref{eq:shock-c} by
$ (\rho^L-\rho^R)(u_1^L-u_1^R)$ to arrive
\[
    \frac{\rho^L}{\sqrt{\theta_{11}^R}}>\frac{\mathrm{C}(\rho^L-\rho^R)}
    {u_1^L-u_1^R}>\frac{\rho^R}{\sqrt{\theta_{11}^L}},
\]
from which one directly gets
\begin{equation}\label{eq:shock-rp}
    \rho^Lp_{11}^L-\rho^Rp_{11}^R>0.
\end{equation}
Similarly, if $\mathrm{C}<0$, we have 
\begin{equation}\label{eq:shock-ru-rp}
    (\rho^L-\rho^R)(u_1^L-u_1^R)<0,\quad
    \rho^Lp_{11}^L-\rho^Rp_{11}^R<0.
\end{equation}

If $S\neq0$, then \eqref{eq:RH1} and \eqref{eq:RH2} can be
reformulated as
\begin{equation}\label{eq:shock-rpu}
    (\rho^L-\rho^R)(p_{11}^L-p_{11}^R)=\rho^L\rho^R(u_1^L-u_1^R)^2>0.
\end{equation}
Here $u_1^L\neq u_1^R$ is used. If $S=0$, then \eqref{eq:RH1} and 
\eqref{eq:RH2} can be reformulated as
\[
    p_{11}^L-p_{11}^R+\frac{\rho^R}{\rho^L}(u_1^R)^2(\rho^R-\rho^L)=0,
\]
thus we have
\begin{equation}\label{eq:shock-rpu2}
    (\rho^L-\rho^R)(p_{11}^L-p_{11}^R)>0.
\end{equation}
Collecting \eqref{eq:shock-rpu} and \eqref{eq:shock-rpu2}, we observe
that one and only one of the following two statements is valid
\begin{enumerate}
    \item $\rho^L>\rho^R$ and $p_{11}^L>p_{11}^R$;
    \item $\rho^L<\rho^R$ and $p_{11}^L<p_{11}^R$.
\end{enumerate}
If $\mathrm{C}>0$, \eqref{eq:shock-rp} indicates that the first
statement is valid. Then we can conclude $u_1^L>u_1^R$ by
\eqref{eq:shock-ru}. Similarly, if $\mathrm{C}<0$, the second
statement is valid and $u_1^L>u_1^R$.

Now we summarize all the discussion on the entropy condition of the
three types of elementary waves in the following theorem:
\begin{theorem}
    For the Riemann problem \eqref{eq:RiemannProblem}, for the wave of
    the family corresponding the eigenvalue
    $\lambda=u_1+\mathrm{C}\sqrt{\theta_{11}}$ of $\tbAM$, the
    macroscopic velocities and pressures on both sides of the wave
    have the relation with the type of the wave as in Table
    \ref{tab:RiemannProblem}.
    \begin{table}[h!]
        \centering
        \begin{tabular}{|l|c||c|}
            \hline
            Wave type   &   Eigenvalue & Velocity and Pressure    \\ [2mm]\hline
            \multirow{2}*{Rarefaction wave} & 
            $\frac{\mathstrut}{\mathstrut}$$\mathrm{C}>0$ & $u_1^L < u_1^R$,
            $p^L_{11}<p^R_{11}$  \\[2mm]\cline{2-3}
            & $\frac{\mathstrut}{\mathstrut}$$\mathrm{C}<0$ & $u_1^L < u_1^R$,
            $p^L_{11}>p^R_{11}$  \\[2mm]\cline{2-3}
            \hline\hline
            \multirow{2}*{Shock wave}   &
            $\frac{\mathstrut}{\mathstrut}$$\mathrm{C}>0$ & $u_1^L>u_1^R$,
            $p^L>p^R$   \\[2mm]\cline{2-3}
            &   $\frac{\mathstrut}{\mathstrut}$$\mathrm{C}<0$ &
            $u_1^L>u_1^R$, $p^L<p^R$   \\  [2mm] \cline{2-3}
            \hline\hline
            {Contact discontinuity}    & ---   & $\frac{\mathstrut}{\mathstrut}$
            $u_1^L=u_1^R$, $p_{11}^L=p_{11}^R$\\
            \hline
        \end{tabular}
        \caption{\label{tab:RiemannProblem}The relation between the
    type classification of elementary wave and the eigenvalue,
macroscopic velocity and pressure.}
    \end{table}
\end{theorem}


\section*{Acknowledgements}
This research was supported in part by the National Natural Science
Foundation of China No. 11325102 and No. 91330205.


\section*{Appendix}
\appendix
\section{Generalized Hermite Polynomials}
\label{sec:HermitePolynomial}
To facilitate the derivation of Grad moment system \cite{Grad}, Grad
gave a note of $D$-dimensional probabilists' isotropic Hermite
polynomials in \cite{Grad1949note}. Maurice M. Mizrahi proposed the
physicists' generalized Hermite polynomials and derived several
properties in \cite{Mizrahi}. In this appendix, we will give a note on
the probabilists' generalized Hermite polynomials to facilitate the
derivation of the content.

\subsection{Definition and Notation}
Consider the normalized $D$-dimensional weight function 
\begin{equation}
    \weight(\bx)=\frac{1}{\sqrt{\det{(2\pi\Theta)}}}\exp\left(-\frac{1}{2}\bx^T\Theta^{-1}\bx\right),
    \quad\text{such that }\int_{\bbR^D}\weight(\bx)\dd\bx=1,
\end{equation}
where $\bx\in\bbR^D$ and $\Theta=(\theta_{ij})\in\bbR^{D\times D}$ is
a symmetrical positive definite matrix.
A probabilists' generalized Hermite polynomials can be defined by:
\begin{equation}\label{eq:definition_HeT}
    \HeT_{\alpha}(\bx)=
    \frac{(-1)^{|\alpha|}}{\weight(\bx)}\pd{^{\alpha}}{\bx^{\alpha}}\weight(\bx),
    \quad \alpha\in\bbN^D,
\end{equation}
where $\alpha=(\alpha_1,\cdots,\alpha_D)$ is a $D$-dimensional
multi-index, $|\alpha|=\sum_{d=1}^D\alpha_d$ and
$\pd{^{\alpha}}{\bx^{\alpha}}$ denotes by
$\dfrac{\partial^{\alpha}}{\partial x_1^{\alpha_1}\cdots \partial
x_D^{\alpha_D}}$. 
And denote the generalized Hermite functions by
\begin{equation}\label{eq:def_mH}
    \mHT_{\alpha}(\bx)=\weight(\bx)\HeT_{\alpha}(\bx),
    \quad \alpha\in\bbN^D.
\end{equation}
If any component of $\alpha$ is negative, $\HeT_{\alpha}$ and
$\mHT_{\alpha}$ are taken as zero for convenience.

Since $\Theta$ is a symmetrical positive definite $D\times D$
matrix, $\Theta^{-1}$ is also a symmetrical positive definite
matrix and denote $\Theta^{-1}=(\theta^{ij})$.  Let $\bX$ denote
$\Theta^{-1}\bx$, i.e.  $X_i=\sum_{j=1}^D\theta^{ij}x_j$.  Then the
corresponding transformation of the operators of partial
differentiation can be written as
\begin{equation}\label{eq:differentiation_relation}
    (\pd{\cdot}{x_1},\cdots,\pd{\cdot}{x_D})^T=
    \Theta^{-1}(\pd{\cdot}{X_1},\cdots,\pd{\cdot}{X_D})^T.
\end{equation}
\subsection{Properties of generalized Hermite polynomials}
\renewcommand{\labelenumi}{\arabic{enumi})}
\begin{enumerate}
    \item
        We first give the following relationships:
        \begin{equation}\label{eq:weight_differentiation}
            \pd{\weight(\bx)}{x_i}=-\weight(\bx)X_i,\quad
            \pd{\weight(\bx)}{X_i}=-\weight(\bx)x_i.
        \end{equation}
        \begin{proof}
            The first relation is obvious, and the second one can be
            directly derived from the first one and
            \eqref{eq:differentiation_relation}.
        \end{proof}
    \item The first few terms of $\HeT_{\alpha}(\bx)$ are,
        for $i,j,k,l=1,\cdots,D$,
        \begin{subequations}
            \begin{align}
            &\HeT_{0}(\bx)=1,\\
            &\HeT_{e_i}(\bx)=X_i,\label{eq:He_1}\\
            &\HeT_{e_i+e_j}(\bx)=X_iX_j-\theta^{ij},\label{eq:He_2}\\
            &\HeT_{e_i+e_j+e_k}(\bx)=X_iX_jX_k-\theta^{ij}X_k-\theta^{ik}X_j-\theta^{jk}X_i,\\
            \begin{split}
            \HeT_{e_i+e_j+e_k+e_l}(\bx)&=X_iX_jX_kX_l
            -\theta^{ij}X_kX_l-\theta^{ik}X_jX_l-\theta^{il}X_jX_k
            -\theta^{jk}X_iX_l\\
            &\quad-\theta^{jl}X_iX_k-\theta^{kl}X_iX_j
            +\theta^{ij}\theta^{kl}+\theta^{ik}\theta^{jl}+\theta^{il}\theta^{jk}.\\
            \end{split}
            \end{align}
        \end{subequations}
        Here $\theta^{ij}=\theta^{ji}$ is used. $e_i$, $i=1,\cdots,D$
        is the $D$-dimensional unit multi-index with its $i$-th entry
        equal to 1.
    \item
        Recurrence relation: for $i=1,\cdots,D$,
        \begin{equation}\label{eq:He_recurrence1}
            \HeT_{\alpha+e_i}(\bx)=X_i\HeT_{\alpha}(\bx)
            -\pd{\HeT_{\alpha}}{x_i}.
        \end{equation}
        \begin{proof}
            Considering the derivation of
            $\weight(\bx)\HeT_{\alpha}(\bx)$ with respect to
            $x_i$, $i=1,\cdots,D$,
            \begin{align*}
                \pd{\weight(\bx)\HeT_{\alpha}(\bx)}{x_i}
                &=\weight(\bx)\pd{\HeT_{\alpha}(\bx)}{x_i}
                -X_i\weight(\bx)\HeT_{\alpha}(\bx)
                & \text{using \eqref{eq:weight_differentiation}}\\
                &=(-1)^{|\alpha|}\pd{^{\alpha+e_i}}{\bx^{\alpha+e_i}}\weight(\bx)
                =-\weight(\bx)\HeT_{\alpha+e_i}(\bx),
                & \text{using \eqref{eq:definition_HeT}}
            \end{align*}
            and comparing the right hand sides of the two rows, we
            obtain \eqref{eq:He_recurrence1}.
        \end{proof}
    \item Differential Equation:
        \begin{equation}\label{eq:He_derivation}
            \pd{\HeT_{\alpha}(\bx)}{x_i}
            =\sum_{j=1}^D\theta^{ij}\alpha_j\HeT_{\alpha-e_j}(\bx).
        \end{equation}
        \begin{proof}
            We use mathematical induction to prove
            \eqref{eq:He_derivation}.

            \textit{Basis:} It is obvious that \eqref{eq:He_derivation}
            holds for $\alpha=0$, since $\HeT_{0}(\bx)=1$.
            
            \textit{Inductive step:} Assume \eqref{eq:He_derivation}
            holds for all $|\alpha|\leq n$, $n\in\bbN$.
            \eqref{eq:He_recurrence1} indicates for all
            $0<|\alpha|\leq n+1$,
            \begin{equation*}
                \HeT_{\alpha}=X_j\HeT_{\alpha-e_j}
                -\sum_{d=1}^D\theta^{jd}(\alpha_d-\delta_{jd})\HeT_{\alpha-e_j-e_d},
            \end{equation*}
            where $j\in\{1,\cdots,D\}$ satisfying $\alpha_j>0$.  For
            any $|\alpha|=n+1$, there exists a $j\in\{1,\cdots,D\}$
            such that $\alpha_j>0$. Then 
            \[
                \pd{\HeT_{\alpha}}{x_i}
                =\pd{X_j\HeT_{\alpha-e_j}}{x_i}-
                \sum_{d=1}^D\theta^{jd}(\alpha_d-\delta_{jd})\pd{\HeT_{\alpha-e_j-e_d}}{x_i}.
            \]
            Since 
            \begin{align*}
                \pd{X_j\HeT_{\alpha-e_j}}{x_i}&= \theta^{ij}\HeT_{\alpha-e_j}
                +X_j\sum_{d=1}^D\theta^{id}(\alpha_d-\delta_{jd})\HeT_{\alpha-e_j-e_d}\\
                &=\theta^{ij}\HeT_{\alpha-e_j}+\sum_{d=1}^D\theta^{id}(\alpha_d-\delta_{jd})
                \left(
                \HeT_{\alpha-e_d}+\sum_{k=1}^D\theta^{jk}(\alpha_k-\delta_{jd}-\delta_{kd})
                \HeT_{\alpha-e_j-e_k-e_d}\right)\\
                &=\sum_{d=1}^D\theta^{id}\alpha_d\HeT_{\alpha-e_d}
                + \sum_{d=1}^D\theta^{jd}(\alpha_d-\delta_{jd})\pd{\HeT_{\alpha-e_j-e_d}}{x_i}.
            \end{align*}
            Hence \eqref{eq:He_derivation} holds for all
            $|\alpha|=n+1$.

            This finishes the proof.
        \end{proof}
    \item
        Recurrence relation again: Combing \eqref{eq:He_recurrence1}
        and \eqref{eq:He_derivation}, we obtain
        for $i=1,\cdots,D$,
        \begin{equation}\label{eq:He_recurrence2}
            \HeT_{\alpha+e_i}(\bx)=X_i\HeT_{\alpha}(\bx)
            -\sum_{j=1}^D\theta^{ij}\alpha_j\HeT_{\alpha-e_j}(\bx).
        \end{equation}
        Furthermore, we have, for $d=1,\cdots,D$,
        \begin{equation}\label{eq:He_recurrence3}
        x_d\HeT_{\alpha}(\bx)=\sum_{j=1}^D\theta_{dj}\HeT_{\alpha+e_j}(\bx)+\alpha_d\HeT_{\alpha-e_d}(\bx).
        \end{equation}
        \begin{proof}
            \eqref{eq:He_recurrence2} is obviously holds.  Since
            $x_d=\sum_{i=1}^D\theta_{id}X_i$ and
            $\sum_{i=1}^D\theta_{di}\theta_{ij}=\delta_{dj}$,
            multiplying \eqref{eq:He_recurrence2} by $\theta_{id}$ and
            summing it by $i$ yield \eqref{eq:He_recurrence3}.
        \end{proof}
    \item Quasi orthogonal relation:
        \begin{equation}\label{eq:He_orthogonal}
            \int_{\bbR^D}\HeT_{\alpha}(\bx)\HeT_{\beta}(\bx)\weight\dd\bx 
            =C_{\alpha, \beta}\delta_{|\alpha|,|\beta|},
        \end{equation}
        where $C_{\alpha,\beta}$ is constant dependent on
        $\alpha,\beta$, and $\Theta$.
        \begin{proof}
            It is obvious \eqref{eq:He_orthogonal} holds for
            $\alpha=\beta=0$ with $C_{0,0}=1$. 
            Without loss of generality, we assume
            $0<|\alpha|\geq|\beta|$ and $\alpha_1>0$.
            Since 
            \begin{equation}\label{eq:He_mH_derivation}
                \pd{\weight\HeT_{\alpha}}{x_i}=(-1)^{|\alpha|}\pd{^{\alpha+e_i}}{\bx^{\alpha+e_i}}\weight
                =-\weight\HeT_{\alpha+e_i}
            \end{equation}
            holds for $i=1,\cdots,D$, using the integration of parts
            on
            $\int_{\bbR^D}\HeT_{\alpha}(\bx)\HeT_{\beta}(\bx)\weight\dd\bx$
            yields
            \begin{align}
            \int_{\bbR^D}\HeT_{\alpha}\HeT_{\beta}\weight\dd\bx
            &=-\int_{\bbR^{D-1}}\dd x_2\cdots\dd x_D\int_{\bbR}
            \HeT_{\beta}\dd\left(\weight\HeT_{\alpha-e_1}\right)\nonumber\\
            &=
            \int_{\bbR^D}\weight\HeT_{\alpha-e_1}
            \sum_{j=1}^D\theta^{1j}\beta_j\HeT_{\beta-e_j}\dd\bx\nonumber\\
            &=\sum_{j=1}^D\theta^{1j}\beta_j
            \int_{\bbR^D}\HeT_{\alpha-e_1}\HeT_{\beta-e_j}\weight\dd\bx,\label{eq:He_orthogonal_recurrence}
            \end{align}
            It's used that 
            $\int_{\bbR^{D-1}}\weight\HeT_{\alpha-e_1}\HeT_{\beta}\dd
            x_2\cdots\dd x_D |_{-\infty}^{\infty}=0$ holding for all
            $\alpha,\beta\in\bbN^D$.

            If $|\alpha|>|\beta|$, repeating
            \eqref{eq:He_orthogonal_recurrence} till some entry of
            the subscript of $\HeT$ negative, we can obtain 
            $\int_{\bbR^D}\HeT_{\alpha}(\bx)\HeT_{\beta}(\bx)\weight\dd\bx=0$.
        \end{proof}
    \item Integral relation: for $\alpha,\beta\in\bbN^D$ and
        $|\alpha|=|\beta|$,
        \begin{equation}\label{eq:integralrelation}
            \int_{\bbR^D}\weight\HeT_{\alpha}(\bx-\boldsymbol{a})\bx^{\beta}\dd\bx=
            \alpha!\delta_{\alpha,\beta},
        \end{equation}
        where $\bx^{\beta}=\prod_{d=1}^Dx_d^{\beta_d}$,
        $\delta_{\alpha,\beta}=\prod_{d=1}^D\delta_{\alpha_d,\beta_d}$,
        and $\boldsymbol{a}$ is a constant vector.
        \begin{proof}
            It is obvious that \eqref{eq:integralrelation} holds for
            $\alpha=0$. Without loss of generality, we assume
            $|\alpha|>0$.
            For convenience, let 
            $C_{\alpha,\beta}=\int_{\bbR^D}\weight\HeT_{\alpha}(\bx-\boldsymbol{a})\bx^{\beta}\dd\bx$.
            Since $|\beta|>0$, there exists an $i\in\{1,\cdots,D\}$
            such that $\beta_i>0$. Using the recurrence relation
            \eqref{eq:He_recurrence3}, we can obtain
            \begin{align*}
                C_{\alpha,\beta}&=\int_{\bbR^D}\weight\HeT_{\alpha}(\bx-\boldsymbol{a})
                \bx^{\beta-e_i}(x_i-a_i+a_i)\dd\bx\\
                &=a_iC_{\alpha,\beta-e_i}+\sum_{d=1}^D\theta_{ij}C_{\alpha+e_j,\beta-e_i}
                +\alpha_iC_{\alpha-e_i,\beta-e_i}.
            \end{align*}
            The quasi-orthogonal \eqref{eq:He_orthogonal} of Hermite
            polynomials indicates that 
            \[
                C_{\alpha,\beta}=0, \quad\text{if } |\beta|<|\alpha|.
            \]
            Thus we have
            \[
                C_{\alpha,\beta}=\alpha_iC_{\alpha-e_i,\beta-e_i}.
            \]
            Since $C_{0,0}=1$, using the mathematical induction on
            $\alpha$, one can prove \eqref{eq:integralrelation} is
            valid.
        \end{proof}
\end{enumerate}
\subsection{Properties of generalized Hermite functions}
\begin{enumerate}
    \item Recursion relation: for $i=1,\cdots,D$,
        \begin{align}
            \mHT_{\alpha+e_i}(\bx)&=X_i\mHT_{\alpha}(\bx)
            -\sum_{j=1}^D\theta^{ij}\alpha_j\mHT_{\alpha-e_j}(\bx),\label{eq:mH_recurrence1}\\
            x_d\mHT_{\alpha}(\bx)&=\sum_{j=1}^D\theta_{jd}\mHT_{\alpha+e_j}(\bx)+\alpha_d\mHT_{\alpha-e_d}(\bx).
            \label{eq:mH_recurrence2}
        \end{align}
        These two equations can be derived directly from
        \eqref{eq:He_recurrence2} and \eqref{eq:He_recurrence3},
        respectively.
    \item Quasi-orthogonality relation:
        \begin{equation}\label{eq:mH_orthogonal}
            \int_{\bbR^D}\mHT_{\alpha}(\bx)\mHT_{\beta}(\bx)\frac{1}{\weight}\dd
            \bx=C_{\alpha, \beta}\delta_{|\alpha|,|\beta|},
        \end{equation}
        where $C_{\alpha,\beta}$ is same as that in
        \eqref{eq:He_orthogonal}. And the equation can be obtained
        from \eqref{eq:He_orthogonal} directly.
    \item Differential relations:
        \begin{align}
            &\pd{\mHT_{\alpha}(\bx)}{x_i}=-\mHT_{\alpha+e_i}(\bx),\\
            & \od{\mH_{\alpha}^{[\Theta(\tau)]}(\bx(\tau))}{\tau}
            =-\sum_{i=1}^D\mH_{\alpha+e_i}^{[\Theta(\tau)]}(\bx(\tau))\od{x_i(\tau)}{\tau}
            +\frac{1}{2}\sum_{i,j=1}^D\mH_{\alpha+e_i+e_j}^{[\Theta(\tau)]}(\bx(\tau))\od{\theta_{ij}(\tau)}{\tau},
        \end{align}
        \begin{proof}
        The first relation is what \eqref{eq:He_mH_derivation} tells.
        For the second one, we first list some useful results in
        matrix calculus as following (see \cite{MatrixCookBook} for
        details). For a symmetrical positive definite matrix
        $\Theta(\tau)=(\theta_{ij}(\tau))\in\bbR^{D\times D}$, and 
        $\Theta^{-1}=(\theta^{ij})=(\theta^{(1)},\cdots,\theta^{(D)})$,
        \begin{subequations}\label{eq:MatrixCalculus}
        \begin{align}
            \od{\bx^T\Theta^{-1}\bx}{\tau}&=2\bx^T\Theta^{-1}\od{\bx}{\tau}
            +\bx^T\od{\Theta^{-1}}{\tau}\bx,\\
            \od{\Theta^{-1}(\tau)}{\tau}&=-\Theta^{-1}\od{\Theta(\tau)}{\tau}\Theta^{-1}
            =-\sum_{i,j=1}^D\theta^{(i)}(\theta^{(j)})^T\od{\theta_{ij}}{\tau},
            \\
            \od{\ln(|\Theta|)}{\tau}&
            ={\mathrm{trace}}\left( \Theta^{-1}\od{\Theta}{\tau} \right)
            =\sum_{i,j=1}^D\theta^{ij}\od{\theta_{ij}}{\tau}.
        \end{align}
        \end{subequations}
        Then we have the relation:
        \begin{align*}
            \od{w^{[\Theta(\tau)]}(\bx(\tau))}{\tau}&=
            \weight\left( -\frac{1}{2}\od{\ln(|\Theta|)}{\tau}
            -\Theta^{-1}\bx\od{\bx}{\tau}
            -\frac{1}{2}\bx^T\od{\Theta^{-1}}{\tau}\bx\right)\\
            &=\weight\sum_{i,j=1}^D\left( 
            -\frac{1}{2}\theta^{ij}\od{\theta_{ij}}{\tau}
            -\theta^{ij}x_j\od{x_i}{\tau}
            +\frac{1}{2}\bx^T\theta^{(i)}(\theta^{(j)})^T\bx\od{\theta_{ij}}{\tau}
            \right)\\
            &=-\weight\sum_{i}^D
            \HeT_{e_i}(\bx)\od{x_i}{\tau}
            +\weight\sum_{i,j=1}^D\frac{1}{2}\HeT_{e_i+e_j}(\bx)
            \od{\theta_{ij}}{\tau}.
        \end{align*}
        Here \eqref{eq:He_1}, \eqref{eq:He_2} and
        \eqref{eq:MatrixCalculus} are used.
        Since the definition of $\mHT$ \eqref{eq:def_mH} indicates
        \[
            (-1)^{|\alpha|}\pd{^{\alpha}}{\bx^{\alpha}}\left(
            \mHT_{\beta} \right)
            =\mHT_{\alpha+\beta},
        \]
        we have
        \begin{align*}
            \od{\mHT}{\tau}&=(-1)^{\alpha}\pd{^\alpha}{\bx^{\alpha}}
            \od{\weight}{\tau}\\
            &=-(-1)^{|\alpha|}\sum_{i=1}^D\od{x_i}{\tau}\pd{^{\alpha}}{\bx^{\alpha}}
            \mHT_{e_i}
            +(-1)^{|\alpha|}\frac{1}{2}\sum_{i,j=1}^D\od{\theta_{ij}}{\tau}
            \pd{^\alpha}{\bx^\alpha}\mHT_{e_i+e_j}\\
            &=-\sum_{i=1}^D\mHT_{\alpha+e_i}\od{x_i}{\tau}
            +\frac{1}{2}\sum_{i,j=1}^D\mH_{\alpha+e_i+e_j}\od{\theta_{ij}}{\tau}.
        \end{align*}
        \end{proof}
    \item Rotation relation: let $\boldsymbol{G}=(g_{ij})_{D\times D}$
        be a rotation matrix, -i.e. $\boldsymbol{G}$ is orthogonal and
        its determinant is 1, then there exists a group of constants
        depending on $\boldsymbol{G}$ such that
        \begin{equation}\label{eq:rotationrelation}
            \mHT_{\alpha}(\bx)=\sum_{|\beta|=|\alpha|}Q_{\alpha}^{\beta}
            \mH^{[\boldsymbol{G}\Theta\boldsymbol{G}^T]}_{\beta}
            (\boldsymbol{G}\bx),
        \end{equation}
        and if we collect $Q_{\alpha}^{\beta}$ as a matrix
        $(Q_{\alpha}^{\beta})|_{|\alpha|=|\beta|=m}$, then it is
        non-singular.
        \begin{proof}
            We define the linear space
            \[
                \mathcal{V}_m=\left\{ p(\bx)\weight\left|\,
                    \int_{\bbR^D}p(\bx)\mHT_{\beta}\dd\bx=0,
                    \forall |\beta|\neq m, ~p(\bx) \text{ is a
                    multivariate polynomial} \right.\right\}.
            \]
            With the quasi-orthogonal relation
            \eqref{eq:mH_orthogonal} of $\mHT_{\alpha}$, it is
            apparent that $\mHT_{\alpha}(\bx)$ with $|\alpha|=m$ is a
            basis of $\mathcal{V}_m$, and the dimension of
            $\mathcal{V}_m$ is $\mN(me_D)-\mN( (m-1)e_D)$.  Hence, we
            just need to prove that
            $\mH^{[\boldsymbol{G}\Theta\boldsymbol{G}^T]}_{\beta}
            (\boldsymbol{G}\bx)$ is also a basis of $\mathcal{V}_m$.
            Since
            \[
                w^{[\boldsymbol{G}\Theta\boldsymbol{G}^T]}(\boldsymbol{G}\bx)
                =\frac{1}{\sqrt{|2\pi\Theta|}}\exp\left( 
                -\frac{1}{2}(\boldsymbol{G}\bx)^T(\boldsymbol{G}\Theta\boldsymbol{G}^T)^{-1}
                (\boldsymbol{G}\bx) \right)
                =\weight(\bx),
            \]
            \newcommand\mHTGalpha{\mH^{[\boldsymbol{G}\Theta\boldsymbol{G}^T]}_{\alpha}(\boldsymbol{G}\bx)}
            $\mHTGalpha$ can also be defined as
            \[
                \mH^{[\boldsymbol{G}\Theta\boldsymbol{G}^T]}_{\alpha}
                (\boldsymbol{G}\bx) = 
                (-1)^{|\alpha|}\pd{^{\alpha}}{(\boldsymbol{G}\bx)}\weight(\bx).
            \]
            This means $\mHTGalpha$ is a rotation of
            $\mHT_{\alpha}(\bx)$, thus
            $\left\{\mHTGalpha\right\}_{|\alpha|=m}$ is linearly
            independent. 

            The quasi orthogonal relation \eqref{eq:mH_orthogonal}
            indicates that for any multi-dimensional polynomial
            $p(\bx)$ with its degree less than $|\alpha|$,
            $\int_{\bbR^D} \mHTGalpha p(\bx)\dd\bx=0$. Hence,
            $\mHTGalpha$ is orthogonal with $\mHT_{\beta}(\bx)$,
            $|\beta|<|\alpha|$.  Analogously, $\mHT_{\beta}(\bx)$,
            $|\beta|>|\alpha|$ is orthogonal with $\mHTGalpha$.
            So we have $\mHTGalpha\in\mathcal{V}_{|\alpha|}$.

            In conclusion, $\mHTGalpha$ with $|\alpha|=m$ is a basis
            of $\mathcal{V}_m$, thus prove the property.
        \end{proof}
\end{enumerate}
\subsection{Properties of one-dimensional case}
In this subsection, we study the generalized Hermite polynomials at
the case $D=1$, then the matrix $\Theta$ degenerates into a scalar
$\theta$, and $\Theta^{-1}$ turns to $1/\theta$. Particularly, if
$\theta=1$, the generalized Hermite polynomials $\Het_n(x)$ is exactly
the ordinary Hermite polynomials $\He_n(x)$. The ordinary Hermite
polynomials can be defined as
\begin{equation}
    \He_n(x)=(-1)^n\exp(x^2/2)\od{ }{x}\exp(-x^2/2),
\end{equation}
and has the properties (see \cite{Abramowitz} for details)
\begin{enumerate}
    \item Parity: $\He_n(-x) = (-1)^n \He_n(x)$;
    \item Recursion relation:
        $\He_{n+1}(x)=x\He_n(x)-n\He_{n-1}(x)$, $n\in\bbN^+$;
    \item Orthogonality relation:
        $\displaystyle\int_{\bbR}\He_n(x)\He_m(x)\exp(-x^2/2)\dd
        x=\sqrt{2\pi}m!\delta_{mn}$;
    \item Differential relation:
        $\He_n(x)'=n\He_{n-1}(x)$.
\end{enumerate}
Since the generalized Hermite polynomials in one-dimensional case
is defined as 
\begin{equation}
    \Het_n(x)=(-1)^n\exp(x^2/2\theta)\od{ }{x}\exp(-x^2/2\theta),
\end{equation}
hence, we can obtain
\begin{equation}
    \HeT(x)=\theta^{-n/2}\He_n(x/\sqrt{\theta}).
\end{equation}
Therefore, $\Het_n(x)$ satisfies the following properties:
\begin{enumerate}
    \item Parity: $\Het_n(-x) = (-1)^n \Het_n(x)$;
    \item Recurrence relation:
        $\Het_{n+1}(x)=\frac{x}{\theta}\Het_n(x)-\frac{n}{\theta}\Het_{n-1}(x)$, $n\in\bbN^+$;
    \item Orthogonal relation:
        $\displaystyle\int_{\bbR}\Het_n(x)\Het_m(x)\exp(-x^2/2)\dd
        x=\sqrt{2\pi}\frac{m!}{\theta^m}\delta_{mn}$;
    \item Differential relation:
        $\He_n(x)'=\frac{n}{\theta}\He_{n-1}(x)$.
\end{enumerate}

Next we discuss the zeros of $\He_n(x)$. The following properties can
be found in many handbooks such as \cite{Chihara}.
\begin{property}
    \begin{enumerate}
        \item $0$ is a zero of $\He_n(x)$ if $n$ is an odd number;
        \item
            There are $n$ different real zeros of $\He_n(x)$;
        \item There is a zero of $\He_{n+1}(x)$ between any two zeros
            of $\He_n(x)$;
        \item There is no same zeros of $\He_n(x)$ and
            $\He_{n+1}(x)$.
    \end{enumerate}
\end{property}
Furthermore, we conjecture that there is no same non-zero zeros of
$\He_n(x)$ and $\He_m(x)$ for all $m,n\in\bbN$ and $m\neq n$. However, 
to our knowledge, no proof for it has been given. We propose it as a
conjecture.
\begin{conjecture}
  For any $m,n\in\bbN$ and $m\neq n$, there is no common non-zero
  zeros of $\He_n(x)$ and $\He_m(x)$, -i.e.  $\nexists x\in\bbR
  \backslash \{ 0 \}$, such that $\He_n(x)=\He_m(x)=0$.
\end{conjecture}
We have verified this conjecture by computer algebra system for $m,n
\leq 1000$.

\section{Proof of Lemma \ref{lem:bAsecond}}
\label{sec:proofLemmabAsecond}
\begin{proof}[Proof of Lemma \ref{lem:bAsecond}]
    Choose $\bR\in\bbR^M$, such that
    \[
        R_1=1,\quad R_2=\rho\lambda,\quad
        R_k=\rho\He_{k-1}^{[\theta_{11}]}(\lambda)/(k-1)!-f_{(k-1)e_1}-\lambda
        f_{(k-2)e_1},\quad k=2,\dots,M,
    \]
    where $\lambda$ is an eigenvalue of $\thA_1$, then we verify that
    $\thA_1 \bR=\lambda \bR$, which is equivalent to 
    \begin{equation}\label{eq:lemmabC1}
        \thA_1(i,\cdot)\bR=\lambda R_i, \quad \text{for }i=1,\dots,M.
    \end{equation}
    It is easy to check \eqref{eq:lemmabC1} holding for $i = 1, 2, 3$.
    For $i=4,\dots,M-1$, since the regularization changes only the
    entries of the last row of $\thA_1$, thus \eqref{eq:hat1} gives us
    the entries of $\thA_1$ as
    \begin{align*}
        &\thA_1(i,1:3)=(i f_{ie_1}, \left(
        (i-1)f_{(i-1)e_1}+\theta_{11}f_{(i-3)e_1} \right)/\rho,
        -2f_{(i-2)e_1}/\rho),\\
        &\thA_1(i,i-1:i+1)=(\theta_{11},0,i).
    \end{align*}
    Note that any entries of $\thA_1(i,\cdot)$, if is not given above,
    is zero. And some entries, which is double defined above, is the
    sum of the both expressions. Thus
    \begin{align*}
        \thA_1(i,\cdot)\bR&=i f_{ie_1}\cdot 1+
        \left((i-1)f_{(i-1)e_1}+\theta_{11}f_{(i-3)e_1}
        \right)/\rho\cdot\rho\lambda \\
        &\quad-2f_{(i-2)e_1}/\rho\cdot\rho\He_2^{[\theta_{11}]}(\lambda)/2!
        +\theta_{11}\cdot\left( 
        \rho\He_{i-2}^{[\theta_{11}]}(\lambda)/(i-2)!-f_{(i-2)e_1}-\lambda
        f_{(i-3)e_1} \right)\\
        &\quad+i\cdot
        \left(\rho\He_{i}^{[\theta_{11}]}(\lambda)/i!-f_{ie_1}-\lambda
        f_{(i-1)e_1}\right)\\
        &=\rho\theta_{11}\frac{\He_{i-2}^{[\theta_{11}]}(\lambda)}{(i-2)!}
        +\rho\frac{i \He_i^{[\theta_{11}]}(\lambda)}{i!}
        -\lambda f_{(i-1)e_1}-\lambda^2f_{(i-2)e_1}\\
        &=\lambda\left( 
        \rho\He_{i-1}^{[\theta_{11}]}(\lambda)/(i-1)!-f_{(i-1)e_1}-\lambda
        f_{(i-2)e_1} \right)\\
        &=\lambda R_i.
    \end{align*}
    Note that
    $\He_{n+1}^{[\theta_{11}]}(\lambda)+n\He_{n-1}^{[\theta_{11}]}(\lambda)=
    \lambda\He_n^{[\theta_{11}]}(\lambda)$ is used in the calculation above.
    For the case $i=M$, the regularization \eqref{eq:regularization}
    and \eqref{eq:hat1} gives us that
    \begin{align*}
        &\thA_1(M, 1:3)=(0,
         \left( -f_{(i-1)e_1}+\theta_{11}f_{(i-3)e_1} \right)/\rho,
        -2f_{(i-2)e_1}/\rho),\\
        &\thA_1(M,M-1:M)=(\theta_{11},0).
    \end{align*}
    Similarly, any entries of $\thA_1(i,\cdot)$, if is not given
    above, is taken as zero. And some entries, which is given twice
    above (when $M\leq4$), is the sum of the both expressions. Thus
    \begin{align*}
        \thA_1(i,\cdot)\bR&=
        \left(-f_{(i-1)e_1}+\theta_{11}f_{(i-3)e_1}
        \right)/\rho\cdot\rho\lambda \\
        &\quad-2f_{(i-2)e_1}/\rho\cdot\rho\He_2^{[\theta_{11}]}(\lambda)/2!
        +\theta_{11}\cdot\left( 
        \rho\He_{i-2}^{[\theta_{11}]}(\lambda)/(i-2)!-f_{(i-2)e_1}-\lambda
        f_{(i-3)e_1} \right)\\
        &=\rho\theta_{11}\frac{\He_{i-2}^{[\theta_{11}]}(\lambda)}{(i-2)!}
        -\lambda f_{(i-1)e_1}-\lambda^2f_{(i-2)e_1}\\
        &=\lambda\left( 
        \rho\He_{i-1}^{[\theta_{11}]}(\lambda)/(i-1)!-f_{(i-1)e_1}-\lambda
        f_{(i-2)e_1} \right)
        -\rho\He_{M}^{[\theta_{11}]}(\lambda)/M!  \\
        &=\lambda R_i
        -\rho\He_{M}^{[\theta_{11}]}(\lambda)/M!.
    \end{align*}
    Hence, if $\lambda$ satisfies $\He_M^{[\theta_{11}]}(\lambda)=0$,
    then $(\lambda, R)$ is a pair of eigenvalue/eigenvector of
    $\thA_1$.

    It is clear that any root of $\He_M^{[\theta_{11}]}(\lambda)$ is
    an eigenvalue of $\thA_1$. Since $\He_M^{[\theta_{11}]}(\lambda)$
    is a monic polynomial, the characteristic polynomial of $\thA_1$
    is $\He_M^{[\theta_{11}]}(\lambda)$. This proves the lemma.
\end{proof}

\section{Proof of Lemma \ref{lem:properprolongation}}
\label{sec:proofLemmaprolongation}
Before we begin the proof of Lemma \ref{lem:properprolongation}, 
we list some results on linear algebra without proof.
\begin{lemma}\label{lem:simpleeigenvalue}
    For a $k\times k$ block lower triangular matrix $\bA\in\bbR^N$
    with the size of diagonal block $n_i\times n_i$,
    $n_1+\cdots+n_k=N$, $\br_i$ is an eigenvector of the $i$-th
    diagonal block for the eigenvalue $\lambda$. 
    If $\lambda$ is a simple eigenvalue of $\bA$, then there exists a
    proper prolongation of $\br_i$.
\end{lemma}
\begin{lemma}\label{lem:multieigenvalue}
    $\bA$ is defined the same as that in Lemma \ref{lem:simpleeigenvalue},
    and denote $A_{ij}$ the $i$-th row, $j$-th column block of $\bA$.
    Each diagonal block of $\bA$ is diagonalizable with real
    eigenvalues.  $\br_i$ is an eigenvector of $A_{ii}$ for the
    eigenvalue $\lambda$. $\lambda$ is an eigenvalue of $A_{jj}$,
    $j\neq i$, and is not an eigenvalue of any other diagonal block of
    $\bA$. If there exists a proper prolongation to the matrix
    \[
        \begin{pmatrix}
            A_{ii}&0\\
            A_{ij}&A_{jj}
        \end{pmatrix},
        \text{ if } i < j, \text{ or }
        \begin{pmatrix}
            A_{jj}&0\\
            A_{ji}&A_{ii}
        \end{pmatrix},
        \text{ if } i > j,
    \]
    then there exists a proper prolongation of $\br_i$ to the matrix
    $\bA$.
\end{lemma}

\begin{proof}[Proof of the Lemma \ref{lem:properprolongation}]
    The case $D=1$ has been proved in \cite{Fan}, here we just
    consider the case $D\geq 2$.
    Define $\bR_{\alpha}=(R_{\alpha,\beta})^T$, where the order of
    $R_{\alpha,\beta}$ in $\bR_{\alpha}$ is the lexicographic order of
    $(\alpha_2,\cdots,\alpha_D,\alpha_1)$, same as that in $\bw'$.
    Similarly, define $\br_{\alpha} =
    (r_{\alpha,1},\cdots,r_{\alpha,M+1-|\hat{\alpha}|})^T$.
    If $\bR_{\alpha}$ is an prolongation of $\br_{\alpha}$, then 
    \begin{equation}\label{eq:bRprolongation}
        R_{\alpha, \beta} = r_{\alpha, \beta_1+1} \text{ with }
        \hat{\beta}=\hat{\alpha}.
    \end{equation}
    If $\bR_{\alpha}$ is an eigenvector of $\tbAdotM$ for the
    eigenvalue $\lambda$, then \eqref{eq:modified_ms} indicates
    $R_{\alpha,\beta}$ satisfying: for $|\beta|\leq M$,
    \begin{subequations}\label{eq:condition_R}
        \begin{align}
            &\rho R_{\alpha,e_1}=\lambda R_{\alpha,0},
            \label{eq:condition_R_mass}\\
            &\frac{1}{\rho}(1+\delta_{1i})R_{\alpha,e_1+e_i}=\lambda
            R_{\alpha,e_i},
            \label{eq:condition_R_momentum}\\
            &p_{ij}R_{\alpha,e_1}
            +p_{1i}R_{\alpha,e_j}+p_{1j}R_{\alpha,e_i}
            +(e_i+e_j+e_1)!R_{\alpha,e_i+e_j+e_1}
            =\lambda(1+\delta_{ij})R_{\alpha,e_i+e_j},\label{eq:condition_R_pressure}\\
            \begin{split}\label{eq:condition_R_f}
                & \sum_{k = 1}^D \theta_{1k}
                R_{\alpha,\beta- e_k} +
                (1-\delta_{|\beta|,M})(\beta_1+1)R_{\alpha,\beta+e_1}\\
                &+\sum_{i,j=1}^D\frac{\tilde{C}_{ij}(\beta)}{2\rho}\left(
                (1+\delta_{ij})R_{\alpha,e_i+e_j}-\theta_{ij}R_{\alpha,0} \right)
                +  \sum_{i = 1}^D
                (1-\delta_{|\beta|,M})(\beta_1+1)f_{\beta-e_i+e_1}R_{\alpha,e_i}\\ 
                &-  \sum_{i = 1}^D \frac{f_{\beta-e_i}}{\rho}
                (1+\delta_{1i})R_{\alpha,e_1+e_i} 
                -  \sum_{i,j=1}^D \frac{(e_i+e_j+e_1)!}{2}
                \frac{f_{\beta-e_i-e_j}}{\rho} R_{\alpha,e_i+e_j+e_1}
                = \lambda R_{\alpha,\beta},~|\beta|\geq3,
            \end{split}
        \end{align}
    \end{subequations}
    where $\tilde{C}_{ij}(\beta)$ is defined in
    \eqref{eq:def_tildeCijd}. We need to verify $\bR_{\alpha,\beta}$
    satisfying \eqref{eq:bRprolongation} and \eqref{eq:condition_R}
    only, which are checked case by case below.
    \begin{enumerate}
        \item $\lambda\neq0$; $D=2$, or $D\geq 3$, and $\lambda$
            satisfying $\He_{M+1}^{[\theta_{11}]}(\lambda)=0$.  If
            $D=2$, the characteristic polynomial of $\tbAdotM$ is
            $\prod_{m=1}^{M+1}\He_{m}^{[\theta_{11}]}(\lambda)$.
            Conjecture \ref{con:He} indicates that each nonzero
            eigenvalue of $\tbAdotM$ is a simple eigenvalue.  If
            $D\geq3$, and $\lambda$ satisfying $\lambda\neq0$ and
            $\He_{M+1}^{[\theta_{11}]}(\lambda)=0$, $\lambda$ is a
            simple eigenvalue of $\tbAdotM$. By Lemma
            \ref{lem:simpleeigenvalue}, there exists a
            proper prolongation of each eigenvector of $\tbAdotM$ 
            associated to these eigenvalues.
        \item $\lambda\neq0$, $D\geq3$ and $\lambda$ satisfying
            $\He_{M+1}^{[\theta_{11}]}(\lambda)\neq0$. 
            This case is corresponding to $|\hat{\alpha}|\geq1$. Let
            $\lambda$ be the $\alpha_1$-th eigenvalue
            $\thA_{\hat{\alpha}}$, then the corresponding eigenvector
            is $\br_{\alpha}$.

            Conjecture \ref{con:He} indicates that $\lambda$ is an
            eigenvalue of $\thA_{\hat{\beta}}$,
            $|\hat{\beta}|=|\hat{\alpha}|$, and is not for any
            $\thA_{\hat{\beta}}$, $|\hat{\beta}|\neq|\hat{\alpha}|$.
            Here we first prolongate the eigenvector $\br_{\alpha}$ to
            the diagonal block of $\tbAdotM$, containing all
            $\thA_{\hat{\beta}}$, $|\hat{\beta}|=|\hat{\alpha}|$ (for
            convenience, denote the diagonal block by
            $\boldsymbol{B}$), then use Lemma
            \ref{lem:multieigenvalue} to obtain a proper prolongation
            of $\br_{\alpha}$.

            Actually, observing \eqref{eq:modified_ms}, we find that
            the equation, including the term $\odd{w_{\alpha}}{t}$, does not
            depend on $w_{\beta}$, $|\beta|=|\alpha|$, which implies that
            $\boldsymbol{B}$ is a block diagonal matrix, and
            particularly, each diagonal block is
            $\thA_{\hat{\alpha}}$. Hence, let $R_{\alpha,
            \beta}=r_{\alpha, \beta_1+1}$ with
            $\hat{\beta}=\hat{\alpha}$, and $R_{\alpha, \beta}=0$ with
            $\hat{\beta}\neq\hat{\alpha}$ and
            $|\hat{\beta}|=|\hat{\alpha}|$, then
            $(R_{\alpha,\beta})|_{|\hat{\beta}|=|\hat{\alpha}|}$ is a
            prolongation of $\br_{\alpha}$ to the matrix
            $\boldsymbol{B}$.  Obviously,
            $(R_{\alpha,\beta})|_{|\hat{\beta}|=|\hat{\alpha}|}$ is a
            proper prolongation of $\br_{\alpha}$ to the matrix
            $\boldsymbol{B}$. With Lemma \ref{lem:multieigenvalue},
            the conclusion is validated.
        \item $\lambda=0$. Since Hermite polynomial $\He_n(x)$ is odd
            function if $n$ is odd, $\lambda=0$ is multi-eigenvalue of
            $\tbAdotM$. We have to check it in cases:
        \begin{enumerate}
            \item Case $|\hat{\alpha}|=0$. If $\lambda=0$ is an
                eigenvalue of $\thA_{0}$, then $M$ is even. Let
                \begin{align}\label{eq:proofCase1}
                    R_{\alpha,0}&=\rho,\quad R_{\alpha, \beta}=0,
                    \quad |\beta|=1,\nonumber\\
                    R_{\alpha,\beta}
                    &=-\frac{1}{\beta_1}\sum_{k=1}^D\theta_{1k}
                    \left( R_{\alpha,\beta-e_1-e_k}-G(\beta-e_1-e_k)
                    \right) + G(\beta),
                     \quad |\beta|>1,
                \end{align}
                where $G(\beta)=\sum_{i,j=1}^D\theta_{ij}f_{\beta-e_i-e_j}$,
                and $(\cdot)_{\beta}$ is taken as zero if any entry of
                $\beta$ is negative. Particularly, $R_{\alpha, \beta}=0$,
                if $|\beta|=1,2,3$, and 
                \begin{equation}\label{eq:case1odd}
                    R_{\alpha,\beta} = G(\beta), \text{ if } |\beta|
                    \text{ is odd.}
                \end{equation}

                Let $\beta=me_1$, $m=0,\cdots,M$, it can be derived that 
                \begin{align*}
                    &R_{\alpha,0}=\rho,~R_{\alpha,
                    e_1}=R_{\alpha,2e_1}=0,\\
                    &R_{\alpha,
                    me_1}-f_{(m-2)e_1}\theta_{11}=-\frac{1}{m}\theta_{11}
                    (R_{\alpha, (m-2)e_1}-f_{(m-4)e_1}\theta_{11}).
                \end{align*}
                Using the recurrence relation of
                $\He^{[\theta_{11}]}(\lambda)$ with $\lambda=0$, we
                find that $R_{\alpha, me_1}=r_{\alpha, m+1}$, where
                $r_{\alpha, m+1}$ is same as that defined in
                \eqref{eq:eigenvector1D} with $\lambda=0$.

                Then we verify that $R_{\alpha, \beta}$ satisfies
                \eqref{eq:condition_R}.  Notice $R_{\alpha,\beta}=0$,
                $|\beta|=1,2,3$, thus \eqref{eq:condition_R_mass},
                \eqref{eq:condition_R_momentum} and
                \eqref{eq:condition_R_pressure} holds, and
                \eqref{eq:condition_R_f} degenerates into, for
                $3\leq|\beta|\leq M$,
                \[
                    \sum_{k = 1}^D \theta_{1k}
                    R_{\alpha,\beta- e_k} +
                    (1-\delta_{|\beta|,M})(\beta_1+1)R_{\alpha,\beta+e_1}
                    -\sum_{i,j=1}^D\frac{\tilde{C}_{ij}(\beta)}{2\rho}\theta_{ij}R_{\alpha,0}=0.
                \]
                For $|\beta|<M$, since $R_{\alpha, \beta-e_k}=0$ holds
                for $|\beta|=3$, $k=1,\cdots,D$,  the equation above
                is exactly what \eqref{eq:proofCase1} tells. For $|\beta|=M$,
                since $M$ is even, $|\beta-e_k|$, $k=1,\cdots,D$, is
                odd, and this equation can be simply derived using
                \eqref{eq:case1odd}.
            \item Case $|\hat{\alpha}|=1$. 
                If $\lambda=0$ is an eigenvalue of $\thA_1$, then $M$
                is odd. Let $\hat{e}_d=\hat{\alpha}$ and
                \begin{align}\label{eq:proofCase2}
                    R_{\alpha,0}&=0, \quad R_{\alpha, e_d}=1,
                    \quad R_{\alpha, \beta}=0, \quad
                    |\beta|=1\text{ and } \beta\neq e_d,
                    \nonumber\\
                    R_{\alpha,\beta}
                    &=-\frac{1}{\beta_1}\sum_{k=1}^D\theta_{1k}
                    \left( R_{\alpha,\beta-e_1-e_k}-G(\beta-e_1-e_k)
                    \right) + G(\beta),
                     \quad |\beta|>1,
                \end{align}
                where $G(\beta)=f_{\beta-e_d}$, and $(\cdot)_{\beta}$
                is taken as zero if any entry of $\beta$ is negative.
                Particularly, $R_{\alpha, me_1}=0$, $m=0,\dots,M$, 
                and 
                \begin{equation}\label{eq:case2even}
                    R_{\alpha, \beta} = G(\beta),\text{ if }
                    |\beta| \text{ is even}.
                \end{equation}

                Let $\beta=e_d+me_1$, $m=0,\cdots,M-1$, it is
                derived that 
                \begin{align*}
                    R_{\alpha,e_d}=1,~R_{\alpha,
                    e_d+e_1}&=1, ~R_{\alpha,e_d+2e_1}=0,\\
                    R_{\alpha, e_d+me_1}-G(e_d+me_1) &=
                    -\frac{1}{m}\theta_{11}
                    (R_{\alpha, e_d+(m-2)e_1}-G(e_d+(m-2)e_1)).
                \end{align*}
                Using the recurrence relation of
                $\He^{[\theta_{11}]}(\lambda)$ with $\lambda=0$, one
                finds that $R_{\alpha, e_d+me_1}=r_{\alpha, m+1}$,
                where $r_{\alpha, m+1}$ is the same as that defined in
                Lemma \ref{lem:bAsecond} with $\lambda=0$.

                Then we verify that $R_{\alpha, \beta}$ satisfies
                \eqref{eq:condition_R}. It is clear that
                \eqref{eq:condition_R_mass},
                \eqref{eq:condition_R_momentum} and
                \eqref{eq:condition_R_pressure} holds. Meanwhile
                \eqref{eq:condition_R_f} degenerates into, for
                $3\leq|\beta|\leq M$,
                \begin{align*}
                    &\sum_{k = 1}^D \theta_{1k}
                    R_{\alpha,\beta- e_k} +
                    (1-\delta_{|\beta|,M})(\beta_1+1)R_{\alpha,\beta+e_1}\\
                    &+(1-\delta_{|\beta|,M})(\beta_1+1)f_{\beta+e_1-e_d}
                    -\sum_{i,j=1}^D\frac{(e_i+e_j+e_1)!}{2}\frac{f_{\beta-e_i-e_j}}{2}
                    R_{\alpha,e_1+e_i+e_j}=0.
                \end{align*}
                To verify this relation is rather tedious but not
                complex, and we have to examine several cases for
                $R_{e_1+e_i+e_j}$. Here we give the idea briefly. First,
                using \eqref{eq:proofCase2} to eliminate
                $R_{\alpha,e_1+e_i+e_j}$. For $|\beta|<M$, one get
                this equation is what \eqref{eq:proofCase2}
                tells. For $|\beta|=M$, since $M$ is odd ,
                $|\beta-e_k|$, $k=1,\cdots,D$, is even, and then this 
                equation is simply derived using \eqref{eq:case2even}.

                Furthermore, the construction of $\bR_{\alpha}$ shows
                the prolongation is proper. 
            \item Case $|\hat{\alpha}|=2$.
                If $\lambda=0$ is an eigenvalue of $\thA_2$, then $M$
                is even. Let $\hat{\gamma}=\hat{\alpha}$,
                $\gamma_1=0$, and
                \begin{align}\label{eq:proofCase3}
                    R_{\alpha,\beta}&=0, \quad |\beta|\leq2, \beta\neq
                    \gamma, \quad R_{\alpha, \gamma}=1, \nonumber\\
                    R_{\alpha,\beta}
                    &=-\frac{1}{\beta_1}\sum_{k=1}^D\theta_{1k}
                    \left( R_{\alpha,\beta-e_1-e_k}-G(\beta-e_1-e_k)
                    \right) + G(\beta),
                     \quad |\beta|>|\gamma|,
                \end{align}
                where $G(\beta)=\dfrac{f_{\beta-\gamma}}{\rho}$, and
                $(\cdot)_{\beta}$ is taken as zero if any entry of
                $\beta$ is negative. Particularly, $R_{\alpha,
                \beta}=0$, $|\hat{\beta}|<2$, and 
                \begin{equation}\label{eq:case3odd}
                    R_{\alpha, \beta} = G(\beta),\text{ if }
                    |\beta| \text{ is odd},
                \end{equation}
                $R_{\alpha,\beta}=0$, $|\beta|=3$.

                Let $\beta=\gamma+me_1$, $m=0,\cdots,M-2$, it can be
                derived that 
                \begin{align*}
                    R_{\alpha,\gamma}=1,~R_{\alpha,
                    \gamma+e_1}&=0, \\
                    R_{\alpha, \gamma+me_1}-G(\gamma+me_1) &=
                    -\frac{1}{m}\theta_{11}
                    (R_{\alpha, \gamma+(m-2)e_1}-G(\gamma+(m-2)e_1)).
                \end{align*}
                Using the recurrence relation of
                $\He^{[\theta_{11}]}(\lambda)$ with $\lambda=0$, we
                can check that $R_{\alpha, \gamma+me_1}=r_{\alpha,
                m+1}$, where $r_{\alpha, m+1}$ is the same as that defined
                in Lemma \ref{lem:bAthird} with $\lambda=0$.

                Then we verify that $R_{\alpha, \beta}$ satisfies
                \eqref{eq:condition_R}. Both the left hand sides and
                right hand sides of \eqref{eq:condition_R_mass},
                \eqref{eq:condition_R_momentum} and
                \eqref{eq:condition_R_pressure} are zero, so these
                equations hold.
                \eqref{eq:condition_R_f} degenerates into, for
                $3\leq|\beta|\leq M$,
                \begin{align*}
                    &\sum_{k = 1}^D \theta_{1k}
                    R_{\alpha,\beta- e_k} +
                    (1-\delta_{|\beta|,M})(\beta_1+1)R_{\alpha,\beta+e_1}\\
                    &\qquad +\frac{1}{\rho} \left(
                    \sum_{k=1}^D\theta_{1k}f_{\beta-\gamma-e_k}+
                    (1-\delta_{|\beta|,M})(1+\beta_1)f_{\beta+e_1-\gamma}
                    \right)=0.
                \end{align*}
                For $|\beta|<M$, the above equation is given by
                \eqref{eq:proofCase3}.  For $|\beta|=M$, since
                $M$ is even, $|\beta-e_k|$, $k=1,\cdots,D$, is odd, and then 
                this equation can be simply derived using
                \eqref{eq:case3odd}.

                Furthermore, the construction of $\bR_{\alpha}$ shows
                the prolongation is proper. 
            \item Case $n=|\hat{\alpha}|\geq3$.
                If $\lambda=0$ is an eigenvalue of
                $\thA_{\hat{\alpha}}$, then $M+1-n$ is odd. Let
                $\hat{\gamma}=\hat{\alpha}$,
                $\gamma_1=0$, and
                \begin{align}\label{eq:proofCase4}
                    R_{\alpha,\beta}&=0, \quad |\beta|\leq n, \beta\neq
                    \gamma, \quad R_{\alpha, \gamma}=1, \nonumber\\
                    R_{\alpha,\beta}
                    &=-\frac{1}{\beta_1}\sum_{k=1}^D\theta_{1k}
                     R_{\alpha,\beta-e_1-e_k}, 
                     \quad |\beta|>|\gamma|,
                \end{align}
                where $(\cdot)_{\beta}$ is taken as zero if any entry
                of $\beta$ is negative. Particularly, $R_{\alpha,
                \beta}=0$, $|\hat{\beta}|<n$, and 
                \begin{equation}\label{eq:case4others}
                    R_{\alpha, \beta} = 0,\text{ if }
                    M+1-|\beta| \text{ is even}.
                \end{equation}

                Let $\beta=\gamma+me_1$, $m=0,\cdots,M+1-n$, it can be
                derived that 
                \begin{align*}
                    R_{\alpha,\gamma}=1,~R_{\alpha,
                    \gamma+e_1}=0, \quad
                    R_{\alpha, \gamma+me_1} = -\frac{1}{m}\theta_{11}
                    R_{\alpha, \gamma+(m-2)e_1}.
                \end{align*}
                Using the recurrence relation of
                $\He^{[\theta_{11}]}(\lambda)$ with $\lambda=0$, we
                can check that $R_{\alpha, \gamma+me_1}=r_{\alpha,
                m+1}$, where $r_{\alpha, m+1}$ is same as that defined
                in Lemma \ref{lem:bAelse} with $\lambda=0$.

                Next we verify that $R_{\alpha, \beta}$ satisfies
                \eqref{eq:condition_R}. Both the left hand sides and
                right hand sides of \eqref{eq:condition_R_mass},
                \eqref{eq:condition_R_momentum} and
                \eqref{eq:condition_R_pressure} are zero, so these
                equations hold.
                \eqref{eq:condition_R_f} degenerates into, for
                $3\leq|\beta|\leq M$,
                \begin{align*}
                    &\sum_{k = 1}^D \theta_{1k}
                    R_{\alpha,\beta- e_k} +
                    (1-\delta_{|\beta|,M})(\beta_1+1)R_{\alpha,\beta+e_1}=0.
                \end{align*} For $|\beta|<n$, both the left hand side
                and the right hand side are zero, so the
                equation holds. For $n\leq |\beta|<M$, this
                equation is exactly what \eqref{eq:proofCase3} tells. For
                $|\beta|=M$, since $M+1-n$ is even, $M+1-|\beta-e_k|$,
                $k=1,\cdots,D$ is odd, then the above equation can be
                simply derived using \eqref{eq:case4others}.

                Furthermore, the construction of $\bR_{\alpha}$ shows
                the prolongation is proper. 
        \end{enumerate}
        All the cases discussion above tells us that each eigenvector
        of a diagonal block for the eigenvalue $\lambda=0$ can be
        prolongate to an eigenvector of $\tbAdotM$, and the
        prolongation is proper.
    \end{enumerate}
    Collecting all the case above, we conclude that each eigenvector
    of each diagonal block of $\tbAdotM$ can be prolongate to an
    eigenvector of $\tbAdotM$, which proves the Lemma.
\end{proof}


\bibliographystyle{plain}
\bibliography{article}
\end{document}